\documentclass[final]{IEEEtran}

\usepackage{amsmath, amssymb, amsbsy, bbm, array, enumerate, calc, scalerel}

\newcommand{\mi}{\mathit} 
\DeclareMathAlphabet{\mib}{T1}{cmr}{bx}{it}
\newcommand{\mc}{\mathcal}

\newcommand{\mf}{\mathfrak}

\newcommand{\mbb}{\mathbb} 
\newcommand{\mbbm}{\mathbbm}
\newcommand{\mbi}{\boldsymbol}

\newcommand{\cref}[1]{\arabic{#1}}
\newcommand{\ctag}[2]{\addtocounter{equation}{1}\setcounter{#1}{\value{equation}}
                      \tag*{#2 (\cref{#1})}}

\newcommand{\geql}{\eqcirc}
\newcommand{\gle}{\prec}   

\newcommand{\feql}{\operatorname{\pmb{\eqcirc}}}
\newcommand{\fle}{\operatorname{\pmb{\prec}}}

\newcommand{\gz}{\bar{0}}
\newcommand{\gu}{\bar{1}}

\newcommand{\fcom}[1]{\pmb{\overline{\phantom{#1}}}\mspace{-9.75mu}#1}
\newcommand{\fcomd}[1]{\pmb{\overline{\phantom{#1}}}\mspace{-47.50mu}#1}
\newcommand{\fcoml}{\pmb{\overline{\parbox{7pt}{\mbox{}\vspace{6.6pt}\mbox{}}}}}

\newcommand{\fvee}{\operatorname{\pmb{\vee}}}
\newcommand{\bigfvee}{\operatornamewithlimits{\pmb{\bigvee}}}

\newcommand{\fwedge}{\operatorname{\pmb{\wedge}}}
\newcommand{\bigfwedge}{\operatornamewithlimits{\pmb{\bigwedge}}}

\newcommand{\bigflimit}{\operatornamewithlimits{\pmb{\mbox{\large X}}}}

\newcommand{\frightarrow}{\operatorname{\pmb{\Rightarrow}}}

\newcommand{\fdiamond}{\operatorname{\pmb{\diamond}}}

\newcommand{\del}{\operatorname{\Delta}}
\newcommand{\fdel}{\operatorname{\pmb{\Delta}}}

\newcommand{\eqvl}{\equiv}

\newcommand{\emptyseq}{\ell}

\newlength{\llll}

\newtheorem{theorem}{Theorem}
\newtheorem{lemma}[theorem]{Lemma}
\newtheorem{corollary}[theorem]{Corollary}

\newenvironment{proof}{\begin{IEEEproof}}
                      {\end{IEEEproof}}  

\newcommand{\qed}{\mbox{}}

\begin{document}

\title{On Multi-step Fuzzy Inference in G\"{o}del Logic}

\author{Du\v{s}an~Guller%
\thanks{D. Guller is with the Department of Applied Informatics, Comenius University, Mlynsk\'a dolina, 842 48 Bratislava, Slovakia
        (e-mail: guller@fmph.uniba.sk).}%
%\thanks{Manuscript received December 23, 2017; revised ?? ??, 20??; accepted ?? ??, 20??.}%
%\thanks{Partially supported by VEGA Grant 1/0592/14.}
}

\markboth{Du\v{s}an Guller \hfill On Multi-step Fuzzy Inference in G\"{o}del Logic \hspace{5mm}}{On Multi-step Fuzzy Inference in G\"{o}del Logic}

\maketitle

\begin{abstract}
This paper addresses the logical and computational foundations of multi-step fuzzy inference using the Mamdani-Assilian type of fuzzy rules
by implementing such inference in G\"{o}del logic with truth constants.
We apply the results achieved in the development of a hyperresolution calculus for this logic. 
We pose three fundamental problems: reachability, stability, the existence of a $k$-cycle in multi-step fuzzy inference and
reduce them to certain deduction and unsatisfiability problems.
The corresponding unsatisfiability problems may be solved using hyperresolution. 
\end{abstract}


\section{Introduction}
\label{S1}

Recent years witness a prospective application potential of the fuzzy logic and inference approach to
emerging technologies in the area of artificial, computational intelligence, and soft computing such as
fuzzy description logics and ontologies, the Semantic Web, fuzzy knowledge/expert systems, 
ambient, swarm intelligence, the Internet of thinks, autonomous driving, etc.  
Fuzzy rules and one-step fuzzy inference is widely exploited in fuzzy controllers since the seventies.
From a viewpoint of artificial intelligence, such a kind of inference has a reactive behavior.
In contrast to control purposes, one-step fuzzy inference is not sufficient for so-called fuzzy reasoning, where
some kind of abstract inference is needed to reach a reasonable conclusion, which stipulates multiple inference steps. 
Descriptions of real-world problems may become rather complex and efficient inference methods and techniques are necessary.
A good choice is to apply and/or generalise those for automated deduction, particularly in many-valued logics.

This paper is a continuation of \cite{Guller2018a}, mainly aiming at theoretical results 
concerning the logical and computational foundations of multi-step fuzzy inference.                       
In \cite{Guller2012b,Guller2015a,Guller2014,Guller2016b,Guller2015b,Guller2015c,Guller2016c,Guller2018b,Guller2019b}, 
we have generalised the well-known hyperresolution principle to the first-order G\"{o}del logic for the general case.
Our approach is based on translation of a formula of G\"{o}del logic to an equivalent satisfiable finite order clausal theory.
We have introduced a notion of quantified atom:
a formula $a$ is a quantified atom if $a=Q x\, p(t_0,\dots,t_n)$
where $Q$ is a quantifier ($\forall$, $\exists$); $p(t_0,\dots,t_n)$ is an atom; $x$ is a variable occurring in $p(t_0,\dots,t_n)$;
for all $i\leq n$, either $t_i=x$, or $x$ does not occur in $t_i$ ($t_i$ is a free term in the quantified atom).
The notion of quantified atom is very important.
It permits us to extend classical unification to quantified atoms without any additional computational cost.
An order clause is then a finite set of order literals of the form
$\varepsilon_1\diamond \varepsilon_2$ where $\varepsilon_i$ is an atom or a truth constant or a quantified atom, and 
$\diamond$ is a connective $\geql$, equality, or $\gle$, strict order.
$\geql$ and $\gle$ are interpreted by the equality and standard strict linear order on $[0,1]$, respectively.
On the basis of the hyperresolution principle, a calculus operating over order clausal theories has been devised.
The calculus is proved to be refutation sound and complete for the countable case 
with respect to the standard $\mbi{G}$-algebra $\mbi{G}=([0,1],{\leq,}\fvee,\fwedge,\frightarrow,\fcoml,\feql,\fle,\fdel,0,1)$ 
augmented by the binary operators $\feql$, $\fle$ for the connectives $\geql$, $\gle$, respectively, and 
by the unary operator $\fdel$ for the projection connective $\del$.
As another step, one may incorporate a countable set of intermediate truth constants of the form $\bar{c}$, $c\in (0,1)$, 
to get a modification of the hyperresolution calculus suitable for automated deduction with explicit partial truth.
We have investigated the canonical standard completeness, where the semantics of G\"{o}del logic is given
by the standard $\mbi{G}$-algebra $\mbi{G}$, and truth constants are interpreted by 'themselves'.
We say that a set $X\supseteq \{\gz,\gu\}$ of truth constants is admissible with respect to suprema and infima 
if, for all $\emptyset\neq Y_1, Y_2\subseteq X$ such that $\bigfvee Y_1=\bigfwedge Y_2$, 
we have $\bigfvee Y_1\in Y_1$ and $\bigfwedge Y_2\in Y_2$.
Then the modified hyperresolution calculus is refutation sound and complete for a countable order clausal theory 
if the set of truth constants occurring in the theory is admissible with respect to suprema and infima.
The achieved results can be applied to the Mamdani-Assilian type of fuzzy rules and inference \cite{MAAS75,Mam76}.
We shall implement this type of fuzzy rules and inference to G\"{o}del logic with truth constants and the connectives $\geql$, $\gle$, $\del$. 
Fuzzy rules will be translated to formulae of G\"{o}del logic, and subsequently, to clausal form.
We shall pose three fundamental problems: reachability, stability, and the existence of a $k$-cycle in multi-step fuzzy inference.
We shall formulate these problems as formulae of G\"{o}del logic and reduce their solving to solving certain deduction problems.
The deduction problems may further be reduced to unsatisfiability ones, and under some finitary restrictions, solved 
using a suitably modified hyperresolution calculus.

The paper is organised as follows.
Section \ref{S2} gives the basic notions and notation concerning the first-order G\"{o}del logic.
Section \ref{S3} deals with translation to clausal form.
Section \ref{S4} is devoted to multi-step fuzzy inference and contains a non-trivial illustration example.
Section~\ref{S5} brings conclusions.

\subsection{Preliminaries}
\label{S1.1}

$\mbb{N}$, $\mbb{Q}$, $\mbb{R}$ designates the set of natural, rational, real numbers, and
$=$, $\leq$, $<$ denotes the standard equality, order, strict order on $\mbb{N}$, $\mbb{Q}$, $\mbb{R}$.
We denote $\mbb{R}_0^+=\{c \,|\, 0\leq c\in \mbb{R}\}$, $\mbb{R}^+=\{c \,|\, 0<c\in \mbb{R}\}$,
$[0,1]=\{c \,|\, c\in \mbb{R}, 0\leq c\leq 1\}$; $[0,1]$ is called the unit interval.
Let $X$, $Y$, $Z$ be sets and $f : X\longrightarrow Y$ a mapping.
By $\|X\|$ we denote the set-theoretic cardinality of $X$.
The relationship of $X$ being a finite subset of $Y$ is denoted as $X\subseteq_{\mc F} Y$.
Let $Z\subseteq X$.
We designate 
$f[Z]=\{f(z) \,|\, z\in Z\}$; $f[Z]$ is called the image of $Z$ under $f$; 
$f|_Z=\{(z,f(z)) \,|\, z\in Z\}$; $f|_Z$ is called the restriction of $f$ onto $Z$.
Let $\gamma\leq \omega$.
A sequence $\delta$ of $X$ is a bijection $\delta : \gamma\longrightarrow X$.
Recall that $X$ is countable if and only if there exists a sequence of $X$.
Let $I$ be an index set, and $S_i\neq \emptyset$, $i\in I$, be sets.
A selector ${\mc S}$ over $\{S_i \,|\, i\in I\}$ is a mapping ${\mc S} : I\longrightarrow \bigcup \{S_i \,|\, i\in I\}$ such that
for all $i\in I$, ${\mc S}(i)\in S_i$.
We denote ${\mc S}\mi{el}(\{S_i \,|\, i\in I\})=\{{\mc S} \,|\, {\mc S}\ \text{\it is a selector over}\ \{S_i \,|\, i\in I\}\}$.
Let $c\in \mbb{R}^+$.
$\log c$ denotes the binary logarithm of $c$.
Let $f, g : \mbb{N}\longrightarrow \mbb{R}_0^+$.
$f$ is of the order of $g$, in symbols $f\in O(g)$, iff there exist $n_0$ and $c^*\in \mbb{R}_0^+$ such that
for all $n\geq n_0$, $f(n)\leq c^*\cdot g(n)$.

\section{First-order G\"{o}del logic}
\label{S2}

Throughout the paper, we shall use the common notions and notation of first-order logic.
By ${\mc L}$ we denote a first-order language. 
$\mi{Var}_{\mc L}$, $\mi{Func}_{\mc L}$, $\mi{Pred}_{\mc L}$, $\mi{Term}_{\mc L}$, $\mi{GTerm}_{\mc L}$, 
$\mi{Atom}_{\mc L}$ denotes
the set of all variables, function symbols, predicate symbols, terms, ground terms, atoms of ${\mc L}$.
$\mi{ar}_{\mc L} : \mi{Func}_{\mc L}\cup \mi{Pred}_{\mc L}\longrightarrow \mbb{N}$ denotes 
the mapping assigning an arity to every function and predicate symbol of ${\mc L}$.
$\mi{cn}\in \mi{Func}_{\mc L}$ such that $\mi{ar}_{\mc L}(\mi{cn})=0$ is called a constant symbol.
Let $\{0,1\}\subseteq C_{\mc L}\subseteq [0,1]$ be countable.
We assume a countable set of truth constants of ${\mc L}$
$\overline{C}_{\mc L}=\{\bar{c} \,|\, c\in C_{\mc L}\}$;
$\gz$, $\gu$ denotes the false, the true in ${\mc L}$; $\bar{c}$, $0<c<1$, is called an intermediate truth constant.
Let $x\in \overline{C}_{\mc L}$ and $X\subseteq \overline{C}_{\mc L}$.
Then there exists a unique $c\in C_{\mc L}$ such that $\bar{c}=x$.
We denote $\underline{x}=c$ and
$\underline{X}=\{\underline{x} \,|\, \underline{x}\in C_{\mc L}, x\in X\}$.
By $\mi{Form}_{\mc L}$ we designate the set of all formulae of ${\mc L}$ built up 
from $\mi{Atom}_{\mc L}$, $\overline{C}_{\mc L}$, $\mi{Var}_{\mc L}$
using the connectives: $\neg$, negation, $\del$, Delta, $\wedge$, conjunction, $\vee$, disjunction, $\rightarrow$, implication,  
$\leftrightarrow$, equivalence, $\geql$, equality, $\gle$, strict order, and
the quantifiers: $\forall$, the universal one, $\exists$, the existential
one.\footnote{We assume a decreasing connective and quantifier precedence:
              $\forall$, $\exists$, $\neg$, $\del$, $\geql$, $\gle$, $\wedge$, $\vee$, $\rightarrow$, $\leftrightarrow$.}
In the paper, we shall assume that ${\mc L}$ is a countable first-order language; 
hence, all the above mentioned sets of symbols and expressions are 
countable.\footnote{If the first-order language in question is not explicitly designated,
                    we shall write denotations without index.} 
Let $\varepsilon$, $\varepsilon_i$, $1\leq i\leq m$, $\upsilon_i$, $1\leq i\leq n$, be
either an expression or a set of expressions or a set of sets of expressions of ${\mc L}$, in general.
By $\mi{vars}(\varepsilon_1,\dots,\varepsilon_m)\subseteq \mi{Var}_{\mc L}$,
$\mi{freevars}(\varepsilon_1,\dots,\varepsilon_m)\subseteq \mi{Var}_{\mc L}$,
$\mi{boundvars}(\varepsilon_1,\dots,\varepsilon_m)\subseteq \mi{Var}_{\mc L}$,
$\mi{funcs}(\varepsilon_1,\dots,\varepsilon_m)\subseteq \mi{Func}_{\mc L}$,
$\mi{preds}(\varepsilon_1,\dots,\varepsilon_m)\subseteq \mi{Pred}_{\mc L}$,
$\mi{atoms}(\varepsilon_1,\dots,\varepsilon_m)\subseteq \mi{Atom}_{\mc L}$,
$\mi{tcons}(\varepsilon_1,\dots,\varepsilon_m)\subseteq \overline{C}_{\mc L}$
we denote the set of all variables, free variables, bound variables, function symbols, predicate symbols, atoms, truth constants of ${\mc L}$
occurring in $\varepsilon_1,\dots,\varepsilon_m$.
$\varepsilon$ is closed iff $\mi{freevars}(\varepsilon)=\emptyset$.
By $\emptyseq$ we denote the empty sequence.
Let $\varepsilon_1,\dots,\varepsilon_m$ and $\upsilon_1,\dots,\upsilon_n$ be sequences.
The length of $\varepsilon_1,\dots,\varepsilon_m$ is defined as $|\varepsilon_1,\dots,\varepsilon_m|=m$.
We define the concatenation of $\varepsilon_1,\dots,\varepsilon_m$ and $\upsilon_1,\dots,\upsilon_n$
as $(\varepsilon_1,\dots,\varepsilon_m),(\upsilon_1,\dots,\upsilon_n)=\varepsilon_1,\dots,\varepsilon_m,\upsilon_1,\dots,\upsilon_n$.
Note that the concatenation is 
associative.\footnote{Several simultaneous applications of the concatenation will be written without parentheses.}

Let $t\in \mi{Term}_{\mc L}$, $\phi\in \mi{Form}_{\mc L}$, $T\subseteq_{\mc F} \mi{Form}_{\mc L}$.
We define the size of $t$ by recursion on the structure of $t$ as follows:
\begin{equation*}
|t|=\left\{\begin{array}{ll}
           1                    &\ \text{\it if}\ t\in \mi{Var}_{\mc L}, \\[1mm]
           1+\sum_{i=1}^n |t_i| &\ \text{\it if}\ t=f(t_1,\dots,t_n).
           \end{array}
    \right. 
\end{equation*}
Subsequently, we define the size of $\phi$ by recursion on the structure of $\phi$ as follows:
\begin{equation*}
|\phi|=\left\{\begin{array}{ll}
              1+\sum_{i=1}^n |t_i| &\ \text{\it if}\ \phi=p(t_1,\dots,t_n)\in \mi{Atom}_{\mc L}, \\[1mm]
              1                    &\ \text{\it if}\ \phi\in \overline{C}_{\mc L}, \\[1mm]
              1+|\phi_1|           &\ \text{\it if}\ \phi=\diamond \phi_1, \\[1mm]
              1+|\phi_1|+|\phi_2|  &\ \text{\it if}\ \phi=\phi_1\diamond \phi_2, \\[1mm]
              2+|\phi_1|           &\ \text{\it if}\ \phi=Q x\, \phi_1.
              \end{array}
       \right.
\end{equation*}
Note that $|t|, |\phi|\geq 1$.
The size of $T$ is defined as $|T|=\sum_{\phi\in T} |\phi|$.
By $\mi{varseq}(\phi)$, $\mi{vars}(\mi{varseq}(\phi))\subseteq \mi{Var}_{\mc L}$, 
we denote the sequence of all variables of ${\mc L}$ occurring in $\phi$ which is built up via the left-right preorder traversal of $\phi$.
For example, $\mi{varseq}(\exists w\, (\forall x\, p(x,x,z)\vee \exists y\, q(x,y,z)))=w,x,x,x,z,y,x,y,z$ and $|w,x,x,x,z,y,x,y,z|=9$. 
A sequence of variables will often be denoted as $\bar{x}$, $\bar{y}$, $\bar{z}$, etc.
Let $Q\in \{\forall,\exists\}$ and $\bar{x}=x_1,\dots,x_n$ be a sequence of variables of ${\mc L}$.
By $Q \bar{x}\, \phi$ we denote $Q x_1\dots Q x_n\, \phi$.
Let $x_1,\dots,x_n\in \mi{freevars}(\phi)$.
$\phi$ may be denoted as $\phi(x_1,\dots,x_n)$.
Let $t_1,\dots,t_n\in \mi{Term}_{\mc L}$ be closed.
By $\phi(t_1,\dots,t_n)\in \mi{Form}_{\mc L}$ we denote the instance of $\phi(x_1,\dots,x_n)$ by the substitution $x_1/t_1,\dots,x_n/t_n$,
defined in the standard manner.

G\"{o}del logic is interpreted by the standard $\mbi{G}$-algebra 
augmented by the operators $\feql$, $\fle$, $\fdel$ for the connectives $\geql$, $\gle$, $\del$, respectively.
\begin{equation*}
\mbi{G}=([0,1],\leq,\fvee,\fwedge,\frightarrow,\fcoml,\feql,\fle,\fdel,0,1)
\end{equation*}
where $\fvee$, $\fwedge$ denotes the supremum, infimum operator on $[0,1]$;
\begin{alignat*}{2}
a\frightarrow b &= \left\{\begin{array}{ll}
                          1 &\ \text{\it if}\ a\leq b, \\[1mm]
                          b &\ \text{\it else};
                          \end{array}
                   \right. 
& 
\fcom{a}        &= \left\{\begin{array}{ll}
                          1 &\ \text{\it if}\ a=0, \\[1mm]
                          0 &\ \text{\it else};
                          \end{array}
                   \right. 
\\[2mm]
a\feql b        &= \left\{\begin{array}{ll}
                          1 &\ \text{\it if}\ a=b, \\[1mm]
                          0 &\ \text{\it else};
                          \end{array}
                   \right. 
& \qquad
a\fle b         &= \left\{\begin{array}{ll}
                          1 &\ \text{\it if}\ a<b, \\[1mm]
                          0 &\ \text{\it else};
                          \end{array}
                   \right. 
\\[2mm]
\fdel a         &= \left\{\begin{array}{ll}
                          1 &\ \text{\it if}\ a=1, \\[1mm]
                          0 &\ \text{\it else}.
                          \end{array}
                   \right.
\end{alignat*}
Recall that $\mbi{G}$ is a complete linearly ordered lattice algebra;
$\fvee$, $\fwedge$ is commutative, associative, idempotent, monotone; 
$0$, $1$ is its neutral 
element; %\footnote{Using the commutativity, associativity, idempotence, monotonicity, neutral element of $\fvee$, $\fwedge$
         %          will not explicitly be referred to.}
%
%
the residuum operator $\frightarrow$ of $\fwedge$ satisfies the condition of residuation:
\begin{equation}
\label{eq0a}
\text{for all}\ a, b, c\in \mbi{G},\ a\fwedge b\leq c\Longleftrightarrow a\leq b\frightarrow c;
\end{equation}
G\"{o}del negation $\fcoml$ satisfies the condition:
\begin{equation}
\label{eq0b}
\text{for all}\ a\in \mbi{G},\ \fcom{a}=a\frightarrow 0;
\end{equation}
$\fdel$ satisfies the 
condition:\footnote{We assume a decreasing operator precedence: $\fcoml$, $\fdel$, $\feql$, $\fle$, $\fwedge$, $\fvee$, $\frightarrow$.}
\begin{equation}
\label{eq0kk}
\text{for all}\ a\in \mbi{G},\ \fdel a=a\feql 1.
\end{equation}
Note that the following properties hold:
\begin{alignat}{1}
\notag
& \hspace{-2.24mm} \text{for all}\ a, b, c\in \mbi{G}, \\[1mm]
\ctag{ceq0d}{(distributivity of $\fvee$ over $\fwedge$)}
& a\fvee b\fwedge c=(a\fvee b)\fwedge (a\fvee c), \\[1mm]
\ctag{ceq0c}{(distributivity of $\fwedge$ over $\fvee$)}
& a\fwedge (b\fvee c)=a\fwedge b\fvee a\fwedge c, \\[1mm]
\label{eq0f}
& a\frightarrow b\fvee c=(a\frightarrow b)\fvee (a\frightarrow c), \\[1mm]
\label{eq0e}
& a\frightarrow b\fwedge c=(a\frightarrow b)\fwedge (a\frightarrow c), \\[1mm]
\label{eq0h}
& a\fvee b\frightarrow c=(a\frightarrow c)\fwedge (b\frightarrow c), \\[1mm]
\label{eq0g}
& a\fwedge b\frightarrow c=(a\frightarrow c)\fvee (b\frightarrow c), \\[1mm]
\label{eq0i}
& a\frightarrow (b\frightarrow c)=a\fwedge b\frightarrow c, \\[1mm]
\label{eq0j}
& ((a\frightarrow b)\frightarrow b)\frightarrow b=a\frightarrow b, \\[1mm]
\label{eq0k}
& (a\frightarrow b)\frightarrow c=((a\frightarrow b)\frightarrow b)\fwedge (b\frightarrow c)\fvee c, \\[1mm]
\label{eq0jj}
& (a\frightarrow b)\frightarrow 0=((a\frightarrow 0)\frightarrow 0)\fwedge (b\frightarrow 0).
\end{alignat}

An interpretation ${\mc I}$ for ${\mc L}$ is a triple
$\big({\mc U}_{\mc I},\{f^{\mc I} \,|\, f\in \mi{Func}_{\mc L}\},\{p^{\mc I} \,|\, p\in \mi{Pred}_{\mc L}\}\big)$ defined as follows: 
${\mc U}_{\mc I}\neq \emptyset$ is the universum of ${\mc I}$;
every $f\in \mi{Func}_{\mc L}$ is interpreted as a function $f^{\mc I} : {\mc U}_{\mc I}^{\mi{ar}_{\mc L}(f)}\longrightarrow {\mc U}_{\mc I}$;
every $p\in \mi{Pred}_{\mc L}$ is interpreted as a $[0,1]$-relation $p^{\mc I} : {\mc U}_{\mc I}^{\mi{ar}_{\mc L}(p)}\longrightarrow [0,1]$.
A variable assignment in ${\mc I}$ is a mapping $e : \mi{Var}_{\mc L}\longrightarrow {\mc U}_{\mc I}$. 
We denote the set of all variable assignments in ${\mc I}$ as ${\mc S}_{\mc I}$.
Let $u\in {\mc U}_{\mc I}$.
A variant $e[x/u]\in {\mc S}_{\mc I}$ of $e$ with respect to $x$ and $u$ is defined by
\begin{equation*}
e[x/u](z)=\left\{\begin{array}{ll}
                 u    &\ \text{\it if}\ z=x, \\[1mm]
                 e(z) &\ \text{\it else}.
                 \end{array}
          \right.
\end{equation*}  
In ${\mc I}$ with respect to $e$, 
we define the value $\|t\|_e^{\mc I}\in {\mc U}_{\mc I}$ of $t$ by recursion on the structure of $t$,
the value $\|\bar{x}\|_e^{\mc I}\in {\mc U}_{\mc I}^{|\bar{x}|}$ of $\bar{x}$,
the truth value $\|\phi\|_e^{\mc I}\in [0,1]$ of $\phi$ by recursion on the structure of $\phi$, as follows:
{\footnotesize
\begin{alignat*}{2}
&    t\in \mi{Var}_{\mc L}, 
& &\ \|t\|_e^{\mc I}=e(t); \\[1mm]
&    t=f(t_1,\dots,t_n), 
& &\ \|t\|_e^{\mc I}=f^{\mc I}(\|t_1\|_e^{\mc I},\dots,\|t_n\|_e^{\mc I}); \\[2mm]
&    \bar{x}=x_1,\dots,x_{|\bar{x}|}, 
& &\ \|\bar{x}\|_e^{\mc I}=e(x_1),\dots,e(x_{|\bar{x}|}); \\[2mm]
&    \phi=p(t_1,\dots,t_n), 
& &\ \|\phi\|_e^{\mc I}=p^{\mc I}(\|t_1\|_e^{\mc I},\dots,\|t_n\|_e^{\mc I}); \\[1mm]
&    \phi=c\in \overline{C}_{\mc L},     
& &\ \|\phi\|_e^{\mc I}=\underline{c}; \\[1mm]
&    \phi=\neg \phi_1,
& &\ \|\phi\|_e^{\mc I}=\fcomd{\|\phi_1\|_e^{\mc I}}; \\[1mm]
&    \phi=\del \phi_1,
& &\ \|\phi\|_e^{\mc I}=\fdel \|\phi_1\|_e^{\mc I}; \\[1mm]
&    \phi=\phi_1\diamond \phi_2,
& &\ \|\phi\|_e^{\mc I}=\|\phi_1\|_e^{\mc I}\fdiamond \|\phi_2\|_e^{\mc I}, \quad \diamond\in \{\wedge,\vee,\rightarrow,\geql,\gle\}; \\[1mm]
&    \phi=\phi_1\leftrightarrow \phi_2,
& &\ \|\phi\|_e^{\mc I}=(\|\phi_1\|_e^{\mc I}\frightarrow \|\phi_2\|_e^{\mc I})\fwedge
                        (\|\phi_2\|_e^{\mc I}\frightarrow \|\phi_1\|_e^{\mc I}); \\[1mm]
&    \phi=\forall x\, \phi_1,
& &\ \|\phi\|_e^{\mc I}=\bigfwedge_{u\in {\mc U}_{\mc I}} \|\phi_1\|_{e[x/u]}^{\mc I}; \\[1mm]
&    \phi=\exists x\, \phi_1,
& &\ \|\phi\|_e^{\mc I}=\bigfvee_{u\in {\mc U}_{\mc I}} \|\phi_1\|_{e[x/u]}^{\mc I}.
\end{alignat*}}%
Let $\phi$ be closed.
Then, for all $e, e'\in {\mc S}_{\mc I}$, $\|\phi\|_e^{\mc I}=\|\phi\|_{e'}^{\mc I}$.
Note that ${\mc S}_{\mc I}\neq \emptyset$.
We denote $\|\phi\|^{\mc I}=\|\phi\|_e^{\mc I}$.

Let ${\mc L}'$ be a first-order language, and ${\mc I}$, ${\mc I}'$ be interpretations for ${\mc L}$, ${\mc L}'$, respectively.
${\mc L}'$ is an expansion of ${\mc L}$ iff $\mi{Func}_{{\mc L}'}\supseteq \mi{Func}_{\mc L}$ and 
$\mi{Pred}_{{\mc L}'}\supseteq \mi{Pred}_{\mc L}$;
on the other side, we say that ${\mc L}$ is a reduct of ${\mc L}'$.
${\mc I}'$ is an expansion of ${\mc I}$ to ${\mc L}'$
iff ${\mc L}'$ is an expansion of ${\mc L}$, ${\mc U}_{{\mc I}'}={\mc U}_{\mc I}$,
for all $f\in \mi{Func}_{\mc L}$, $f^{{\mc I}'}=f^{\mc I}$,
for all $p\in \mi{Pred}_{\mc L}$, $p^{{\mc I}'}=p^{\mc I}$;
on the other side, we say that ${\mc I}$ is a reduct of ${\mc I}'$ to ${\mc L}$, in symbols ${\mc I}={\mc I}'|_{\mc L}$.

A theory of ${\mc L}$ is a set of formulae of ${\mc L}$.
$\phi$ is true in ${\mc I}$ with respect to $e$, written as ${\mc I}\models_e \phi$, iff $\|\phi\|_e^{\mc I}=1$.
${\mc I}$ is a model of $\phi$, in symbols ${\mc I}\models \phi$, iff, for all $e\in {\mc S}_{\mc I}$, ${\mc I}\models_e \phi$.
Let $\phi'\in \mi{Form}_{\mc L}$ and $T\subseteq \mi{Form}_{\mc L}$.
${\mc I}$ is a model of $T$, in symbols ${\mc I}\models T$, iff, for all $\phi\in T$, ${\mc I}\models \phi$.
$\phi$ is a logically valid formula iff, for every interpretation ${\mc I}$ for ${\mc L}$, ${\mc I}\models \phi$.
$\phi$ is equivalent to $\phi'$, in symbols $\phi\eqvl \phi'$, 
iff, for every interpretation ${\mc I}$ for ${\mc L}$ and $e\in {\mc S}_{\mc I}$, $\|\phi\|_e^{\mc I}=\|\phi'\|_e^{\mc I}$.

\section{Translation to clausal form}
\label{S3}

In the propositional case \cite{Guller2010}, we have proposed some translation of a formula to an equivalent conjunctive normal form ({\it CNF}) 
containing literals of the form $a$ or $a\rightarrow b$ or $(a\rightarrow b)\rightarrow b$ 
where $a$ is an atom, and $b$ is an atom or the truth constant $\gz$.
An output equivalent {\it CNF} may be of exponential size with respect to the size of the input formula; 
we had laid no restrictions on use of the distributivity law~(\cref{ceq0d}) during translation to conjunctive normal form.
To avoid this disadvantage, we have devised translation to {\it CNF} via interpolation using new atoms,
which produces an output {\it CNF} of linear size at the cost of being only equivalent satisfiable to the input formula \cite{Guller2018a}.
A similar approach exploiting the renaming subformulae technique can be found in \cite{Tse70,PLGR86,Boy92,Hah94b,NOROWE98,She04}.
A {\it CNF} can further be translated to a finite set of order clauses.
An order clause is a finite set of order literals of the form $\varepsilon_1\diamond \varepsilon_2$
where $\varepsilon_i$ is an atom or a truth constant ($\gz$, $\gu$), and $\diamond$ is a connective $\geql$ or $\gle$.

We now describe some generalisation of the mentioned translation to the first-order case.
At first, we introduce a notion of quantified atom.
Let $a\in \mi{Form}_{\mc L}$.
$a$ is a quantified atom of ${\mc L}$ iff $a=Q x\, p(t_0,\dots,t_n)$
where $p(t_0,\dots,t_n)\in \mi{Atom}_{\mc L}$, $x\in \mi{vars}(p(t_0,\dots,t_n))$,
for all $i\leq n$, either $t_i=x$ or $x\not\in \mi{vars}(t_i)$.
$\mi{QAtom}_{\mc L}\subseteq \mi{Form}_{\mc L}$ denotes the set of all quantified atoms of ${\mc L}$.
$\mi{QAtom}_{\mc L}^Q\subseteq \mi{QAtom}_{\mc L}$, $Q\in \{\forall,\exists\}$, denotes the set of all quantified atoms of ${\mc L}$ 
of the form $Q x\, a$.
Let $\varepsilon_i$, $1\leq i\leq m$, be
either an expression or a set of expressions or a set of sets of expressions of ${\mc L}$, in general.
By $\mi{qatoms}(\varepsilon_1,\dots,\varepsilon_m)\subseteq \mi{QAtom}_{\mc L}$ we denote the set of all quantified atoms of ${\mc L}$
occurring in $\varepsilon_1,\dots,\varepsilon_m$.
We denote $\mi{qatoms}^Q(\varepsilon_1,\dots,\varepsilon_m)=\mi{qatoms}(\varepsilon_1,\dots,\varepsilon_m)\cap \mi{QAtom}_{\mc L}^Q$, 
$Q\in \{\forall,\exists\}$.
Let $p(t_1,\dots,t_n)\in \mi{Atom}_{\mc L}$, $c\in \overline{C}_{\mc L}$, $Q x\, q(t_0,\dots,t_n)\in \mi{QAtom}_{\mc L}$.
We denote
\begin{alignat*}{1}
& \mi{freetermseq}(p(t_1,\dots,t_n))=t_1,\dots,t_n, \\
& \mi{freetermseq}(c)=\emptyseq, \\
& \mi{freetermseq}(Q x\, q(t_0,\dots,t_n))= \\
& \hspace{40.1mm} \{(i,t_i) \,|\, i\leq n, x\not\in \mi{vars}(t_i)\}, \\
& \mi{boundindset}(Q x\, q(t_0,\dots,t_n))=\{i \,|\, i\leq n, t_i=x\}\neq \emptyset.
\end{alignat*}

We further introduce order clauses in G\"{o}del logic.
Let $l\in \mi{Form}_{\mc L}$.
$l$ is an order literal of ${\mc L}$ iff $l=\varepsilon_1\diamond \varepsilon_2$,
$\varepsilon_i\in \mi{Atom}_{\mc L}\cup \overline{C}_{\mc L}\cup \mi{QAtom}_{\mc L}$, $\diamond\in \{\geql,\gle\}$.
The set of all order literals of ${\mc L}$ is designated as $\mi{OrdLit}_{\mc L}\subseteq \mi{Form}_{\mc L}$.
An order clause of ${\mc L}$ is a finite set of order literals of ${\mc L}$.
Since $=$ is symmetric, $\geql$ is commutative;
hence, for all $\varepsilon_1\geql \varepsilon_2\in \mi{OrdLit}_{\mc L}$, we identify
$\varepsilon_1\geql \varepsilon_2$ with $\varepsilon_2\geql \varepsilon_1\in \mi{OrdLit}_{\mc L}$ with respect to order clauses.
An order clause $\{l_0,\dots,l_n\}\neq \emptyset$ is written in the form $l_0\vee\cdots\vee l_n$.
The empty order clause $\emptyset$ is denoted as $\square$.
An order clause $\{l\}$ is called unit and denoted as $l$;
if it does not cause the ambiguity with the denotation of the single order literal $l$ in a given context.
We designate the set of all order clauses of ${\mc L}$ as $\mi{OrdCl}_{\mc L}$.
Let $l, l_0,\dots,l_n\in \mi{OrdLit}_{\mc L}$ and $C, C'\in \mi{OrdCl}_{\mc L}$.
We define the size of $C$ as $|C|=\sum_{l\in C} |l|$.
By $l_0\vee\cdots\vee l_n\vee C$ we denote $\{l_0,\dots,l_n\}\cup C$
where, for all $i, i'\leq n$ and $i\neq i'$, $l_i\not\in C$, $l_i\neq l_{i'}$.
By $C\vee C'$ we denote $C\cup C'$.
$C$ is a subclause of $C'$, in symbols $C\sqsubseteq C'$, iff $C\subseteq C'$.
An order clausal theory of ${\mc L}$ is a set of order clauses of ${\mc L}$.
A unit order clausal theory is a set of unit order clauses; in other words, we say that an order clausal theory is unit.

Let $\phi, \phi'\in \mi{Form}_{\mc L}$, $T, T'\subseteq \mi{Form}_{\mc L}$, $S, S'\subseteq \mi{OrdCl}_{\mc L}$,
${\mc I}$ be an interpretation for ${\mc L}$, $e\in {\mc S}_{\mc I}$.
Note that ${\mc I}\models_e l$ if and only if either $l=\varepsilon_1\geql \varepsilon_2$, $\|\varepsilon_1\geql \varepsilon_2\|_e^{\mc I}=1$, 
$\|\varepsilon_1\|_e^{\mc I}=\|\varepsilon_2\|_e^{\mc I}$; or 
$l=\varepsilon_1\gle \varepsilon_2$, $\|\varepsilon_1\gle \varepsilon_2\|_e^{\mc I}=1$, 
$\|\varepsilon_1\|_e^{\mc I}<\|\varepsilon_2\|_e^{\mc I}$.
$C$ is true in ${\mc I}$ with respect to $e$, written as ${\mc I}\models_e C$, iff there exists $l^*\in C$ such that ${\mc I}\models_e l^*$.
${\mc I}$ is a model of $C$, in symbols ${\mc I}\models C$, iff, for all $e\in {\mc S}_{\mc I}$, ${\mc I}\models_e C$.
${\mc I}$ is a model of $S$, in symbols ${\mc I}\models S$, iff, for all $C\in S$, ${\mc I}\models C$.
Let $\varepsilon_1\in \{\phi,T,C,S\}$ and $\varepsilon_2\in \{\phi',T',C',S'\}$.
$\varepsilon_2$ is a logical consequence of $\varepsilon_1$, in symbols $\varepsilon_1\models \varepsilon_2$,
iff, for every interpretation ${\mc I}$ for ${\mc L}$, if ${\mc I}\models \varepsilon_1$, then ${\mc I}\models \varepsilon_2$.
$\varepsilon_1$ is satisfiable iff there exists a model of $\varepsilon_1$.
Note that both $\square$ and $\square\in S$ are unsatisfiable. 
$\varepsilon_1$ is equisatisfiable to $\varepsilon_2$ iff $\varepsilon_1$ is satisfiable if and only if $\varepsilon_2$ is satisfiable.
Let $S\subseteq_{\mc F} \mi{OrdCl}_{\mc L}$.
We define the size of $S$ as $|S|=\sum_{C\in S} |C|$.
Let $x_1,\dots,x_n\in \mi{freevars}(S)$ and $t_1,\dots,t_n\in \mi{Term}_{\mc L}$ be closed.
By $S(x_1/t_1,\dots,x_n/t_n)\subseteq_{\mc F} \mi{OrdCl}_{\mc L}$ we denote the instance of $S$ by the substitution $x_1/t_1,\dots,x_n/t_n$,
defined in the standard manner.
We shall assume a fresh function symbol $\tilde{f}_0$ such that $\tilde{f}_0\not\in \mi{Func}_{\mc L}$.
We denote $\mbb{I}=\mbb{N}\times \mbb{N}$; $\mbb{I}$ will be exploited as a countably infinite set of indices.
We shall assume a countably infinite set of fresh predicate symbols $\tilde{\mbb{P}}=\{\tilde{p}_\mbbm{i} \,|\, \mbbm{i}\in \mbb{I}\}$ such that 
$\tilde{\mbb{P}}\cap \mi{Pred}_{\mc L}=\emptyset$.

\subsection{A computational point of view}
\label{S3.1}

From a computational point of view, the worst case time and space complexity will be estimated using the logarithmic cost measurement.
Notice that if the estimated upper bound on the space complexity is equal to the estimated upper bound on the time complexity
for some algorithm, then it will not be explicitly stated.
Since our computational framework is only slightly different from that in \cite{Guller2018a}, analogous considerations will be shorten.
Let $n_s\in \mbb{N}$ ($n_s$ can be viewed as an offset in the memory) and 
$E$ be either a term or a formula or an order clause or a finite theory or a finite order clausal theory.
$E$ can be represented by a tree-like data structure ${\mc D}(E)$ having nodes data records.
In Table \ref{tab0}, we introduce all possible forms of data record.
\begin{table}[t]
\caption{Forms of data record}\label{tab0}
\vspace{-6mm}
\centering
\begin{minipage}[t]{\linewidth}
\footnotesize
\begin{IEEEeqnarray*}{*LL*}
\hline \hline \\[0mm]
\text{Expression} & \text{Data record} \\[1mm]
\hline \\[2mm]
x_\mi{index}\in \mi{Var}
& \framebox{$x,\mi{index}$} \\[1mm]
f_\mi{index}(t_1,\dots,t_n)\in \mi{Term}
& \framebox{$f,\mi{index},\mi{pointer}_{t_1},\dots,\mi{pointer}_{t_n}$} \\[1mm]
c_\mi{index}\in \overline{C}
& \framebox{$c,\mi{index}$} \\[1mm]
p_\mi{index}(t_1,\dots,t_n)\in \mi{Atom} 
& \framebox{$p,\mi{index},\mi{pointer}_{t_1},\dots,\mi{pointer}_{t_n}$} \\[1mm]
\diamond \phi_1\in \mi{Form} 
& \framebox{$\diamond,\mi{pointer}_{\phi_1}$} \\[1mm]
\phi_1\diamond \phi_2\in \mi{Form} 
& \framebox{$\diamond,\mi{pointer}_{\phi_1},\mi{pointer}_{\phi_2}$} \\[1mm] 
Q\, x_\mi{index}\, \phi_1\in \mi{Form}
& \framebox{$Q,\mi{pointer}_{x_\mi{index}},\mi{pointer}_{\phi_1}$} \\[1mm]
\square\in \mi{OrdCl} 
& \framebox{$\square$} \\[1mm]
l\vee C\in \mi{OrdCl} 
& \framebox{$|,\mi{pointer}_l,\mi{pointer}_C$} \\[1mm]
\emptyset\subseteq \mi{Form}, \mi{OrdCl} 
& \framebox{$\emptyset$} \\[1mm]
\{\phi\}\cup T\subseteq_{\mc F} \mi{Form}, \phi\not\in T \qquad
& \framebox{$\&,\mi{pointer}_\phi,\mi{pointer}_T$} \\[1mm]
\{C\}\cup S\subseteq_{\mc F} \mi{OrdCl}, C\not\in S 
& \framebox{$\&,\mi{pointer}_C,\mi{pointer}_S$} \\[2mm] 
\hline \hline
\end{IEEEeqnarray*}
\vspace{-4mm} \\
$\mi{pointer}_E$ denotes a pointer which references ${\mc D}(E)$.
\end{minipage}
\vspace{-2mm}
\end{table}
Concerning Table~\ref{tab0}, a data record consists of a field of length in $O(1)$ and
of a constant number (with respect to the size of the input) of indices, pointers.
Variable, function, truth constant, predicate symbols occurring in $E$ can be indexed by indices of the form $(n_s,j)\in \mbb{I}$.
The length of a data record of ${\mc D}(E)$ is in $O(\log (1+n_s)+\log (1+|E|))$.
The time complexity of an elementary operation on ${\mc D}(E)$ is in $O(\log (1+n_s)+\log (1+|E|))$.
The number of data records occurring in ${\mc D}(E)$ is in $O(|E|)$, and the size of ${\mc D}(E)$ is in $O(|E|\cdot (\log (1+n_s)+\log (1+|E|)))$.

Let ${\mc A}$ be an algorithm with inputs $E_0$, $E_1$ which uses only $E_j$, $j=0,\dots,q$, 
where $E_j$ is either a term or a formula or an order clause or a finite theory or a finite order clausal theory;
$q\geq 1$ is a constant (with respect to the size of the input);
there exists a constant $r\geq 1$ satisfying, for all $j\leq q$, $|E_j|\in O(|E_0|^r+|E_1|^r)$.
$\#{\mc O}_{\mc A}(E_0,E_1)\geq 1$ denotes the number of all elementary operations executed by 
${\mc A}$;\footnote{If the algorithm in question is not explicitly designated, we shall only write $\#{\mc O}(E_0,E_1)$.}
we assume that ${\mc A}$ executes at least one elementary operation.
The length of a data record is in $O(\log (1+n_s)+\log (1+|E_0|+|E_1|))$.
The time complexity of an elementary operation on ${\mc D}(E_j)$ executed by ${\mc A}$ is in $O(\log (1+n_s)+\log (1+|E_0|+|E_1|))$.
The number of data records occurring in ${\mc D}(E_j)$, $j=0,\dots,q$, is in $O(|E_0|^r+|E_1|^r)$.
The size of ${\mc D}(E_j)$, $j=0,\dots,q$, {\it area of data records}, is 
in $O((|E_0|^r+|E_1|^r)\cdot (\log (1+n_s)+\log (1+|E_0|+|E_1|)))$.

${\mc A}$ also uses several auxiliary data structures: {\it stack}, {\it index generator}, {\it addressing unit}.
{\it stack} consists of a finite number of frames.
A frame is of the form \framebox{$\mi{field},\mi{pointer}$} where $\mi{field}$ is of length in $O(1)$, and 
$\mi{pointer}$ is a copy of that occurring in ${\mc D}(E_j)$ of length in $O(\log (1+n_s)+\log (1+|E_0|+|E_1|))$.
The length of a frame of {\it stack} is in $O(\log (1+n_s)+\log (1+|E_0|+|E_1|))$.
The time complexity of an elementary operation on {\it stack} executed by ${\mc A}$ is in $O(\log (1+n_s)+\log (1+|E_0|+|E_1|))$.
The size of {\it stack} is in $O(\#{\mc O}_{\mc A}(E_0,E_1)\cdot (\log (1+n_s)+\log (1+|E_0|+|E_1|)))$.

{\it index generator} serves for generating fresh predicate symbols of the form $\tilde{p}_\mbbm{i}\in \tilde{\mbb{P}}$.
It consists of a constant number of indices, of length in $O(\log (1+n_s)+\log (1+|E_0|+|E_1|))$. 
The size of {\it index generator} is in $O(\log (1+n_s)+\log (1+|E_0|+|E_1|))$.
The time complexity of an elementary operation on {\it index generator} executed by ${\mc A}$ is in $O(\log (1+n_s)+\log (1+|E_0|+|E_1|))$.

{\it addressing unit} consists of a constant number of address registers. 
The memory can be arranged as follows: 
{\footnotesize
\settowidth{\llll}{g}
\begin{equation*}
\Big[\ \text{\framebox{{\it stack\phantom{g}\hspace{-\llll}}}} \quad \text{\framebox{{\it index generator}}} \quad 
       \text{\framebox{{\it addressing unit}}} \quad \text{\framebox{{\it area of data records}}}\ \Big].     
\end{equation*}}%
The length of an address register and the size of {\it addressing unit} is in $O(\log (1+n_s)+\log (\#{\mc O}_{\mc A}(E_0,E_1)+|E_0|+|E_1|))$.
The time complexity of an elementary operation on {\it addressing unit} executed by ${\mc A}$ is 
in $O(\log (1+n_s)+\log (\#{\mc O}_{\mc A}(E_0,E_1)+|E_0|+|E_1|))$.
The size of the memory is the total size of {\it stack}, {\it index generator}, {\it addressing unit}, {\it area of data records} 
in $O((\#{\mc O}_{\mc A}(E_0,E_1)+|E_0|^r+|E_1|^r)\cdot (\log (1+n_s)+\log (1+|E_0|+|E_1|)))$.

We assume that ${\mc A}$ executes only elementary operations 
on {\it stack}, {\it index generator}, {\it addressing unit}, {\it area of data records}.
We get that the time complexity of an elementary operation executed by ${\mc A}$ is 
in $O(\log (1+n_s)+\log (\#{\mc O}_{\mc A}(E_0,E_1)+|E_0|+|E_1|))$.
We conclude that
\begin{alignat}{1}
\label{eq00t}   
& \begin{minipage}[t]{\linewidth-15mm}
  the time complexity of ${\mc A}$ on $E_0$ and $E_1$ is 
  in $O(\#{\mc O}_{\mc A}(E_0,E_1)\cdot (\log (1+n_s)+\log (\#{\mc O}_{\mc A}(E_0,E_1)+|E_0|+|E_1|)))$;                  
  \end{minipage}
\end{alignat}
\begin{alignat}{1}
\label{eq00s}     
& \begin{minipage}[t]{\linewidth-15mm}
  the space complexity of ${\mc A}$ on $E_0$ and $E_1$ is 
  in $O((\#{\mc O}_{\mc A}(E_0,E_1)+|E_0|^r+|E_1|^r)\cdot (\log (1+n_s)+\log (1+|E_0|+|E_1|)))$.
  \end{minipage}
\end{alignat}

\subsection{A technical treatment}
\label{S3.2}

We firstly explain generalised translation of a formula to clausal form informally.
Let us consider a formula
$\phi=\forall x\, \big(\exists y\, (q(x,y,z)\gle \gu)\rightarrow \forall z\, r(x,y,z)\geql \overline{0.3}\big)\in \mi{Form}_{\mc L}$.
All the variables occurring in $\phi$ are $x$, $y$, and $z$. 
During translation, we shall introduce auxiliary atoms of the form $\tilde{p}_i(x,y,z)$ with fresh predicate symbols.
Auxiliary atoms will correspond to respective subformulae of $\phi$.
The initial theory of the translation reads as follows:
{\footnotesize
\begin{IEEEeqnarray*}{LR}
\Big\{
  \tilde{p}_0(x,y,z)\geql \gu,
& \\
\phantom{\Big\{}
  \tilde{p}_0(x,y,z)\leftrightarrow 
  \forall x\, \big(\underbrace{\exists y\, (q(x,y,z)\gle \gu)\rightarrow 
                               \forall z\, r(x,y,z)\geql \overline{0.3}}_{\tilde{p}_1(x,y,z)}\big)\Big\}.
& \quad (\ref{eq0rr5+}) 
\end{IEEEeqnarray*}}%
We have introduced an auxiliary atom $\tilde{p}_0(x,y,z)$ corresponding to the entire formula $\phi$.
$\tilde{p}_0(x,y,z)$ is set to equal $\gu$, which causes a positive start of the translation.
The resulting order clausal theory will be equisatisfiable to $\phi$.
The correspondence between $\tilde{p}_0(x,y,z)$ and $\phi$ is expressed by the second equivalence.
The main connective of $\phi$ is a quantifier $\forall x$.
So, we shall apply the unary interpolation rule (\ref{eq0rr5+}) for $\forall$ from Table \ref{tab3}.
We introduce an auxiliary atom $\tilde{p}_1(x,y,z)$ 
for the subformula $\exists y\, (q(x,y,z)\gle \gu)\rightarrow \forall z\, r(x,y,z)\geql \overline{0.3}$ and 
rewrite the initial theory as follows:
{\footnotesize
\begin{IEEEeqnarray*}{LR}
\Big\{
  \tilde{p}_0(x,y,z)\geql \gu,
  \tilde{p}_0(x,y,z)\geql \forall x\, \tilde{p}_1(x,y,z),
& \\
\phantom{\Big\{}
  \tilde{p}_1(x,y,z)\leftrightarrow \big(\underbrace{\exists y\, (q(x,y,z)\gle \gu)}_{\tilde{p}_2(x,y,z)}\rightarrow 
                                         \underbrace{\forall z\, r(x,y,z)\geql \overline{0.3}}_{\tilde{p}_3(x,y,z)}\big)\Big\}.
& \quad (\ref{eq0rr3+}) 
\end{IEEEeqnarray*}}%
We see that the main connective of the subformula $\exists y\, (q(x,y,z)\gle \gu)\rightarrow \forall z\, r(x,y,z)\geql \overline{0.3}$ is
an implication.
Hence, we apply the binary interpolation rule (\ref{eq0rr3+}) for implication from Table~\ref{tab2},
introduce auxiliary atoms $\tilde{p}_2(x,y,z)$, $\tilde{p}_3(x,y,z)$ for the respective subformulae, and obtain another intermediate theory:  
{\footnotesize
\begin{IEEEeqnarray*}{LR}
\Big\{
  \tilde{p}_0(x,y,z)\geql \gu,
  \tilde{p}_0(x,y,z)\geql \forall x\, \tilde{p}_1(x,y,z),
& \\
\phantom{\Big\{}
  \tilde{p}_2(x,y,z)\gle \tilde{p}_3(x,y,z)\vee \tilde{p}_2(x,y,z)\geql \tilde{p}_3(x,y,z)\vee 
& \\
\phantom{\Big\{} \quad
\tilde{p}_1(x,y,z)\geql \tilde{p}_3(x,y,z), 
& \\
\phantom{\Big\{}
  \tilde{p}_3(x,y,z)\gle \tilde{p}_2(x,y,z)\vee \tilde{p}_1(x,y,z)\geql \gu,
& \\
\phantom{\Big\{}
  \tilde{p}_2(x,y,z)\leftrightarrow \exists y\, (\underbrace{q(x,y,z)\gle \gu}_{\tilde{p}_4(x,y,z)}),
& \\
\phantom{\Big\{}
  \tilde{p}_3(x,y,z)\leftrightarrow 
  \underbrace{\forall z\, r(x,y,z)}_{\tilde{p}_5(x,y,z)}\geql \underbrace{\overline{0.3}}_{\tilde{p}_6(x,y,z)}\Big\}.  
& \quad (\ref{eq0rr6+}), (\ref{eq0rr7+}) 
\end{IEEEeqnarray*}}%
We have got two recursive subcases for the subformulae $\exists y\, (q(x,y,z)\gle \gu)$ and $\forall z\, r(x,y,z)\geql \overline{0.3}$ 
(the last two equivalences).
The subcases can be solved by the unary interpolation rule~(\ref{eq0rr6+}) for $\exists$ from Table \ref{tab3} and
by the binary interpolation rule (\ref{eq0rr7+}) for equality from Table~\ref{tab2}, respectively:
{\footnotesize
\begin{IEEEeqnarray*}{LR}
\Big\{
  \tilde{p}_0(x,y,z)\geql \gu,
  \tilde{p}_0(x,y,z)\geql \forall x\, \tilde{p}_1(x,y,z),
& \\
\phantom{\Big\{}
  \tilde{p}_2(x,y,z)\gle \tilde{p}_3(x,y,z)\vee \tilde{p}_2(x,y,z)\geql \tilde{p}_3(x,y,z)\vee 
& \\
\phantom{\Big\{} \quad
\tilde{p}_1(x,y,z)\geql \tilde{p}_3(x,y,z), 
& \\
\phantom{\Big\{}
  \tilde{p}_3(x,y,z)\gle \tilde{p}_2(x,y,z)\vee \tilde{p}_1(x,y,z)\geql \gu,
& \\
\phantom{\Big\{}
  \tilde{p}_2(x,y,z)\geql \exists y\, \tilde{p}_4(x,y,z),
  \tilde{p}_4(x,y,z)\leftrightarrow \underbrace{q(x,y,z)}_{\tilde{p}_7(x,y,z)}\gle \gu,
& \\
\phantom{\Big\{}
  \tilde{p}_5(x,y,z)\geql \tilde{p}_6(x,y,z)\vee \tilde{p}_3(x,y,z)\geql \gz,
& \\
\phantom{\Big\{}
  \tilde{p}_5(x,y,z)\gle \tilde{p}_6(x,y,z)\vee \tilde{p}_6(x,y,z)\gle \tilde{p}_5(x,y,z)\vee 
& \\
\phantom{\Big\{} \quad
\tilde{p}_3(x,y,z)\geql \gu,
& \\
\phantom{\Big\{}
  \tilde{p}_5(x,y,z)\leftrightarrow \forall z \underbrace{r(x,y,z)}_{\tilde{p}_8(x,y,z)},
  \tilde{p}_6(x,y,z)\geql \overline{0.3}\Big\}.
& \quad (\ref{eq0rr888+}), (\ref{eq0rr5+})
\end{IEEEeqnarray*}}%
Applying the unary interpolation rule (\ref{eq0rr888+}) for strict order from Table \ref{tab3} and 
Rule (\ref{eq0rr5+}) for $\forall$ to the last two recursive subcases, respectively, we reach the atomic level and 
get the resulting order clausal theory:
{\footnotesize
\begin{IEEEeqnarray*}{LR}
\Big\{
  \tilde{p}_0(x,y,z)\geql \gu,
  \tilde{p}_0(x,y,z)\geql \forall x\, \tilde{p}_1(x,y,z),
& \\
\phantom{\Big\{}
  \tilde{p}_2(x,y,z)\gle \tilde{p}_3(x,y,z)\vee \tilde{p}_2(x,y,z)\geql \tilde{p}_3(x,y,z)\vee 
& \\
\phantom{\Big\{} \quad
\tilde{p}_1(x,y,z)\geql \tilde{p}_3(x,y,z),
& \\
\phantom{\Big\{}
  \tilde{p}_3(x,y,z)\gle \tilde{p}_2(x,y,z)\vee \tilde{p}_1(x,y,z)\geql \gu,
& \\
\phantom{\Big\{}
  \tilde{p}_2(x,y,z)\geql \exists y\, \tilde{p}_4(x,y,z),
& \\
\phantom{\Big\{}
  \tilde{p}_7(x,y,z)\gle \gu\vee \tilde{p}_4(x,y,z)\geql \gz,
  \tilde{p}_7(x,y,z)\geql \gu\vee \tilde{p}_4(x,y,z)\geql \gu,
& \\
\phantom{\Big\{}
  \tilde{p}_7(x,y,z)\geql q(x,y,z),
& \\
\phantom{\Big\{}
  \tilde{p}_5(x,y,z)\geql \tilde{p}_6(x,y,z)\vee \tilde{p}_3(x,y,z)\geql \gz,
& \\
\phantom{\Big\{}
  \tilde{p}_5(x,y,z)\gle \tilde{p}_6(x,y,z)\vee \tilde{p}_6(x,y,z)\gle \tilde{p}_5(x,y,z)\vee \tilde{p}_3(x,y,z)\geql \gu,
& \\
\phantom{\Big\{}
  \tilde{p}_5(x,y,z)\geql \forall z\, \tilde{p}_8(x,y,z),
  \tilde{p}_8(x,y,z)\geql r(x,y,z), 
  \tilde{p}_6(x,y,z)\geql \overline{0.3}\Big\}.
\end{IEEEeqnarray*}}%
As mentioned above, we have started the translation positively ($\tilde{p}_0(x,y,z)\geql \gu$); 
therefore, the resulting theory is equisatisfiable to $\phi$.

Translation of a formula or theory to clausal form is based on the following lemmata:

\begin{lemma}
\label{le111}
Let $n_\theta\in \mbb{N}$ and $\theta\in \mi{Form}_{\mc L}$.
There exists $\theta'\in \mi{Form}_{\mc L}$ such that 
\begin{enumerate}[\rm (a)]
\item
$\theta'\eqvl \theta$; 
\item 
$|\theta'|\leq 2\cdot |\theta|$; 
$\theta'$ can be built up from $\theta$ via a postorder traversal of $\theta$ with $\#{\mc O}(\theta)\in O(|\theta|)$ and
the time complexity in $O(|\theta|\cdot (\log (1+n_\theta)+\log |\theta|))$;
\item
$\theta'$ does not contain $\neg$ and $\del$; 
\item
$\theta'\in \overline{C}_{\mc L}$; or 
for every subformula of $\theta'$ of the form $\varepsilon_1\diamond \varepsilon_2$, $\diamond\in \{\wedge,\vee,\leftrightarrow\}$, 
$\varepsilon_i\neq \gz, \gu$, $\{\varepsilon_1,\varepsilon_2\}\not\subseteq \overline{C}_{\mc L}$;
for every subformula of $\theta'$ of the form $\varepsilon_1\rightarrow \varepsilon_2$,
$\varepsilon_1\neq \gz, \gu$, $\varepsilon_2\neq \gu$, $\{\varepsilon_1,\varepsilon_2\}\not\subseteq \overline{C}_{\mc L}$;
for every subformula of $\theta'$ of the form $\varepsilon_1\geql \varepsilon_2$,
$\{\varepsilon_1,\varepsilon_2\}\not\subseteq \overline{C}_{\mc L}$;
for every subformula of $\theta'$ of the form $\varepsilon_1\gle \varepsilon_2$,
$\varepsilon_1\neq \gu$, $\varepsilon_2\neq \gz$, $\{\varepsilon_1,\varepsilon_2\}\not\subseteq \overline{C}_{\mc L}$;
for every subformula of $\theta'$ of the form $Q x\, \varepsilon_1$, $\varepsilon_1\not\in \overline{C}_{\mc L}$;
\item
$\mi{tcons}(\theta')-\{\gz,\gu\}\subseteq \mi{tcons}(\theta)-\{\gz,\gu\}$.
\end{enumerate}
\end{lemma}

\begin{proof}
The proof is by induction on the structure of $\theta$ using (\ref{eq0b}), (\ref{eq0kk}), and the obvious simplification identities
on $\mbi{G}$ with respect to $0$, $1$, $c_1, c_2\in [0,1]$, $\fvee$, $\fwedge$, $\frightarrow$, $\feql$, $\fle$ 
(e.g. $0\fvee a=a$, $1\fvee a=1$, $a\fvee a=a$, $c_1\fvee c_2=c_1$ iff $c_1\geq c_2$;
      $0\fwedge a=0$, $1\fwedge a=a$, $a\fwedge a=a$, $c_1\fwedge c_2=c_1$ iff $c_1\leq c_2$;
      $0\frightarrow a=1$, $1\frightarrow a=a$, $a\frightarrow 1=1$, $a\frightarrow a=1$, 
      $c_1\frightarrow c_2=1$ iff $c_1\leq c_2$, $c_1\frightarrow c_2=c_2$ iff $c_1>c_2$;
      $a\feql a=1$, $c_1\feql c_2=1$ iff $c_1=c_2$, $c_1\feql c_2=0$ iff $c_1\neq c_2$;
      $a\fle 0=0$, $1\fle a=0$, $a\fle a=0$, $c_1\fle c_2=1$ iff $c_1<c_2$, $c_1\fle c_2=0$ iff $c_1\geq c_2$, etc.);
the postorder traversal of $\theta$ uses the input $\theta$ and the output $\theta'$,
$|\theta'|\underset{\text{(b)}}{\leq} 2\cdot |\theta|\in O(|\theta|)$,
$\#{\mc O}(\theta)\in O(|\theta|)$;
by (\ref{eq00t}) for $n_\theta$, $\theta$, $\emptyset$, $\theta'$, $q=2$, $r=1$,
the time complexity of the postorder traversal of $\theta$ is
in $O(\#{\mc O}(\theta)\cdot (\log (1+n_\theta)+\log (\#{\mc O}(\theta)+|\theta|)))\subseteq
    O(|\theta|\cdot (\log (1+n_\theta)+\log |\theta|))$;
by (\ref{eq00s}) for $n_\theta$, $\theta$, $\emptyset$, $\theta'$, $q=2$, $r=1$,
the space complexity of the postorder traversal of $\theta$ is
in $O((\#{\mc O}(\theta)+|\theta|)\cdot (\log (1+n_\theta)+\log |\theta|))\subseteq
    O(|\theta|\cdot (\log (1+n_\theta)+\log |\theta|))$.
\qed
\end{proof}

\begin{lemma}
\label{le11}
Let $n_\theta\in \mbb{N}$, $\theta\in \mi{Form}_{\mc L}-\{\gz,\gu\}$, {\rm (c,d)} of Lemma \ref{le111} hold for $\theta$;
$\bar{x}$ be a sequence of variables of ${\mc L}$, $\mi{vars}(\theta)\subseteq \mi{vars}(\bar{x})$; 
$j_\mbbm{i}\in \mbb{N}$, $\mbbm{i}=(n_\theta,j_\mbbm{i})\in \{(n_\theta,j) \,|\, j\in \mbb{N}\}\subseteq \mbb{I}$,
$\tilde{p}_\mbbm{i}\in \tilde{\mbb{P}}$, $\mi{ar}(\tilde{f}_0)=\mi{ar}(\tilde{p}_\mbbm{i})=|\bar{x}|$.
There exist $n_J\geq j_\mbbm{i}$, 
$J=\{(n_\theta,j) \,|\, j_\mbbm{i}+1\leq j\leq n_J\}\subseteq \{(n_\theta,j) \,|\, j\in \mbb{N}\}\subseteq \mbb{I}$, $\mbbm{i}\not\in J$, 
$S\subseteq_{\mc F} \mi{OrdCl}_{{\mc L}\cup \{\tilde{p}_\mbbm{i}\}\cup \{\tilde{p}_\mbbm{j} \,|\, \mbbm{j}\in J\}}$ such that
\begin{enumerate}[\rm (a)]
\item
$\|J\|\leq |\theta|-1$;
\item
there exists an interpretation ${\mf A}$ for ${\mc L}\cup \{\tilde{p}_\mbbm{i}\}$ and
${\mf A}\models \tilde{p}_\mbbm{i}(\bar{x})\leftrightarrow \theta\in \mi{Form}_{{\mc L}\cup \{\tilde{p}_\mbbm{i}\}}$ if and only if 
there exists an interpretation ${\mf A}'$ for ${\mc L}\cup \{\tilde{p}_\mbbm{i}\}\cup \{\tilde{p}_\mbbm{j} \,|\, \mbbm{j}\in J\}$ and
${\mf A}'\models S$, 
satisfying ${\mf A}={\mf A}'|_{{\mc L}\cup \{\tilde{p}_\mbbm{i}\}}$;
\item
$|S|\leq 27\cdot |\theta|\cdot (1+|\bar{x}|)$,
$S$ can be built up from $\theta$ and $\tilde{f}_0(\bar{x})$ via a preorder traversal of $\theta$
with $\#{\mc O}(\theta,\tilde{f}_0(\bar{x}))\in O(|\theta|\cdot (1+|\bar{x}|))$; 
\item
for all $C\in S$,        
$\emptyset\neq \mi{preds}(C)\cap \tilde{\mbb{P}}\subseteq \{\tilde{p}_\mbbm{i}\}\cup \{\tilde{p}_\mbbm{j} \,|\, \mbbm{j}\in J\}$,
$\tilde{p}_\mbbm{i}(\bar{x})\geql \gu, \tilde{p}_\mbbm{i}(\bar{x})\gle \gu\not\in S$;
\item
$\mi{tcons}(S)-\{\gz,\gu\}=\mi{tcons}(\theta)-\{\gz,\gu\}$.
\end{enumerate}
\end{lemma}

\begin{proof}
We proceed by induction on the structure of $\theta$.

Case 1 (the base case):
$\theta\in \mi{Atom}_{\mc L}\cup \overline{C}_{\mc L}$.
We put $n_J=j_\mbbm{i}$ and $J=\emptyset\subseteq \{(n_\theta,j) \,|\, j\in \mbb{N}\}\subseteq \mbb{I}$.
Then $n_J\geq j_\mbbm{i}$, $\mbbm{i}\not\in J$,
$\tilde{p}_\mbbm{i}(\bar{x})\in \mi{Atom}_{{\mc L}\cup \{\tilde{p}_\mbbm{i}\}}$,
$\tilde{p}_\mbbm{i}(\bar{x})\geql \theta\in \mi{OrdLit}_{{\mc L}\cup \{\tilde{p}_\mbbm{i}\}}$.
We put $S=\{\tilde{p}_\mbbm{i}(\bar{x})\geql \theta\}\subseteq_{\mc F} \mi{OrdCl}_{{\mc L}\cup \{\tilde{p}_\mbbm{i}\}}$.

(a) and (c--e) can be proved straightforwardly.

For every interpretation ${\mf A}$ for ${\mc L}\cup \{\tilde{p}_\mbbm{i}\}$,
for all $e\in {\mc S}_{\mf A}$, 
${\mf A}\models_e \tilde{p}_\mbbm{i}(\bar{x})\leftrightarrow \theta\in \mi{Form}_{{\mc L}\cup \{\tilde{p}_\mbbm{i}\}}$ if and only if 
${\mf A}\models_e \tilde{p}_\mbbm{i}(\bar{x})\geql \theta$;
${\mf A}\models \tilde{p}_\mbbm{i}(\bar{x})\leftrightarrow \theta$ if and only if ${\mf A}\models S$.
Hence, there exists an interpretation ${\mf A}$ for ${\mc L}\cup \{\tilde{p}_\mbbm{i}\}$ and
${\mf A}\models \tilde{p}_\mbbm{i}(\bar{x})\leftrightarrow \theta$
if and only if there exists an interpretation ${\mf A}'$ for ${\mc L}\cup \{\tilde{p}_\mbbm{i}\}$ and
${\mf A}'\models S$, satisfying ${\mf A}={\mf A}'={\mf A}'|_{{\mc L}\cup \{\tilde{p}_\mbbm{i}\}}$; 
(b) holds.

Case 2 (the induction case): 
$\theta\in \mi{Form}_{\mc L}-(\mi{Atom}_{\mc L}\cup \overline{C}_{\mc L})$.
We have that (c,d) of Lemma \ref{le111} hold for $\theta$.
We distinguish two cases for $\theta$.

Case 2.1 (the binary interpolation case):
$\theta=\theta_1\diamond \theta_2$, $\diamond\in \{\wedge,\vee,\rightarrow,\leftrightarrow,\geql,\gle\}$,
$\theta_i\in \mi{Form}_{\mc L}-\{\gz,\gu\}$.
We have that (c,d) of Lemma \ref{le111} hold for $\theta$, $\mi{vars}(\theta)\subseteq \mi{vars}(\bar{x})$.
Then, for both $i$, 
(c,d) of Lemma \ref{le111} hold for $\theta_i$,
$\mi{vars}(\theta_i)\subseteq \mi{vars}(\theta)\subseteq \mi{vars}(\bar{x})$.
We put $j_{\mbbm{i}_1}=j_\mbbm{i}+1$ and $\mbbm{i}_1=(n_\theta,j_{\mbbm{i}_1})\in \{(n_\theta,j) \,|\, j\in \mbb{N}\}\subseteq \mbb{I}$. 
$\tilde{p}_{\mbbm{i}_1}\in \tilde{\mbb{P}}$.
We put $\mi{ar}(\tilde{p}_{\mbbm{i}_1})=|\bar{x}|$.
We get by the induction hypothesis for $\theta_1$, $j_{\mbbm{i}_1}$, $\mbbm{i}_1$, $\tilde{p}_{\mbbm{i}_1}$ that
there exist $n_{J_1}\geq j_{\mbbm{i}_1}$,
$J_1=\{(n_\theta,j) \,|\, j_{\mbbm{i}_1}+1\leq j\leq n_{J_1}\}\subseteq \{(n_\theta,j) \,|\, j\in \mbb{N}\}\subseteq \mbb{I}$, 
$\mbbm{i}_1\not\in J_1$,
$S_1\subseteq_{\mc F} \mi{OrdCl}_{{\mc L}\cup \{\tilde{p}_{\mbbm{i}_1}\}\cup \{\tilde{p}_\mbbm{j} \,|\, \mbbm{j}\in J_1\}}$, and
(a--e) hold for $\theta_1$, $\tilde{p}_{\mbbm{i}_1}$, $J_1$, $S_1$.
We put $j_{\mbbm{i}_2}=n_{J_1}+1$ and $\mbbm{i}_2=(n_\theta,j_{\mbbm{i}_2})\in \{(n_\theta,j) \,|\, j\in \mbb{N}\}\subseteq \mbb{I}$.
$\tilde{p}_{\mbbm{i}_2}\in \tilde{\mbb{P}}$.
We put $\mi{ar}(\tilde{p}_{\mbbm{i}_2})=|\bar{x}|$.
We get by the induction hypothesis for $\theta_2$, $j_{\mbbm{i}_2}$, $\mbbm{i}_2$, $\tilde{p}_{\mbbm{i}_2}$ that
there exist $n_{J_2}\geq j_{\mbbm{i}_2}$,
$J_2=\{(n_\theta,j) \,|\, j_{\mbbm{i}_2}+1\leq j\leq n_{J_2}\}\subseteq \{(n_\theta,j) \,|\, j\in \mbb{N}\}\subseteq \mbb{I}$, 
$\mbbm{i}_2\not\in J_2$,
$S_2\subseteq_{\mc F} \mi{OrdCl}_{{\mc L}\cup \{\tilde{p}_{\mbbm{i}_2}\}\cup \{\tilde{p}_\mbbm{j} \,|\, \mbbm{j}\in J_2\}}$, and
(a--e) hold for $\theta_2$, $\tilde{p}_{\mbbm{i}_2}$, $J_2$, $S_2$.
We put $n_J=n_{J_2}$ and $J=\{(n_\theta,j) \,|\, j_\mbbm{i}+1\leq j\leq n_J\}\subseteq \{(n_\theta,j) \,|\, j\in \mbb{N}\}\subseteq \mbb{I}$.
Then $j_\mbbm{i}<j_{\mbbm{i}_1}\leq n_{J_1}<j_{\mbbm{i}_2}\leq n_J$, $\mbbm{i}\not\in J$,
\begin{alignat}{1}
\label{eq1a}
& J=\{\mbbm{i}_1\}\cup J_1\cup \{\mbbm{i}_2\}\cup J_2, \\[0mm]
\label{eq1b}
& \{\mbbm{i}\}, \{\mbbm{i}_1\}, J_1, \{\mbbm{i}_2\}, J_2\ \text{are pairwise disjoint}.
\end{alignat}   
\begin{table*}[p]
\caption{Binary interpolation rules for $\wedge$, $\vee$, $\rightarrow$, $\leftrightarrow$, $\geql$, $\gle$}\label{tab2}
\vspace{-6mm}
\centering
\begin{minipage}[t]{\linewidth-30mm}
\footnotesize
\begin{IEEEeqnarray}{*LL}
\hline \hline \notag \\[0mm]
\notag 
\text{\bf Case} & \\[1mm]
\hline \notag \\[2mm]
\label{eq0rr1+}
\mbi{\theta=\theta_1\wedge \theta_2} & 
\dfrac{\tilde{p}_\mbbm{i}(\bar{x})\leftrightarrow \theta_1\wedge \theta_2}
      {\left\{\begin{array}{l}
              \tilde{p}_{\mbbm{i}_1}(\bar{x})\gle \tilde{p}_{\mbbm{i}_2}(\bar{x})\vee \tilde{p}_{\mbbm{i}_1}(\bar{x})\geql \tilde{p}_{\mbbm{i}_2}(\bar{x})\vee \tilde{p}_\mbbm{i}(\bar{x})\geql \tilde{p}_{\mbbm{i}_2}(\bar{x}), \\
              \tilde{p}_{\mbbm{i}_2}(\bar{x})\gle \tilde{p}_{\mbbm{i}_1}(\bar{x})\vee \tilde{p}_\mbbm{i}(\bar{x})\geql \tilde{p}_{\mbbm{i}_1}(\bar{x}),               
              \tilde{p}_{\mbbm{i}_1}(\bar{x})\leftrightarrow \theta_1, \tilde{p}_{\mbbm{i}_2}(\bar{x})\leftrightarrow \theta_2
              \end{array}\right\}} \\[2mm]
\IEEEeqnarraymulticol{2}{l}{
|\text{Consequent}|=
15+10\cdot |\bar{x}|+|\tilde{p}_{\mbbm{i}_1}(\bar{x})\leftrightarrow \theta_1|+|\tilde{p}_{\mbbm{i}_2}(\bar{x})\leftrightarrow \theta_2|\leq
27\cdot (1+|\bar{x}|)+|\tilde{p}_{\mbbm{i}_1}(\bar{x})\leftrightarrow \theta_1|+|\tilde{p}_{\mbbm{i}_2}(\bar{x})\leftrightarrow \theta_2|} \notag \\[6mm]
\label{eq0rr2+}
\mbi{\theta=\theta_1\vee \theta_2} & 
\dfrac{\tilde{p}_\mbbm{i}(\bar{x})\leftrightarrow (\theta_1\vee \theta_2)}
      {\left\{\begin{array}{l}
              \tilde{p}_{\mbbm{i}_1}(\bar{x})\gle \tilde{p}_{\mbbm{i}_2}(\bar{x})\vee \tilde{p}_{\mbbm{i}_1}(\bar{x})\geql \tilde{p}_{\mbbm{i}_2}(\bar{x})\vee \tilde{p}_\mbbm{i}(\bar{x})\geql \tilde{p}_{\mbbm{i}_1}(\bar{x}), \\
              \tilde{p}_{\mbbm{i}_2}(\bar{x})\gle \tilde{p}_{\mbbm{i}_1}(\bar{x})\vee \tilde{p}_\mbbm{i}(\bar{x})\geql \tilde{p}_{\mbbm{i}_2}(\bar{x}),               
              \tilde{p}_{\mbbm{i}_1}(\bar{x})\leftrightarrow \theta_1, \tilde{p}_{\mbbm{i}_2}(\bar{x})\leftrightarrow \theta_2
              \end{array}\right\}} \\[2mm]
\IEEEeqnarraymulticol{2}{l}{
|\text{Consequent}|=
15+10\cdot |\bar{x}|+|\tilde{p}_{\mbbm{i}_1}(\bar{x})\leftrightarrow \theta_1|+|\tilde{p}_{\mbbm{i}_2}(\bar{x})\leftrightarrow \theta_2|\leq
27\cdot (1+|\bar{x}|)+|\tilde{p}_{\mbbm{i}_1}(\bar{x})\leftrightarrow \theta_1|+|\tilde{p}_{\mbbm{i}_2}(\bar{x})\leftrightarrow \theta_2|} \notag \\[6mm]
\label{eq0rr3+}
\begin{array}{l}
\mbi{\theta=\theta_1\rightarrow \theta_2,} \\ 
\mbi{\theta_2\neq \gz} 
\end{array} & 
\dfrac{\tilde{p}_\mbbm{i}(\bar{x})\leftrightarrow (\theta_1\rightarrow \theta_2)}
      {\left\{\begin{array}{l}
              \tilde{p}_{\mbbm{i}_1}(\bar{x})\gle \tilde{p}_{\mbbm{i}_2}(\bar{x})\vee \tilde{p}_{\mbbm{i}_1}(\bar{x})\geql \tilde{p}_{\mbbm{i}_2}(\bar{x})\vee \tilde{p}_\mbbm{i}(\bar{x})\geql \tilde{p}_{\mbbm{i}_2}(\bar{x}), \\
              \tilde{p}_{\mbbm{i}_2}(\bar{x})\gle \tilde{p}_{\mbbm{i}_1}(\bar{x})\vee \tilde{p}_\mbbm{i}(\bar{x})\geql \gu,               
              \tilde{p}_{\mbbm{i}_1}(\bar{x})\leftrightarrow \theta_1, \tilde{p}_{\mbbm{i}_2}(\bar{x})\leftrightarrow \theta_2
              \end{array}\right\}} \\[2mm]
\IEEEeqnarraymulticol{2}{l}{
|\text{Consequent}|=
15+9\cdot |\bar{x}|+|\tilde{p}_{\mbbm{i}_1}(\bar{x})\leftrightarrow \theta_1|+|\tilde{p}_{\mbbm{i}_2}(\bar{x})\leftrightarrow \theta_2|\leq
27\cdot (1+|\bar{x}|)+|\tilde{p}_{\mbbm{i}_1}(\bar{x})\leftrightarrow \theta_1|+|\tilde{p}_{\mbbm{i}_2}(\bar{x})\leftrightarrow \theta_2|} \notag \\[6mm]
\label{eq0rr33+}
\mbi{\theta=\theta_1\leftrightarrow \theta_2} & 
\dfrac{\tilde{p}_\mbbm{i}(\bar{x})\leftrightarrow (\theta_1\leftrightarrow \theta_2)}
      {\left\{\begin{array}{l}
              \tilde{p}_{\mbbm{i}_1}(\bar{x})\gle \tilde{p}_{\mbbm{i}_2}(\bar{x})\vee \tilde{p}_{\mbbm{i}_1}(\bar{x})\geql \tilde{p}_{\mbbm{i}_2}(\bar{x})\vee \tilde{p}_\mbbm{i}(\bar{x})\geql \tilde{p}_{\mbbm{i}_2}(\bar{x}), \\
              \tilde{p}_{\mbbm{i}_2}(\bar{x})\gle \tilde{p}_{\mbbm{i}_1}(\bar{x})\vee \tilde{p}_{\mbbm{i}_2}(\bar{x})\geql \tilde{p}_{\mbbm{i}_1}(\bar{x})\vee \tilde{p}_\mbbm{i}(\bar{x})\geql \tilde{p}_{\mbbm{i}_1}(\bar{x}), \\
              \tilde{p}_{\mbbm{i}_1}(\bar{x})\gle \tilde{p}_{\mbbm{i}_2}(\bar{x})\vee \tilde{p}_{\mbbm{i}_2}(\bar{x})\gle \tilde{p}_{\mbbm{i}_1}(\bar{x})\vee \tilde{p}_\mbbm{i}(\bar{x})\geql \gu, 
              \tilde{p}_{\mbbm{i}_1}(\bar{x})\leftrightarrow \theta_1, \tilde{p}_{\mbbm{i}_2}(\bar{x})\leftrightarrow \theta_2
              \end{array}\right\}} \\[2mm]
\IEEEeqnarraymulticol{2}{l}{
|\text{Consequent}|=
27+17\cdot |\bar{x}|+|\tilde{p}_{\mbbm{i}_1}(\bar{x})\leftrightarrow \theta_1|+|\tilde{p}_{\mbbm{i}_2}(\bar{x})\leftrightarrow \theta_2|\leq
27\cdot (1+|\bar{x}|)+|\tilde{p}_{\mbbm{i}_1}(\bar{x})\leftrightarrow \theta_1|+|\tilde{p}_{\mbbm{i}_2}(\bar{x})\leftrightarrow \theta_2|} \notag \\[6mm]
\label{eq0rr7+}
\begin{array}{l}
\mbi{\theta=\theta_1\geql \theta_2,} \\
\mbi{\theta_i\neq \gz, \gu} 
\end{array} & 
\dfrac{\tilde{p}_\mbbm{i}(\bar{x})\leftrightarrow (\theta_1\geql \theta_2)}
      {\left\{\begin{array}{l}
              \tilde{p}_{\mbbm{i}_1}(\bar{x})\geql \tilde{p}_{\mbbm{i}_2}(\bar{x})\vee \tilde{p}_\mbbm{i}(\bar{x})\geql \gz, \\
              \tilde{p}_{\mbbm{i}_1}(\bar{x})\gle \tilde{p}_{\mbbm{i}_2}(\bar{x})\vee \tilde{p}_{\mbbm{i}_2}(\bar{x})\gle \tilde{p}_{\mbbm{i}_1}(\bar{x})\vee \tilde{p}_\mbbm{i}(\bar{x})\geql \gu, 
              \tilde{p}_{\mbbm{i}_1}(\bar{x})\leftrightarrow \theta_1, \tilde{p}_{\mbbm{i}_2}(\bar{x})\leftrightarrow \theta_2
              \end{array}\right\}} \\[2mm]
\IEEEeqnarraymulticol{2}{l}{
|\text{Consequent}|=
15+8\cdot |\bar{x}|+|\tilde{p}_{\mbbm{i}_1}(\bar{x})\leftrightarrow \theta_1|+|\tilde{p}_{\mbbm{i}_2}(\bar{x})\leftrightarrow \theta_2|\leq
27\cdot (1+|\bar{x}|)+|\tilde{p}_{\mbbm{i}_1}(\bar{x})\leftrightarrow \theta_1|+|\tilde{p}_{\mbbm{i}_2}(\bar{x})\leftrightarrow \theta_2|} \notag \\[6mm]
\label{eq0rr8+}
\begin{array}{l}
\mbi{\theta=\theta_1\gle \theta_2,} \\ 
\mbi{\theta_1\neq \gz, \theta_2\neq \gu} 
\end{array} \qquad & 
\dfrac{\tilde{p}_\mbbm{i}(\bar{x})\leftrightarrow (\theta_1\gle \theta_2)}
      {\left\{\begin{array}{l}
              \tilde{p}_{\mbbm{i}_1}(\bar{x})\gle \tilde{p}_{\mbbm{i}_2}(\bar{x})\vee \tilde{p}_\mbbm{i}(\bar{x})\geql \gz, \\
              \tilde{p}_{\mbbm{i}_2}(\bar{x})\gle \tilde{p}_{\mbbm{i}_1}(\bar{x})\vee \tilde{p}_{\mbbm{i}_2}(\bar{x})\geql \tilde{p}_{\mbbm{i}_1}(\bar{x})\vee \tilde{p}_\mbbm{i}(\bar{x})\geql \gu, 
              \tilde{p}_{\mbbm{i}_1}(\bar{x})\leftrightarrow \theta_1, \tilde{p}_{\mbbm{i}_2}(\bar{x})\leftrightarrow \theta_2
              \end{array}\right\}} \\[2mm]
\IEEEeqnarraymulticol{2}{l}{
|\text{Consequent}|=
15+8\cdot |\bar{x}|+|\tilde{p}_{\mbbm{i}_1}(\bar{x})\leftrightarrow \theta_1|+|\tilde{p}_{\mbbm{i}_2}(\bar{x})\leftrightarrow \theta_2|\leq
27\cdot (1+|\bar{x}|)+|\tilde{p}_{\mbbm{i}_1}(\bar{x})\leftrightarrow \theta_1|+|\tilde{p}_{\mbbm{i}_2}(\bar{x})\leftrightarrow \theta_2|} \notag \\[2mm]
\hline \hline \notag
\end{IEEEeqnarray}
\vspace{-4mm} \\
$\text{Consequent}$ denotes the consequent ("denominator") part of an interpolation rule.
\end{minipage}
\vspace{-2mm}
\end{table*}
In Table~\ref{tab2}, for every form of $\theta$, a binary interpolation rule of the respective form
\begin{alignat*}{1}
& \dfrac{\tilde{p}_\mbbm{i}(\bar{x})\leftrightarrow \theta\in \mi{Form}_{{\mc L}\cup \{\tilde{p}_\mbbm{i}\}}}
        {\mi{ClPrefix}\cup \{\gamma_1\}\cup \{\gamma_2\}\subseteq_{\mc F}
         \mi{Form}_{{\mc L}\cup \{\tilde{p}_\mbbm{i}\}\cup \{\tilde{p}_{\mbbm{i}_1}\}\cup \{\tilde{p}_{\mbbm{i}_2}\}}}, \\[0mm]
& \mi{ClPrefix}\subseteq_{\mc F}
  \mi{OrdCl}_{{\mc L}\cup \{\tilde{p}_\mbbm{i}\}\cup \{\tilde{p}_{\mbbm{i}_1}\}\cup \{\tilde{p}_{\mbbm{i}_2}\}}, \\
& \gamma_i=\tilde{p}_{\mbbm{i}_i}(\bar{x})\leftrightarrow \theta_i\in \mi{Form}_{{\mc L}\cup \{\tilde{p}_{\mbbm{i}_i}\}}, 
\end{alignat*}
is assigned.
We put 
\begin{alignat*}{1}
& S=\mi{ClPrefix}\cup S_1\cup S_2\subseteq_{\mc F} \\
& \phantom{S=\mbox{}}
    \mi{OrdCl}_{{\mc L}\cup \{\tilde{p}_\mbbm{i}\}\cup 
                \{\tilde{p}_\mbbm{j} \,|\, \mbbm{j}\in \{\mbbm{i}_1\}\cup J_1\cup \{\mbbm{i}_2\}\cup J_2\overset{\text{(\ref{eq1a})}}{=\!\!=} 
                                                       J\}}.
\end{alignat*}
It can be proved that
\begin{equation}
\label{eq1c}
\mi{ClPrefix}, S_1, S_2\ \text{are pairwise disjoint}.
\end{equation}

(a) and (c--e) can be proved straightforwardly.

Let ${\mf A}$ be an interpretation for ${\mc L}\cup \{\tilde{p}_\mbbm{i}\}$. 
We define an expansion ${\mf A}^\#$ of ${\mf A}$ 
to ${\mc L}\cup \{\tilde{p}_\mbbm{i}\}\cup \{\tilde{p}_{\mbbm{i}_1}\}\cup \{\tilde{p}_{\mbbm{i}_2}\}$ as follows:
\begin{equation*}
\tilde{p}_{\mbbm{i}_i}^{{\mf A}^\#}(u_1,\dots,u_{|\bar{x}|})=
\left\{\begin{array}{ll}
       \|\theta_i\|_e^{\mf A} &\ \text{\it if there exists}\ e\in {\mc S}_{\mf A}\ \text{\it such that} \\ 
                              &\ \quad \|\bar{x}\|_e^{\mf A}=u_1,\dots,u_{|\bar{x}|}, \\[1mm]
       0                      &\ \text{\it else}, \hfill i=1,2.
       \end{array}
\right. 
\end{equation*}
Then, for both $i$, for all $e\in {\mc S}_{{\mf A}^\#}$, 
\begin{alignat*}{1}
& \|\tilde{p}_{\mbbm{i}_i}(\bar{x})\|_e^{{\mf A}^\#}\!\!=\tilde{p}_{\mbbm{i}_i}^{{\mf A}^\#}(\|\bar{x}\|_e^{{\mf A}^\#})=
  \tilde{p}_{\mbbm{i}_i}^{{\mf A}^\#}(\|\bar{x}\|_e^{\mf A})=\|\theta_i\|_e^{\mf A}=\|\theta_i\|_e^{{\mf A}^\#}\!\!, \\[1mm]
& \|\gamma_i\|_e^{{\mf A}^\#}\!\!=
  (\|\tilde{p}_{\mbbm{i}_i}(\bar{x})\|_e^{{\mf A}^\#}\!\!\frightarrow \|\theta_i\|_e^{{\mf A}^\#})\fwedge
  (\|\theta_i\|_e^{{\mf A}^\#}\!\!\frightarrow \|\tilde{p}_{\mbbm{i}_i}(\bar{x})\|_e^{{\mf A}^\#})= \\
& \phantom{\|\gamma_i\|_e^{{\mf A}^\#}\!\!=\mbox{}}
  \|\theta_i\|_e^{{\mf A}^\#}\frightarrow \|\theta_i\|_e^{{\mf A}^\#}=1; 
\end{alignat*}
${\mf A}^\#\models \gamma_i$.

Let ${\mf A}\models \tilde{p}_\mbbm{i}(\bar{x})\leftrightarrow \theta$.
Then ${\mf A}^\#\models \tilde{p}_\mbbm{i}(\bar{x})\leftrightarrow \theta$,
for all $e\in {\mc S}_{{\mf A}^\#}$, 
{\footnotesize
\begin{alignat*}{1}
& \|\tilde{p}_\mbbm{i}(\bar{x})\leftrightarrow \theta\|_e^{{\mf A}^\#}=1,\ 
  \|\tilde{p}_\mbbm{i}(\bar{x})\leftrightarrow (\theta_1\diamond \theta_2)\|_e^{{\mf A}^\#}=1, \\[1mm]
& (\|\tilde{p}_\mbbm{i}(\bar{x})\|_e^{{\mf A}^\#}\frightarrow (\|\theta_1\|_e^{{\mf A}^\#}\fdiamond \|\theta_2\|_e^{{\mf A}^\#}))\fwedge \\
& \quad
  ((\|\theta_1\|_e^{{\mf A}^\#}\fdiamond \|\theta_2\|_e^{{\mf A}^\#})\frightarrow \|\tilde{p}_\mbbm{i}(\bar{x})\|_e^{{\mf A}^\#})=1, \\[1mm]
& \|\tilde{p}_\mbbm{i}(\bar{x})\|_e^{{\mf A}^\#}\frightarrow (\|\theta_1\|_e^{{\mf A}^\#}\fdiamond \|\theta_2\|_e^{{\mf A}^\#})=1, \\
& (\|\theta_1\|_e^{{\mf A}^\#}\fdiamond \|\theta_2\|_e^{{\mf A}^\#})\frightarrow \|\tilde{p}_\mbbm{i}(\bar{x})\|_e^{{\mf A}^\#}=1, \\[1mm]
& \|\tilde{p}_\mbbm{i}(\bar{x})\|_e^{{\mf A}^\#}\leq (\|\theta_1\|_e^{{\mf A}^\#}\fdiamond \|\theta_2\|_e^{{\mf A}^\#}),\
  (\|\theta_1\|_e^{{\mf A}^\#}\fdiamond \|\theta_2\|_e^{{\mf A}^\#})\leq \|\tilde{p}_\mbbm{i}(\bar{x})\|_e^{{\mf A}^\#}, \\[1mm]
& \|\tilde{p}_\mbbm{i}(\bar{x})\|_e^{{\mf A}^\#}=
  (\|\tilde{p}_{\mbbm{i}_1}(\bar{x})\|_e^{{\mf A}^\#}\fdiamond \|\tilde{p}_{\mbbm{i}_2}(\bar{x})\|_e^{{\mf A}^\#}),
\end{alignat*}}%
concerning Table \ref{tab2}, for every form of $\theta$, for all $C\in \mi{ClPrefix}$, ${\mf A}^\#\models_e C$;
${\mf A}^\#\models \mi{ClPrefix}$,
for both $i$, ${\mf A}^\#|_{{\mc L}\cup \{\tilde{p}_{\mbbm{i}_i}\}}\models \gamma_i$,
by the induction hypothesis (b) for $\theta_i$, ${\mf A}^\#|_{{\mc L}\cup \{\tilde{p}_{\mbbm{i}_i}\}}$, 
there exists an interpretation ${\mf A}_i$ for ${\mc L}\cup \{\tilde{p}_{\mbbm{i}_i}\}\cup \{\tilde{p}_\mbbm{j} \,|\, \mbbm{j}\in J_i\}$, and 
${\mf A}_i\models S_i$, ${\mf A}_i|_{{\mc L}\cup \{\tilde{p}_{\mbbm{i}_i}\}}={\mf A}^\#|_{{\mc L}\cup \{\tilde{p}_{\mbbm{i}_i}\}}$.
By (\ref{eq1b}), $\{\tilde{p}_\mbbm{i}\}$, $\{\tilde{p}_{\mbbm{i}_1}\}$, $\{\tilde{p}_\mbbm{j} \,|\, \mbbm{j}\in J_1\}$,
$\{\tilde{p}_{\mbbm{i}_2}\}$, $\{\tilde{p}_\mbbm{j} \,|\, \mbbm{j}\in J_2\}$ are pairwise disjoint.
We define an expansion ${\mf A}'$ of ${\mf A}$ 
to ${\mc L}\cup \{\tilde{p}_\mbbm{i}\}\cup \{\tilde{p}_\mbbm{j} \,|\, \mbbm{j}\in J\overset{\text{(\ref{eq1a})}}{=\!\!=}
                                                                                  \{\mbbm{i}_1\}\cup J_1\cup \{\mbbm{i}_2\}\cup J_2\}$
as follows:
\begin{equation*}
\tilde{p}_\mbbm{j}^{{\mf A}'}=\tilde{p}_\mbbm{j}^{{\mf A}_i}, \quad \mbbm{j}\in \{\mbbm{i}_i\}\cup J_i, i=1,2.
\end{equation*}
We get 
\begin{alignat*}{1}
& {\mf A}'|_{{\mc L}\cup \{\tilde{p}_\mbbm{i}\}\cup \{\tilde{p}_{\mbbm{i}_1}\}\cup \{\tilde{p}_{\mbbm{i}_2}\}}=
  {\mf A}^\#\models \mi{ClPrefix}, \\ 
& \text{for both}\ i, {\mf A}'|_{{\mc L}\cup \{\tilde{p}_{\mbbm{i}_i}\}\cup \{\tilde{p}_\mbbm{j} \,|\, \mbbm{j}\in J_i\}}=
                      {\mf A}_i\models S_i; \\
& {\mf A}'\models S, {\mf A}'|_{{\mc L}\cup \{\tilde{p}_\mbbm{i}\}}={\mf A}.
\end{alignat*}

Let ${\mf A}'$ be an interpretation 
for ${\mc L}\cup \{\tilde{p}_\mbbm{i}\}\cup \{\tilde{p}_\mbbm{j} \,|\, \mbbm{j}\in J\overset{\text{(\ref{eq1a})}}{=\!\!=}
                                                                                   \{\mbbm{i}_1\}\cup J_1\cup \{\mbbm{i}_2\}\cup J_2\}$ such that 
${\mf A}'\models S$.
We denote ${\mf A}^\#={\mf A}'|_{{\mc L}\cup \{\tilde{p}_\mbbm{i}\}\cup \{\tilde{p}_{\mbbm{i}_1}\}\cup \{\tilde{p}_{\mbbm{i}_2}\}}$.
Then ${\mf A}^\#\models \mi{ClPrefix}$, 
for both $i$,
${\mf A}'|_{{\mc L}\cup \{\tilde{p}_{\mbbm{i}_i}\}\cup \{\tilde{p}_\mbbm{j} \,|\, \mbbm{j}\in J_i\}}\models S_i$,
by the induction hypothesis (b) for $\theta_i$,
${\mf A}'|_{{\mc L}\cup \{\tilde{p}_{\mbbm{i}_i}\}\cup \{\tilde{p}_\mbbm{j} \,|\, \mbbm{j}\in J_i\}}$, 
${\mf A}'|_{{\mc L}\cup \{\tilde{p}_{\mbbm{i}_i}\}}\models \gamma_i$,
${\mf A}^\#\models \gamma_i$,
for all $e\in {\mc S}_{{\mf A}^\#}$, 
{\footnotesize
\begin{alignat*}{1}
& \|\gamma_i\|_e^{{\mf A}^\#}=1, \\[1mm]
& (\|\tilde{p}_{\mbbm{i}_i}(\bar{x})\|_e^{{\mf A}^\#}\frightarrow \|\theta_i\|_e^{{\mf A}^\#})\fwedge
  (\|\theta_i\|_e^{{\mf A}^\#}\frightarrow \|\tilde{p}_{\mbbm{i}_i}(\bar{x})\|_e^{{\mf A}^\#})=1, \\[1mm]
& (\|\tilde{p}_{\mbbm{i}_i}(\bar{x})\|_e^{{\mf A}^\#}\frightarrow \|\theta_i\|_e^{{\mf A}^\#})=1,\
  (\|\theta_i\|_e^{{\mf A}^\#}\frightarrow \|\tilde{p}_{\mbbm{i}_i}(\bar{x})\|_e^{{\mf A}^\#})=1, \\[1mm]
& \|\tilde{p}_{\mbbm{i}_i}(\bar{x})\|_e^{{\mf A}^\#}\leq \|\theta_i\|_e^{{\mf A}^\#},\
  \|\theta_i\|_e^{{\mf A}^\#}\leq \|\tilde{p}_{\mbbm{i}_i}(\bar{x})\|_e^{{\mf A}^\#}, \\[1mm]
& \|\tilde{p}_{\mbbm{i}_i}(\bar{x})\|_e^{{\mf A}^\#}=\|\theta_i\|_e^{{\mf A}^\#},  
\end{alignat*}}%
for all $C\in \mi{ClPrefix}$, ${\mf A}^\#\models_e C$,
concerning Table \ref{tab2}, for every form of $\theta$, 
{\footnotesize
\begin{alignat*}{1}
& \|\tilde{p}_\mbbm{i}(\bar{x})\|_e^{{\mf A}^\#}=
  (\|\tilde{p}_{\mbbm{i}_1}(\bar{x})\|_e^{{\mf A}^\#}\fdiamond \|\tilde{p}_{\mbbm{i}_2}(\bar{x})\|_e^{{\mf A}^\#}), \\[1mm]
& \|\tilde{p}_\mbbm{i}(\bar{x})\|_e^{{\mf A}^\#}=(\|\theta_1\|_e^{{\mf A}^\#}\fdiamond \|\theta_2\|_e^{{\mf A}^\#}), \\[1mm]
& \|\tilde{p}_\mbbm{i}(\bar{x})\leftrightarrow (\theta_1\diamond \theta_2)\|_e^{{\mf A}^\#}=1, \\[1mm]
& \|\tilde{p}_\mbbm{i}(\bar{x})\leftrightarrow \theta\|_e^{{\mf A}^\#}=1;
\end{alignat*}}%
${\mf A}^\#\models \tilde{p}_\mbbm{i}(\bar{x})\leftrightarrow \theta$,
${\mf A}'|_{{\mc L}\cup \{\tilde{p}_\mbbm{i}\}}={\mf A}^\#|_{{\mc L}\cup \{\tilde{p}_\mbbm{i}\}}\models
 \tilde{p}_\mbbm{i}(\bar{x})\leftrightarrow \theta$.
We put ${\mf A}={\mf A}'|_{{\mc L}\cup \{\tilde{p}_\mbbm{i}\}}$, an interpretation for ${\mc L}\cup \{\tilde{p}_\mbbm{i}\}$.
Then ${\mf A}\models \tilde{p}_\mbbm{i}(\bar{x})\leftrightarrow \theta$, 
${\mf A}={\mf A}'|_{{\mc L}\cup \{\tilde{p}_\mbbm{i}\}}$; (b) holds.

Case 2.2 (the unary interpolation case):
Either $\theta=\theta_1\rightarrow \gz$ or $\theta=\theta_1\geql \gz$ or $\theta=\theta_1\geql \gu$ or 
$\theta=\gz\gle \theta_1$ or $\theta_1\gle \gu$ or $\theta=Q x\, \theta_1$, $\theta_1\in \mi{Form}_{\mc L}-\{\gz,\gu\}$.
We have that (c,d) of Lemma~\ref{le111} hold for $\theta$, $\mi{vars}(\theta)\subseteq \mi{vars}(\bar{x})$.
Then (c,d) of Lemma \ref{le111} hold for $\theta_1$,
$\mi{vars}(\theta_1)\subseteq \mi{vars}(\theta)\subseteq \mi{vars}(\bar{x})$,
$x\in \mi{vars}(\theta)\subseteq \mi{vars}(\bar{x})$. 
We put $j_{\mbbm{i}_1}=j_\mbbm{i}+1$ and $\mbbm{i}_1=(n_\theta,j_{\mbbm{i}_1})\in \{(n_\theta,j) \,|\, j\in \mbb{N}\}\subseteq \mbb{I}$. 
$\tilde{p}_{\mbbm{i}_1}\in \tilde{\mbb{P}}$.
We put $\mi{ar}(\tilde{p}_{\mbbm{i}_1})=|\bar{x}|$.
We get by the induction hypothesis for $\theta_1$, $j_{\mbbm{i}_1}$, $\mbbm{i}_1$, $\tilde{p}_{\mbbm{i}_1}$ that
there exist $n_{J_1}\geq j_{\mbbm{i}_1}$,
$J_1=\{(n_\theta,j) \,|\, j_{\mbbm{i}_1}+1\leq j\leq n_{J_1}\}\subseteq \{(n_\theta,j) \,|\, j\in \mbb{N}\}\subseteq \mbb{I}$, 
$\mbbm{i}_1\not\in J_1$,
$S_1\subseteq_{\mc F} \mi{OrdCl}_{{\mc L}\cup \{\tilde{p}_{\mbbm{i}_1}\}\cup \{\tilde{p}_\mbbm{j} \,|\, \mbbm{j}\in J_1\}}$, and
(a--e) hold for $\theta_1$, $\tilde{p}_{\mbbm{i}_1}$, $J_1$, $S_1$.
We put $n_J=n_{J_1}$ and $J=\{(n_\theta,j) \,|\, j_\mbbm{i}+1\leq j\leq n_J\}\subseteq \{(n_\theta,j) \,|\, j\in \mbb{N}\}\subseteq \mbb{I}$.
Then $j_\mbbm{i}<j_{\mbbm{i}_1}\leq n_J$, $\mbbm{i}\not\in J$,
\begin{alignat}{1}
\label{eq1f}
& J=\{\mbbm{i}_1\}\cup J_1, \\[0mm]
\label{eq1g}
& \{\mbbm{i}\}, \{\mbbm{i}_1\}, J_1\ \text{are pairwise disjoint}.
\end{alignat}   
\begin{table*}[p]
\caption{Unary interpolation rules for $\rightarrow$, $\geql$, $\gle$, $\forall$, $\exists$}\label{tab3}
\vspace{-6mm}
\centering
\begin{minipage}[t]{\linewidth-65mm}
\footnotesize
\begin{IEEEeqnarray}{*LL}
\hline \hline \notag \\[0mm]
\notag 
\text{\bf Case} & \\[1mm]
\hline \notag \\[2mm]
\label{eq0rr4+}
\mbi{\theta=\theta_1\rightarrow \gz} \qquad & 
\dfrac{\tilde{p}_\mbbm{i}(\bar{x})\leftrightarrow (\theta_1\rightarrow \gz)}
      {\{\tilde{p}_{\mbbm{i}_1}(\bar{x})\geql \gz\vee \tilde{p}_\mbbm{i}(\bar{x})\geql \gz,
         \gz\gle \tilde{p}_{\mbbm{i}_1}(\bar{x})\vee \tilde{p}_\mbbm{i}(\bar{x})\geql \gu, 
         \tilde{p}_{\mbbm{i}_1}(\bar{x})\leftrightarrow \theta_1\}} \\[2mm]
\IEEEeqnarraymulticol{2}{l}{
|\text{Consequent}|=
12+4\cdot |\bar{x}|+|\tilde{p}_{\mbbm{i}_1}(\bar{x})\leftrightarrow \theta_1|\leq
27\cdot (1+|\bar{x}|)+|\tilde{p}_{\mbbm{i}_1}(\bar{x})\leftrightarrow \theta_1|} \notag \\[6mm]
\label{eq0rr77+}
\mbi{\theta=\theta_1\geql \gz} & 
\dfrac{\tilde{p}_\mbbm{i}(\bar{x})\leftrightarrow (\theta_1\geql \gz)}
      {\{\tilde{p}_{\mbbm{i}_1}(\bar{x})\geql \gz\vee \tilde{p}_\mbbm{i}(\bar{x})\geql \gz, 
         \gz\gle \tilde{p}_{\mbbm{i}_1}(\bar{x})\vee \tilde{p}_\mbbm{i}(\bar{x})\geql \gu, 
         \tilde{p}_{\mbbm{i}_1}(\bar{x})\leftrightarrow \theta_1\}} \\[2mm]
\IEEEeqnarraymulticol{2}{l}{
|\text{Consequent}|=
12+4\cdot |\bar{x}|+|\tilde{p}_{\mbbm{i}_1}(\bar{x})\leftrightarrow \theta_1|\leq
27\cdot (1+|\bar{x}|)+|\tilde{p}_{\mbbm{i}_1}(\bar{x})\leftrightarrow \theta_1|} \notag \\[6mm]
\label{eq0rr777+}
\mbi{\theta=\theta_1\geql \gu} & 
\dfrac{\tilde{p}_\mbbm{i}(\bar{x})\leftrightarrow (\theta_1\geql \gu)}
      {\{\tilde{p}_{\mbbm{i}_1}(\bar{x})\geql \gu\vee \tilde{p}_\mbbm{i}(\bar{x})\geql \gz, 
         \tilde{p}_{\mbbm{i}_1}(\bar{x})\gle \gu\vee \tilde{p}_\mbbm{i}(\bar{x})\geql \gu, 
         \tilde{p}_{\mbbm{i}_1}(\bar{x})\leftrightarrow \theta_1\}} \\[2mm]
\IEEEeqnarraymulticol{2}{l}{
|\text{Consequent}|=
12+4\cdot |\bar{x}|+|\tilde{p}_{\mbbm{i}_1}(\bar{x})\leftrightarrow \theta_1|\leq
27\cdot (1+|\bar{x}|)+|\tilde{p}_{\mbbm{i}_1}(\bar{x})\leftrightarrow \theta_1|} \notag \\[6mm]
\label{eq0rr88+}
\mbi{\theta=\gz\gle \theta_1} & 
\dfrac{\tilde{p}_\mbbm{i}(\bar{x})\leftrightarrow (\gz\gle \theta_1)}
      {\{\gz\gle \tilde{p}_{\mbbm{i}_1}(\bar{x})\vee \tilde{p}_\mbbm{i}(\bar{x})\geql \gz, 
         \tilde{p}_{\mbbm{i}_1}(\bar{x})\geql \gz\vee \tilde{p}_\mbbm{i}(\bar{x})\geql \gu,                              
         \tilde{p}_{\mbbm{i}_1}(\bar{x})\leftrightarrow \theta_1\}} \\[2mm]
\IEEEeqnarraymulticol{2}{l}{
|\text{Consequent}|=
12+4\cdot |\bar{x}|+|\tilde{p}_{\mbbm{i}_1}(\bar{x})\leftrightarrow \theta_1|\leq
27\cdot (1+|\bar{x}|)+|\tilde{p}_{\mbbm{i}_1}(\bar{x})\leftrightarrow \theta_1|} \notag \\[6mm] 
\label{eq0rr888+}
\mbi{\theta=\theta_1\gle \gu} & 
\dfrac{\tilde{p}_\mbbm{i}(\bar{x})\leftrightarrow (\theta_1\gle \gu)}
      {\{\tilde{p}_{\mbbm{i}_1}(\bar{x})\gle \gu\vee \tilde{p}_\mbbm{i}(\bar{x})\geql \gz, 
         \tilde{p}_{\mbbm{i}_1}(\bar{x})\geql \gu\vee \tilde{p}_\mbbm{i}(\bar{x})\geql \gu, 
         \tilde{p}_{\mbbm{i}_1}(\bar{x})\leftrightarrow \theta_1\}} \\[2mm]
\IEEEeqnarraymulticol{2}{l}{
|\text{Consequent}|=
12+4\cdot |\bar{x}|+|\tilde{p}_{\mbbm{i}_1}(\bar{x})\leftrightarrow \theta_1|\leq
27\cdot (1+|\bar{x}|)+|\tilde{p}_{\mbbm{i}_1}(\bar{x})\leftrightarrow \theta_1|} \notag \\[6mm]
\label{eq0rr5+}
\mbi{\theta=\forall x\, \theta_1} & 
\dfrac{\tilde{p}_\mbbm{i}(\bar{x})\leftrightarrow \forall x\, \theta_1}
      {\{\tilde{p}_\mbbm{i}(\bar{x})\geql \forall x\, \tilde{p}_{\mbbm{i}_1}(\bar{x}),
         \tilde{p}_{\mbbm{i}_1}(\bar{x})\leftrightarrow \theta_1\}} \\[2mm]
\IEEEeqnarraymulticol{2}{l}{
|\text{Consequent}|=
5+2\cdot |\bar{x}|+|\tilde{p}_{\mbbm{i}_1}(\bar{x})\leftrightarrow \theta_1|\leq
27\cdot (1+|\bar{x}|)+|\tilde{p}_{\mbbm{i}_1}(\bar{x})\leftrightarrow \theta_1|} \notag \\[6mm]
\label{eq0rr6+}
\mbi{\theta=\exists x\, \theta_1} & 
\dfrac{\tilde{p}_\mbbm{i}(\bar{x})\leftrightarrow \exists x\, \theta_1}
      {\{\tilde{p}_\mbbm{i}(\bar{x})\geql \exists x\, \tilde{p}_{\mbbm{i}_1}(\bar{x}),
         \tilde{p}_{\mbbm{i}_1}(\bar{x})\leftrightarrow \theta_1\}} \\[2mm]
\IEEEeqnarraymulticol{2}{l}{
|\text{Consequent}|=
5+2\cdot |\bar{x}|+|\tilde{p}_{\mbbm{i}_1}(\bar{x})\leftrightarrow \theta_1|\leq
27\cdot (1+|\bar{x}|)+|\tilde{p}_{\mbbm{i}_1}(\bar{x})\leftrightarrow \theta_1|} \notag \\[2mm]
\hline \hline \notag
\end{IEEEeqnarray}
\vspace{-4mm} \\
$\text{Consequent}$ denotes the consequent ("denominator") part of an interpolation rule. 
\end{minipage}
\vspace{-2mm}
\end{table*}
In Table~\ref{tab3}, for every form of $\theta$, a unary interpolation rule of the respective form
\begin{alignat*}{1}
& \dfrac{\tilde{p}_\mbbm{i}(\bar{x})\leftrightarrow \theta\in \mi{Form}_{{\mc L}\cup \{\tilde{p}_\mbbm{i}\}}}
        {\mi{ClPrefix}\cup \{\gamma_1\}\subseteq_{\mc F} 
         \mi{Form}_{{\mc L}\cup \{\tilde{p}_\mbbm{i}\}\cup \{\tilde{p}_{\mbbm{i}_1}\}}}, \\[0mm]
& \mi{ClPrefix}\subseteq_{\mc F} \mi{OrdCl}_{{\mc L}\cup \{\tilde{p}_\mbbm{i}\}\cup \{\tilde{p}_{\mbbm{i}_1}\}}, \\
& \gamma_1=\tilde{p}_{\mbbm{i}_1}(\bar{x})\leftrightarrow \theta_1\in \mi{Form}_{{\mc L}\cup \{\tilde{p}_{\mbbm{i}_1}\}},
\end{alignat*}
is assigned.
We put
\begin{equation*} 
S=\mi{ClPrefix}\cup S_1\subseteq_{\mc F}
  \mi{OrdCl}_{{\mc L}\cup \{\tilde{p}_\mbbm{i}\}\cup 
              \{\tilde{p}_\mbbm{j} \,|\, \mbbm{j}\in \{\mbbm{i}_1\}\cup J_1\overset{\text{(\ref{eq1f})}}{=\!\!=} J\}}.
\end{equation*}
It can be proved that
\begin{equation}
\label{eq1h}
\mi{ClPrefix}\cap S_1=\emptyset.
\end{equation}

(a) and (c--e) can be proved straightforwardly.

Let ${\mf A}$ be an interpretation for ${\mc L}\cup \{\tilde{p}_\mbbm{i}\}$. 
We define an expansion ${\mf A}^\#$ of ${\mf A}$ to ${\mc L}\cup \{\tilde{p}_\mbbm{i}\}\cup \{\tilde{p}_{\mbbm{i}_1}\}$ as follows:
\begin{equation*}
\tilde{p}_{\mbbm{i}_1}^{{\mf A}^\#}(u_1,\dots,u_{|\bar{x}|})=
\left\{\begin{array}{ll}
       \|\theta_1\|_e^{\mf A} &\ \text{\it if there exists}\ e\in {\mc S}_{\mf A}\ \text{\it such that} \\ 
                              &\ \quad \|\bar{x}\|_e^{\mf A}=u_1,\dots,u_{|\bar{x}|}, \\[1mm]
       0                      &\ \text{\it else}.
       \end{array}
\right. 
\end{equation*}
Then, for all $e\in {\mc S}_{{\mf A}^\#}$, 
\begin{alignat*}{1}
& \|\tilde{p}_{\mbbm{i}_1}(\bar{x})\|_e^{{\mf A}^\#}\!\!=\tilde{p}_{\mbbm{i}_1}^{{\mf A}^\#}(\|\bar{x}\|_e^{{\mf A}^\#})=
  \tilde{p}_{\mbbm{i}_1}^{{\mf A}^\#}(\|\bar{x}\|_e^{\mf A})=\|\theta_1\|_e^{\mf A}=\|\theta_1\|_e^{{\mf A}^\#}\!\!, \\[1mm]
& \|\gamma_1\|_e^{{\mf A}^\#}\!\!=
  (\|\tilde{p}_{\mbbm{i}_1}(\bar{x})\|_e^{{\mf A}^\#}\!\!\frightarrow \|\theta_1\|_e^{{\mf A}^\#})\fwedge
  (\|\theta_1\|_e^{{\mf A}^\#}\!\!\frightarrow \|\tilde{p}_{\mbbm{i}_1}(\bar{x})\|_e^{{\mf A}^\#})= \\
& \phantom{\|\gamma_1\|_e^{{\mf A}^\#}\!\!=\mbox{}}
  \|\theta_1\|_e^{{\mf A}^\#}\frightarrow \|\theta_1\|_e^{{\mf A}^\#}=1;
\end{alignat*}
${\mf A}^\#\models \gamma_1$.

Let ${\mf A}\models \tilde{p}_\mbbm{i}(\bar{x})\leftrightarrow \theta$.
Then ${\mf A}^\#\models \tilde{p}_\mbbm{i}(\bar{x})\leftrightarrow \theta$,
for all $e\in {\mc S}_{{\mf A}^\#}$,
{\footnotesize
\begin{alignat*}{1}
& \|\tilde{p}_\mbbm{i}(\bar{x})\leftrightarrow \theta\|_e^{{\mf A}^\#}=1, \\[2mm]
& \|\tilde{p}_\mbbm{i}(\bar{x})\leftrightarrow (\theta_1\diamond \gz)\|_e^{{\mf A}^\#}=1, \\[1mm]
& \|\tilde{p}_\mbbm{i}(\bar{x})\leftrightarrow (\theta_1\diamond \gu)\|_e^{{\mf A}^\#}=1, \\[1mm]
& \|\tilde{p}_\mbbm{i}(\bar{x})\leftrightarrow (\gz\gle \theta_1)\|_e^{{\mf A}^\#}=1, \\[1mm]
& \|\tilde{p}_\mbbm{i}(\bar{x})\leftrightarrow Q x\, \theta_1\|_e^{{\mf A}^\#}=1, \\[2mm]
& (\|\tilde{p}_\mbbm{i}(\bar{x})\|_e^{{\mf A}^\#}\frightarrow (\|\theta_1\|_e^{{\mf A}^\#}\fdiamond 0))\fwedge
  ((\|\theta_1\|_e^{{\mf A}^\#}\fdiamond 0)\frightarrow \|\tilde{p}_\mbbm{i}(\bar{x})\|_e^{{\mf A}^\#})=1, \\[1mm]
& (\|\tilde{p}_\mbbm{i}(\bar{x})\|_e^{{\mf A}^\#}\frightarrow (\|\theta_1\|_e^{{\mf A}^\#}\fdiamond 1))\fwedge
  ((\|\theta_1\|_e^{{\mf A}^\#}\fdiamond 1)\frightarrow \|\tilde{p}_\mbbm{i}(\bar{x})\|_e^{{\mf A}^\#})=1, \\[1mm]
& (\|\tilde{p}_\mbbm{i}(\bar{x})\|_e^{{\mf A}^\#}\frightarrow (0\fle \|\theta_1\|_e^{{\mf A}^\#}))\fwedge
  ((0\fle \|\theta_1\|_e^{{\mf A}^\#})\frightarrow \|\tilde{p}_\mbbm{i}(\bar{x})\|_e^{{\mf A}^\#})=1, \\[1mm]
& (\|\tilde{p}_\mbbm{i}(\bar{x})\|_e^{{\mf A}^\#}\frightarrow \bigflimit_{u\in {\mc U}_{{\mf A}^\#}} \|\theta_1\|_{e[x/u]}^{{\mf A}^\#})\fwedge \\
& \quad
  (\bigflimit_{u\in {\mc U}_{{\mf A}^\#}} \|\theta_1\|_{e[x/u]}^{{\mf A}^\#}\frightarrow \|\tilde{p}_\mbbm{i}(\bar{x})\|_e^{{\mf A}^\#})=1, \\[2mm]
& \|\tilde{p}_\mbbm{i}(\bar{x})\|_e^{{\mf A}^\#}\leq (\|\theta_1\|_e^{{\mf A}^\#}\fdiamond 0),\
  (\|\theta_1\|_e^{{\mf A}^\#}\fdiamond 0)\leq \|\tilde{p}_\mbbm{i}(\bar{x})\|_e^{{\mf A}^\#}, \\[1mm]
& \|\tilde{p}_\mbbm{i}(\bar{x})\|_e^{{\mf A}^\#}\leq (\|\theta_1\|_e^{{\mf A}^\#}\fdiamond 1),\
  (\|\theta_1\|_e^{{\mf A}^\#}\fdiamond 1)\leq \|\tilde{p}_\mbbm{i}(\bar{x})\|_e^{{\mf A}^\#}, \\[1mm]
& \|\tilde{p}_\mbbm{i}(\bar{x})\|_e^{{\mf A}^\#}\leq (0\fle \|\theta_1\|_e^{{\mf A}^\#}),\
  (0\fle \|\theta_1\|_e^{{\mf A}^\#})\leq \|\tilde{p}_\mbbm{i}(\bar{x})\|_e^{{\mf A}^\#}, \\[1mm]
& \|\tilde{p}_\mbbm{i}(\bar{x})\|_e^{{\mf A}^\#}\leq \bigflimit_{u\in {\mc U}_{{\mf A}^\#}} \|\theta_1\|_{e[x/u]}^{{\mf A}^\#},\
  \bigflimit_{u\in {\mc U}_{{\mf A}^\#}} \|\theta_1\|_{e[x/u]}^{{\mf A}^\#}\leq \|\tilde{p}_\mbbm{i}(\bar{x})\|_e^{{\mf A}^\#}, \\[2mm]
& \|\tilde{p}_\mbbm{i}(\bar{x})\|_e^{{\mf A}^\#}=(\|\tilde{p}_{\mbbm{i}_1}(\bar{x})\|_e^{{\mf A}^\#}\fdiamond 0), \\[1mm]
& \|\tilde{p}_\mbbm{i}(\bar{x})\|_e^{{\mf A}^\#}=(\|\tilde{p}_{\mbbm{i}_1}(\bar{x})\|_e^{{\mf A}^\#}\fdiamond 1), \\[1mm]
& \|\tilde{p}_\mbbm{i}(\bar{x})\|_e^{{\mf A}^\#}=(0\fle \|\tilde{p}_{\mbbm{i}_1}(\bar{x})\|_e^{{\mf A}^\#}), \\[1mm]
& \|\tilde{p}_\mbbm{i}(\bar{x})\|_e^{{\mf A}^\#}=\bigflimit_{u\in {\mc U}_{{\mf A}^\#}} \|\tilde{p}_{\mbbm{i}_1}(\bar{x})\|_{e[x/u]}^{{\mf A}^\#},
\quad \bigflimit\in \Big\{\bigfwedge,\bigfvee\Big\},
\end{alignat*}}%
concerning Table \ref{tab3}, for every form of $\theta$, for all $C\in \mi{ClPrefix}$, ${\mf A}^\#\models_e C$;
${\mf A}^\#\models \mi{ClPrefix}$, 
${\mf A}^\#|_{{\mc L}\cup \{\tilde{p}_{\mbbm{i}_1}\}}\models \gamma_1$,
by the induction hypothesis (b) for $\theta_1$, ${\mf A}^\#|_{{\mc L}\cup \{\tilde{p}_{\mbbm{i}_1}\}}$, 
there exists an interpretation ${\mf A}_1$ for ${\mc L}\cup \{\tilde{p}_{\mbbm{i}_1}\}\cup \{\tilde{p}_\mbbm{j} \,|\, \mbbm{j}\in J_1\}$, and 
${\mf A}_1\models S_1$, ${\mf A}_1|_{{\mc L}\cup \{\tilde{p}_{\mbbm{i}_1}\}}={\mf A}^\#|_{{\mc L}\cup \{\tilde{p}_{\mbbm{i}_1}\}}$.
By (\ref{eq1g}), $\{\tilde{p}_\mbbm{i}\}$, $\{\tilde{p}_{\mbbm{i}_1}\}$, $\{\tilde{p}_\mbbm{j} \,|\, \mbbm{j}\in J_1\}$ are pairwise disjoint.
We define an expansion ${\mf A}'$ of ${\mf A}$ 
to ${\mc L}\cup \{\tilde{p}_\mbbm{i}\}\cup \{\tilde{p}_\mbbm{j} \,|\, \mbbm{j}\in J\overset{\text{(\ref{eq1f})}}{=\!\!=} 
                                                                                  \{\mbbm{i}_1\}\cup J_1\}$
as follows:
\begin{equation*}
\tilde{p}_\mbbm{j}^{{\mf A}'}=\tilde{p}_\mbbm{j}^{{\mf A}_1}, \quad \mbbm{j}\in \{\mbbm{i}_1\}\cup J_1.
\end{equation*}
We get
\begin{alignat*}{1}
& {\mf A}'|_{{\mc L}\cup \{\tilde{p}_\mbbm{i}\}\cup \{\tilde{p}_{\mbbm{i}_1}\}}={\mf A}^\#\models \mi{ClPrefix}, \\ 
& {\mf A}'|_{{\mc L}\cup \{\tilde{p}_{\mbbm{i}_1}\}\cup \{\tilde{p}_\mbbm{j} \,|\, \mbbm{j}\in J_1\}}={\mf A}_1\models S_1; \\
& {\mf A}'\models S, {\mf A}'|_{{\mc L}\cup \{\tilde{p}_\mbbm{i}\}}={\mf A}.
\end{alignat*}

Let ${\mf A}'$ be an interpretation 
for ${\mc L}\cup \{\tilde{p}_\mbbm{i}\}\cup \{\tilde{p}_\mbbm{j} \,|\, \mbbm{j}\in J\overset{\text{(\ref{eq1f})}}{=\!\!=}
                                                                                   \{\mbbm{i}_1\}\cup J_1\}$ such that 
${\mf A}'\models S$.
We put ${\mf A}^\#={\mf A}'|_{{\mc L}\cup \{\tilde{p}_\mbbm{i}\}\cup \{\tilde{p}_{\mbbm{i}_1}\}}$, 
an interpretation for ${\mc L}\cup \{\tilde{p}_\mbbm{i}\}\cup \{\tilde{p}_{\mbbm{i}_1}\}$.
Then ${\mf A}^\#\models \mi{ClPrefix}$, 
${\mf A}'|_{{\mc L}\cup \{\tilde{p}_{\mbbm{i}_1}\}\cup \{\tilde{p}_\mbbm{j} \,|\, \mbbm{j}\in J_1\}}\models S_1$,
by the induction hypothesis (b) for $\theta_1$,
${\mf A}'|_{{\mc L}\cup \{\tilde{p}_{\mbbm{i}_1}\}\cup \{\tilde{p}_\mbbm{j} \,|\, \mbbm{j}\in J_1\}}$, 
${\mf A}'|_{{\mc L}\cup \{\tilde{p}_{\mbbm{i}_1}\}}\models \gamma_1$,
${\mf A}^\#\models \gamma_1$,
for all $e\in {\mc S}_{{\mf A}^\#}$, 
{\footnotesize
\begin{alignat*}{1}
& \|\gamma_1\|_e^{{\mf A}^\#}=1, \\[1mm]
& (\|\tilde{p}_{\mbbm{i}_1}(\bar{x})\|_e^{{\mf A}^\#}\frightarrow \|\theta_1\|_e^{{\mf A}^\#})\fwedge
  (\|\theta_1\|_e^{{\mf A}^\#}\frightarrow \|\tilde{p}_{\mbbm{i}_1}(\bar{x})\|_e^{{\mf A}^\#})=1, \\[1mm]
& (\|\tilde{p}_{\mbbm{i}_1}(\bar{x})\|_e^{{\mf A}^\#}\frightarrow \|\theta_1\|_e^{{\mf A}^\#})=1,\
  (\|\theta_1\|_e^{{\mf A}^\#}\frightarrow \|\tilde{p}_{\mbbm{i}_1}(\bar{x})\|_e^{{\mf A}^\#})=1, \\[1mm]
& \|\tilde{p}_{\mbbm{i}_1}(\bar{x})\|_e^{{\mf A}^\#}\leq \|\theta_1\|_e^{{\mf A}^\#},\
  \|\theta_1\|_e^{{\mf A}^\#}\leq \|\tilde{p}_{\mbbm{i}_1}(\bar{x})\|_e^{{\mf A}^\#}, \\[1mm]
& \|\tilde{p}_{\mbbm{i}_1}(\bar{x})\|_e^{{\mf A}^\#}=\|\theta_1\|_e^{{\mf A}^\#};  
\end{alignat*}}%
for all $e\in {\mc S}_{{\mf A}^\#}$,
for all $C\in \mi{ClPrefix}$, ${\mf A}^\#\models_e C$,
concerning Table \ref{tab3}, for every form of $\theta$, 
{\footnotesize
\begin{alignat*}{1}
& \|\tilde{p}_\mbbm{i}(\bar{x})\|_e^{{\mf A}^\#}=(\|\tilde{p}_{\mbbm{i}_1}(\bar{x})\|_e^{{\mf A}^\#}\fdiamond 0),\
  \|\tilde{p}_\mbbm{i}(\bar{x})\|_e^{{\mf A}^\#}=(\|\tilde{p}_{\mbbm{i}_1}(\bar{x})\|_e^{{\mf A}^\#}\fdiamond 1), \\[1mm]
& \|\tilde{p}_\mbbm{i}(\bar{x})\|_e^{{\mf A}^\#}=(0\fle \|\tilde{p}_{\mbbm{i}_1}(\bar{x})\|_e^{{\mf A}^\#}),\
  \|\tilde{p}_\mbbm{i}(\bar{x})\|_e^{{\mf A}^\#}=\bigflimit_{u\in {\mc U}_{{\mf A}^\#}} \|\tilde{p}_{\mbbm{i}_1}(\bar{x})\|_{e[x/u]}^{{\mf A}^\#}, \\[1mm]
& \|\tilde{p}_\mbbm{i}(\bar{x})\|_e^{{\mf A}^\#}=(\|\theta_1\|_e^{{\mf A}^\#}\fdiamond 0),\
  \|\tilde{p}_\mbbm{i}(\bar{x})\|_e^{{\mf A}^\#}=(\|\theta_1\|_e^{{\mf A}^\#}\fdiamond 1), \\[1mm]
& \|\tilde{p}_\mbbm{i}(\bar{x})\|_e^{{\mf A}^\#}=(0\fle \|\theta_1\|_e^{{\mf A}^\#}),\
  \|\tilde{p}_\mbbm{i}(\bar{x})\|_e^{{\mf A}^\#}=\bigflimit_{u\in {\mc U}_{{\mf A}^\#}} \|\theta_1\|_{e[x/u]}^{{\mf A}^\#}, \\[1mm]
& \|\tilde{p}_\mbbm{i}(\bar{x})\leftrightarrow (\theta_1\diamond \gz)\|_e^{{\mf A}^\#}=1,\
  \|\tilde{p}_\mbbm{i}(\bar{x})\leftrightarrow (\theta_1\diamond \gu)\|_e^{{\mf A}^\#}=1, \\[1mm]
& \|\tilde{p}_\mbbm{i}(\bar{x})\leftrightarrow (\gz\gle \theta_1)\|_e^{{\mf A}^\#}=1,\
  \|\tilde{p}_\mbbm{i}(\bar{x})\leftrightarrow Q x\, \theta_1\|_e^{{\mf A}^\#}=1, \\[1mm]
& \|\tilde{p}_\mbbm{i}(\bar{x})\leftrightarrow \theta\|_e^{{\mf A}^\#}=1,
\quad \bigflimit\in \Big\{\bigfwedge,\bigfvee\Big\};
\end{alignat*}}%
${\mf A}^\#\models \tilde{p}_\mbbm{i}(\bar{x})\leftrightarrow \theta$,
${\mf A}'|_{{\mc L}\cup \{\tilde{p}_\mbbm{i}\}}={\mf A}^\#|_{{\mc L}\cup \{\tilde{p}_\mbbm{i}\}}\models
 \tilde{p}_\mbbm{i}(\bar{x})\leftrightarrow \theta$.
We put ${\mf A}={\mf A}'|_{{\mc L}\cup \{\tilde{p}_\mbbm{i}\}}$, an interpretation for ${\mc L}\cup \{\tilde{p}_\mbbm{i}\}$.
Then ${\mf A}\models \tilde{p}_\mbbm{i}(\bar{x})\leftrightarrow \theta$, 
${\mf A}={\mf A}'|_{{\mc L}\cup \{\tilde{p}_\mbbm{i}\}}$; (b) holds.

So, in all Cases 1, 2.1, 2.2, (a--e) hold.
The induction is completed.
\qed
\end{proof}

\begin{lemma}
\label{le1}
Let $n_\phi\in \mbb{N}$ and $\phi\in \mi{Form}_{\mc L}$.
There exist either $J_\phi=\emptyset$, or $n_{J_\phi}$, $J_\phi=\{(n_\phi,j) \,|\, j\leq n_{J_\phi}\}$, 
$J_\phi\subseteq \{(n_\phi,j) \,|\, j\in \mbb{N}\}\subseteq \mbb{I}$, and
$S_\phi\subseteq_{\mc F} \mi{OrdCl}_{{\mc L}\cup \{\tilde{p}_\mbbm{j} \,|\, \mbbm{j}\in J_\phi\}}$ such that 
\begin{enumerate}[\rm (a)]
\item
$\|J_\phi\|\leq 2\cdot |\phi|$; 
\item
either $J_\phi=\emptyset$, $S_\phi=\{\square\}$, or $J_\phi=S_\phi=\emptyset$, or $J_\phi\neq \emptyset$, $\square\not\in S_\phi\neq \emptyset$; 
\item
there exists an interpretation ${\mf A}$ for ${\mc L}$ and ${\mf A}\models \phi$                 
if and only if there exists an interpretation ${\mf A}'$ for ${\mc L}\cup \{\tilde{p}_\mbbm{j} \,|\, \mbbm{j}\in J_\phi\}$ and
${\mf A}'\models S_\phi$, satisfying ${\mf A}={\mf A}'|_{\mc L}$; 
\item
$|S_\phi|\in O(|\phi|^2)$; 
the number of all elementary operations of the translation of $\phi$ to $S_\phi$ is in $O(|\phi|^2)$;
the time complexity of the translation of $\phi$ to $S_\phi$ is in $O(|\phi|^2\cdot (\log (1+n_\phi)+\log |\phi|))$;
\item
if $S_\phi\neq \emptyset, \{\square\}$, then $J_\phi\neq \emptyset$,
for all $C\in S_\phi$, $\emptyset\neq \mi{preds}(C)\cap \tilde{\mbb{P}}\subseteq \{\tilde{p}_\mbbm{j} \,|\, \mbbm{j}\in J_\phi\}$; 
\item
$\mi{tcons}(S_\phi)-\{\gz,\gu\}\subseteq \mi{tcons}(\phi)-\{\gz,\gu\}$.
\end{enumerate}
\end{lemma}

\begin{proof}
By Lemma \ref{le111} for $n_\phi$, $\phi$, 
there exists $\phi'\in \mi{Form}_{\mc L}$ such that (a--e) of Lemma \ref{le111} hold for $n_\phi$, $\phi$, $\phi'$.
We distinguish three cases for $\phi'$.

Case 1:
$\phi'\in \overline{C}_{\mc L}-\{\gu\}$.
We put $J_\phi=\emptyset\subseteq \{(n_\phi,j) \,|\, j\in \mbb{N}\}\subseteq \mbb{I}$ and
$S_\phi=\{\square\}\subseteq_{\mc F} \mi{OrdCl}_{\mc L}$.

(a), (b), and (d--f) can be proved straightforwardly.

For every interpretation ${\mf A}$ for ${\mc L}$,
${\mf A}\not\models \phi'\overset{\text{Lemma \ref{le111}(a)}}{\eqvl\!\!\eqvl\!\!\eqvl\!\!\eqvl\!\!\eqvl\!\!\eqvl\!\!\eqvl\!\!\eqvl\!\!\eqvl} 
                    \phi$,
${\mf A}\not\models S_\phi$;
trivially, 
there exists an interpretation ${\mf A}$ for ${\mc L}$ and ${\mf A}\models \phi$ if and only if 
there exists an interpretation ${\mf A}'$ for ${\mc L}$ and ${\mf A}'\models S_\phi$, 
satisfying ${\mf A}={\mf A}'|_{\mc L}$; (c) holds.

Case 2:
$\phi'=\gu$.
We put $J_\phi=\emptyset\subseteq \{(n_\phi,j) \,|\, j\in \mbb{N}\}\subseteq \mbb{I}$ and
$S_\phi=\emptyset\subseteq_{\mc F} \mi{OrdCl}_{\mc L}$.

(a), (b), and (d--f) can be proved straightforwardly.

For every interpretation ${\mf A}$ for ${\mc L}$,
${\mf A}\models \phi'\overset{\text{Lemma \ref{le111}(a)}}{\eqvl\!\!\eqvl\!\!\eqvl\!\!\eqvl\!\!\eqvl\!\!\eqvl\!\!\eqvl\!\!\eqvl\!\!\eqvl} \phi$,
${\mf A}\models S_\phi$;
there exists an interpretation ${\mf A}$ for ${\mc L}$ and ${\mf A}\models \phi$ if and only if 
there exists an interpretation ${\mf A}'$ for ${\mc L}$ and ${\mf A}'\models S_\phi$, 
satisfying ${\mf A}={\mf A}'={\mf A}'|_{\mc L}$; (c) holds.

Case 3: 
$\phi'\not\in \overline{C}_{\mc L}$.  
We have that $\phi'\in \mi{Form}_{\mc L}$, (c,d) of Lemma \ref{le111} hold for $\phi'$.
We put $\bar{x}=\mi{varseq}(\phi')$.
Then $\phi'\in \mi{Form}_{\mc L}-\overline{C}_{\mc L}\subseteq \mi{Form}_{\mc L}-\{\gz,\gu\}$, 
$\mi{vars}(\phi')=\mi{vars}(\bar{x})$, $|\bar{x}|\leq |\phi'|$.
We put $j_\mbbm{i}=0$ and $\mbbm{i}=(n_\phi,j_\mbbm{i})\in \{(n_\phi,j) \,|\, j\in \mbb{N}\}\subseteq \mbb{I}$.
$\tilde{p}_\mbbm{i}\in \tilde{\mbb{P}}$.
We put $\mi{ar}(\tilde{p}_\mbbm{i})=|\bar{x}|$.
We get by Lemma~\ref{le11} for $n_\phi$, $\phi'$ that
there exist $n_J\geq j_\mbbm{i}$,
$J=\{(n_\phi,j) \,|\, 1\leq j\leq n_J\}\subseteq \{(n_\phi,j) \,|\, j\in \mbb{N}\}\subseteq \mbb{I}$, $\mbbm{i}\not\in J$,
$S\subseteq_{\mc F} \mi{OrdCl}_{{\mc L}\cup \{\tilde{p}_\mbbm{i}\}\cup \{\tilde{p}_\mbbm{j} \,|\, \mbbm{j}\in J\}}$, and
(a--e) of Lemma \ref{le11} hold for $\phi'$.
We put $n_{J_\phi}=n_J$ and $J_\phi=\{(n_\phi,j) \,|\, j\leq n_{J_\phi}\}\subseteq \{(n_\phi,j) \,|\, j\in \mbb{N}\}\subseteq \mbb{I}$.
Then 
\begin{alignat}{1}
\label{eq1k}
& J_\phi=\{(n_\phi,j_\mbbm{i})\}\cup \{(n_\phi,j) \,|\, 1\leq j\leq n_J\}=\{\mbbm{i}\}\cup J, \\[1mm]
\label{eq1l}
& \{\mbbm{i}\}\cap J=\emptyset.
\end{alignat}
We put $S_\phi=\{\tilde{p}_\mbbm{i}(\bar{x})\geql \gu\}\cup S\subseteq_{\mc F} 
               \mi{OrdCl}_{{\mc L}\cup \{\tilde{p}_\mbbm{j} \,|\, \mbbm{j}\in \{\mbbm{i}\}\cup J\overset{\text{(\ref{eq1k})}}{=\!\!=} J_\phi\}}$.
\begin{equation}
\label{eq1m}
\{\tilde{p}_\mbbm{i}(\bar{x})\geql \gu\}\cap S\overset{\text{Lemma \ref{le11}(d)}}{=\!\!=\!\!=\!\!=\!\!=\!\!=\!\!=\!\!=\!\!=} \emptyset.
\end{equation}

(a), (b), (e), and (f) can be proved straightforwardly.

Let ${\mf A}$ be an interpretation for ${\mc L}$ such that ${\mf A}\models \phi$.
Then 
${\mf A}\models \phi\overset{\text{Lemma \ref{le111}(a)}}{\eqvl\!\!\eqvl\!\!\eqvl\!\!\eqvl\!\!\eqvl\!\!\eqvl\!\!\eqvl\!\!\eqvl\!\!\eqvl} \phi'$.
We define an expansion ${\mf A}^\#$ of ${\mf A}$ to ${\mc L}\cup \{\tilde{p}_\mbbm{i}\}$ as follows:
\begin{equation*}
\tilde{p}_\mbbm{i}^{{\mf A}^\#}(u_1,\dots,u_{|\bar{x}|})=
\left\{\begin{array}{ll}
       \|\phi'\|_e^{\mf A} &\ \text{\it if there exists}\ e\in {\mc S}_{\mf A}\ \text{\it such that} \\  
                           &\ \quad \|\bar{x}\|_e^{\mf A}=u_1,\dots,u_{|\bar{x}|}, \\[1mm]
       0                   &\ \text{\it else}.
       \end{array}
\right.
\end{equation*}
Then, for all $e\in {\mc S}_{{\mf A}^\#}$, 
$\|\tilde{p}_\mbbm{i}(\bar{x})\|_e^{{\mf A}^\#}=\tilde{p}_\mbbm{i}^{{\mf A}^\#}(\|\bar{x}\|_e^{{\mf A}^\#})=
 \tilde{p}_\mbbm{i}^{{\mf A}^\#}(\|\bar{x}\|_e^{\mf A})=\|\phi'\|_e^{\mf A}=\|\phi'\|_e^{{\mf A}^\#}=1$,
$\|\tilde{p}_\mbbm{i}(\bar{x})\geql \gu\|_e^{{\mf A}^\#}=\|\tilde{p}_\mbbm{i}(\bar{x})\|_e^{{\mf A}^\#}\feql 1=1\feql 1=1$,
\begin{alignat*}{1}
& \|\tilde{p}_\mbbm{i}(\bar{x})\leftrightarrow \phi'\|_e^{{\mf A}^\#}= \\
& (\|\tilde{p}_\mbbm{i}(\bar{x})\|_e^{{\mf A}^\#}\frightarrow \|\phi'\|_e^{{\mf A}^\#})\fwedge
  (\|\phi'\|_e^{{\mf A}^\#}\frightarrow \|\tilde{p}_\mbbm{i}(\bar{x})\|_e^{{\mf A}^\#})= \\
& 1\frightarrow 1=1;
\end{alignat*}
${\mf A}^\#\models \tilde{p}_\mbbm{i}(\bar{x})\geql \gu$,
${\mf A}^\#\models \tilde{p}_\mbbm{i}(\bar{x})\leftrightarrow \phi'$,
by Lemma \ref{le11}(b) for ${\mf A}^\#$, 
there exists an interpretation ${\mf A}'$ 
for ${\mc L}\cup \{\tilde{p}_\mbbm{j} \,|\, \mbbm{j}\in \{\mbbm{i}\}\cup J\overset{\text{(\ref{eq1k})}}{=\!\!=} J_\phi\}$, and
${\mf A}'\models S$, ${\mf A}'|_{{\mc L}\cup \{\tilde{p}_\mbbm{i}\}}={\mf A}^\#$;
${\mf A}'|_{{\mc L}\cup \{\tilde{p}_\mbbm{i}\}}={\mf A}^\#\models \tilde{p}_\mbbm{i}(\bar{x})\geql \gu$;
${\mf A}'\models S_\phi$, ${\mf A}'|_{\mc L}={\mf A}$.

Let ${\mf A}'$ be an interpretation 
for ${\mc L}\cup \{\tilde{p}_\mbbm{j} \,|\, \mbbm{j}\in J_\phi\overset{\text{(\ref{eq1k})}}{=\!\!=} \{\mbbm{i}\}\cup J\}$ such that
${\mf A}'\models S_\phi$.
Then ${\mf A}'\models \tilde{p}_\mbbm{i}(\bar{x})\geql \gu, S$,
by Lemma \ref{le11}(b),
${\mf A}'|_{{\mc L}\cup \{\tilde{p}_\mbbm{i}\}}\models \tilde{p}_\mbbm{i}(\bar{x})\leftrightarrow \phi'$,
for all $e\in {\mc S}_{{\mf A}'}$,
$1=\|\tilde{p}_\mbbm{i}(\bar{x})\geql \gu\|_e^{{\mf A}'}=\|\tilde{p}_\mbbm{i}(\bar{x})\|_e^{{\mf A}'}\feql 1$,
$\|\tilde{p}_\mbbm{i}(\bar{x})\|_e^{{\mf A}'}=1$,
$1=\|\tilde{p}_\mbbm{i}(\bar{x})\leftrightarrow \phi'\|_e^{{\mf A}'}=
   (\|\tilde{p}_\mbbm{i}(\bar{x})\|_e^{{\mf A}'}\frightarrow \|\phi'\|_e^{{\mf A}'})\fwedge                                        
   (\|\phi'\|_e^{{\mf A}'}\frightarrow \|\tilde{p}_\mbbm{i}(\bar{x})\|_e^{{\mf A}'})=
   (1\frightarrow \|\phi'\|_e^{{\mf A}'})\fwedge (\|\phi'\|_e^{{\mf A}'}\frightarrow 1)$,
$1\frightarrow \|\phi'\|_e^{{\mf A}'}=1$, $\|\phi'\|_e^{{\mf A}'}=1$;                                                             
${\mf A}'\models \phi'\overset{\text{Lemma \ref{le111}(a)}}{\eqvl\!\!\eqvl\!\!\eqvl\!\!\eqvl\!\!\eqvl\!\!\eqvl\!\!\eqvl\!\!\eqvl\!\!\eqvl} \phi$.
We put ${\mf A}={\mf A}'|_{\mc L}$, an interpretation for ${\mc L}$.
Then ${\mf A}\models \phi$, ${\mf A}={\mf A}'|_{\mc L}$; (c) holds.

We have $|\bar{x}|\leq |\phi'|$.
Then
$|S_\phi|\overset{\text{(\ref{eq1m})}}{=\!\!=} |\{\tilde{p}_\mbbm{i}(\bar{x})\geql \gu\}|+|S|=
 3+|\bar{x}|+|S|\underset{\text{Lemma \ref{le11}(c)}}{\leq} 
 3+|\bar{x}|+27\cdot |\phi'|\cdot (1+|\bar{x}|)\leq 
 3+|\phi'|+27\cdot |\phi'|\cdot (1+|\phi'|)\leq 58\cdot |\phi'|^2\underset{\text{Lemma \ref{le111}(b)}}{\leq} 232\cdot |\phi|^2\in O(|\phi|^2)$;
the translation of $\phi$ to $S_\phi$ uses the input $\phi$, the output $S_\phi$,
auxiliary $\phi'$, $\tilde{f}_0(\bar{x})$, $\{\tilde{p}_\mbbm{i}(\bar{x})\geql \gu\}$, $S$;
we have that $\phi'$ can be built up from $\phi$ via a postorder traversal of $\phi$ with $\#{\mc O}(\phi)\in O(|\phi|)$;
the test $\phi'\not\in \overline{C}_{\mc L}$ is with $\#{\mc O}(\phi')\in O(1)$;
$\tilde{f}_0(\bar{x})$ can be built up from $\phi'$ via the left-right preorder traversal of $\phi'$
with $\#{\mc O}(\phi')\in O(|\phi'|)\underset{\text{Lemma \ref{le111}(b)}}{\subseteq} O(|\phi|)$;
$\{\tilde{p}_\mbbm{i}(\bar{x})\geql \gu\}$ can be built up from $\tilde{f}_0(\bar{x})$
with $\#{\mc O}(\tilde{f}_0(\bar{x}))\in O(|\{\tilde{p}_\mbbm{i}(\bar{x})\geql \gu\}|)=O(1+|\bar{x}|)\subseteq 
      O(|\phi'|)\underset{\text{Lemma \ref{le111}(b)}}{\subseteq} O(|\phi|)$;
by Lemma \ref{le11}(c), $S$ can be built up from $\phi'$ and $\tilde{f}_0(\bar{x})$ via a preorder traversal of $\phi'$
with $\#{\mc O}(\phi',\tilde{f}_0(\bar{x}))\in O(|\phi'|\cdot (1+|\bar{x}|))\subseteq O(|\phi'|\cdot (1+|\phi'|))=
      O(|\phi'|^2)\underset{\text{Lemma \ref{le111}(b)}}{\subseteq} O(|\phi|^2)$;
$S_\phi$ can be built up from $\{\tilde{p}_\mbbm{i}(\bar{x})\geql \gu\}$ and $S$ by copying and concatenating
with $\#{\mc O}(\{\tilde{p}_\mbbm{i}(\bar{x})\geql \gu\},S)\in O(|S_\phi|)\subseteq O(|\phi|^2)$;
the number of all elementary operations of the translation of $\phi$ to $S_\phi$ $\#{\mc O}(\phi)\in O(|\phi|^2)$;
by (\ref{eq00t}) for $n_\phi$, $\phi$, $\emptyset$,
$\phi'$, $\tilde{f}_0(\bar{x})$, $\{\tilde{p}_\mbbm{i}(\bar{x})\geql \gu\}$, $S$, $S_\phi$, $q=6$, $r=2$,
the time complexity of the translation of $\phi$ to $S_\phi$ is
in $O(\#{\mc O}(\phi)\cdot (\log (1+n_\phi)+\log (\#{\mc O}(\phi)+|\phi|)))\subseteq O(|\phi|^2\cdot (\log (1+n_\phi)+\log |\phi|))$;
by (\ref{eq00s}) for $n_\phi$, $\phi$, $\emptyset$, 
$\phi'$, $\tilde{f}_0(\bar{x})$, $\{\tilde{p}_\mbbm{i}(\bar{x})\geql \gu\}$, $S$, $S_\phi$, $q=6$, $r=2$,
the space complexity of the translation of $\phi$ to $S_\phi$ is
in $O((\#{\mc O}(\phi)+|\phi|^2)\cdot (\log (1+n_\phi)+\log |\phi|))\subseteq O(|\phi|^2\cdot (\log (1+n_\phi)+\log |\phi|))$; (d) holds.

So, in all Cases 1--3, (a--f) hold.
\qed
\end{proof}

\begin{corollary}
\label{cor12}
Let $n_0\in \mbb{N}$ and $T\subseteq \mi{Form}_{\mc L}$.
There exist $J_T\subseteq \{(i,j) \,|\, i\geq n_0\}\subseteq \mbb{I}$ and 
$S_T\subseteq \mi{OrdCl}_{{\mc L}\cup \{\tilde{p}_\mbbm{j} \,|\, \mbbm{j}\in J_T\}}$ such that 
\begin{enumerate}[\rm (a)]
\item
either $J_T=\emptyset$, $S_T=\{\square\}$, or $J_T=S_T=\emptyset$, or $J_T\neq \emptyset$, $\square\not\in S_T\neq \emptyset$; 
\item
there exists an interpretation ${\mf A}$ for ${\mc L}$ and ${\mf A}\models T$ if and only if 
there exists an interpretation ${\mf A}'$ for ${\mc L}\cup \{\tilde{p}_\mbbm{j} \,|\, \mbbm{j}\in J_T\}$ and
${\mf A}'\models S_T$, satisfying ${\mf A}={\mf A}'|_{\mc L}$; 
\item
if $T\subseteq_{\mc F} \mi{Form}_{\mc L}$, then $J_T\subseteq_{\mc F} \{(i,j) \,|\, i\geq n_0\}\subseteq \mbb{I}$, $\|J_T\|\leq 2\cdot |T|$, 
$S_T\subseteq_{\mc F} \mi{OrdCl}_{{\mc L}\cup \{\tilde{p}_\mbbm{j} \,|\, \mbbm{j}\in J_T\}}$, $|S_T|\in O(|T|^2)$; 
the number of all elementary operations of the translation of $T$ to $S_T$ is in $O(|T|^2)$;
the time complexity of the translation of $T$ to $S_T$ is in $O(|T|^2\cdot \log (1+n_0+|T|))$; 
\item
if $S_T\neq \emptyset, \{\square\}$, then $J_T\neq \emptyset$,
for all $C\in S_T$, $\emptyset\neq \mi{preds}(C)\cap \tilde{\mbb{P}}\subseteq \{\tilde{p}_\mbbm{j} \,|\, \mbbm{j}\in J_T\}$;
\item
$\mi{tcons}(S_T)-\{\gz,\gu\}\subseteq \mi{tcons}(T)-\{\gz,\gu\}$.
\end{enumerate}
\end{corollary}

\begin{proof}
A straightforward consequence of Lemma \ref{le1}.
\qed
\end{proof}

The deduction problem of a formula $\phi\in \mi{Form}_{\mc L}$ from a theory $T\subseteq \mi{Form}_{\mc L}$ can be reduced 
to the unsatisfiability of a certain order clausal theory $S_T^\phi$ so that
$T\models \phi$ if and only if $S_T^\phi$ is unsatisfiable.
The theory $S_T^\phi$ can be obtained by the proposed translation.
The main idea is that all the formulae from $T$ are translated positively -- by Lemma~\ref{le1}, Corollary~\ref{cor12}.
We always introduce fresh predicates for auxiliary atoms corresponding to subformulae and 
start translation of every formula $\psi\in T$ with the first clause in the initial theory of the form $\tilde{p}_\mi{int}(\bar{x})\geql \gu$ 
where $\tilde{p}_\mi{int}(\bar{x})$ is the auxiliary atom corresponding to the entire formula $\psi$, 
$\mi{vars}(\psi)=\mi{vars}(\bar{x})$.
Hence, the resulting order clausal theory for every $\psi\in T$ will be equisatisfiable to $\psi$.
However, in contrast, the translation of $\phi$ has to start negatively.
Let $\bar{x}$ be a sequence of all the variables occurring in $\phi$. 
We shall translate its universal closure $\forall \bar{x}\, \phi$, 
$\mi{vars}(\forall \bar{x}\, \phi)=\mi{vars}(\phi)=\mi{vars}(\bar{x})$, so that
the first clause in the initial theory will be of the form $\tilde{p}_\mi{int}(\bar{x})\gle \gu$;
the auxiliary atom $\tilde{p}_\mi{int}(\bar{x})$ corresponds to $\forall \bar{x}\, \phi$.
Then the resulting order clausal theory for $\phi$ will be equisatisfiable to $\forall \bar{x}\, \phi\gle \gu$.
Notice that for the closed formula $\phi$ from the informal example at the beginning of this section,
we just replace the first clause $\tilde{p}_0(x,y,z)\geql \gu$ with $\tilde{p}_0(x,y,z)\gle \gu$ in the resulting theory, 
which then becomes equisatisfiable to $\phi\gle \gu$.
Finally, the resulting order clausal theory $S_T^\phi$ is obtained as the union of all the partial theories for every $\psi\in T$ and $\phi$.

\begin{theorem}
\label{T1}
Let $n_0\in \mbb{N}$, $\phi\in \mi{Form}_{\mc L}$, $T\subseteq \mi{Form}_{\mc L}$. 
There exist $J_T^\phi\subseteq \{(i,j) \,|\, i\geq n_0\}\subseteq \mbb{I}$ and
$S_T^\phi\subseteq \mi{OrdCl}_{{\mc L}\cup \{\tilde{p}_\mbbm{j} \,|\, \mbbm{j}\in J_T^\phi\}}$ such that 
\begin{enumerate}[\rm (i)]
\item
there exists an interpretation ${\mf A}$ for ${\mc L}$, and ${\mf A}\models T$, ${\mf A}\not\models \phi$, if and only if 
there exists an interpretation ${\mf A}'$ for ${\mc L}\cup \{\tilde{p}_\mbbm{j} \,|\, \mbbm{j}\in J_T^\phi\}$ and ${\mf A}'\models S_T^\phi$, 
satisfying ${\mf A}={\mf A}'|_{\mc L}$; 
\item
$T\models \phi$ if and only if $S_T^\phi$ is unsatisfiable;
\item
if $T\subseteq_{\mc F} \mi{Form}_{\mc L}$, then $J_T^\phi\subseteq_{\mc F} \{(i,j) \,|\, i\geq n_0\}\subseteq \mbb{I}$, 
$\|J_T^\phi\|\in O(|T|+|\phi|)$, 
$S_T^\phi\subseteq_{\mc F} \mi{OrdCl}_{{\mc L}\cup \{\tilde{p}_\mbbm{j} \,|\, \mbbm{j}\in J_T^\phi\}}$, $|S_T^\phi|\in O(|T|^2+|\phi|^2)$; 
the number of all elementary operations of the translation of $T$ and $\phi$ to $S_T^\phi$ is in $O(|T|^2+|\phi|^2)$;
the time complexity of the translation of $T$ and $\phi$ to $S_T^\phi$ is 
in $O(|T|^2\cdot \log (1+n_0+|T|)+|\phi|^2\cdot (\log (1+n_0)+\log |\phi|))$;
\item
$\mi{tcons}(S_T^\phi)-\{\gz,\gu\}\subseteq (\mi{tcons}(\phi)\cup \mi{tcons}(T))-\{\gz,\gu\}$.
\end{enumerate}
\end{theorem}

\begin{proof}
We get by Corollary \ref{cor12} for $n_0+1$ that 
there exist $J_T\subseteq \{(i,j) \,|\, i\geq n_0+1\}\subseteq \mbb{I}$,
$S_T\subseteq \mi{OrdCl}_{{\mc L}\cup \{\tilde{p}_\mbbm{j} \,|\, \mbbm{j}\in J_T\}}$, and
(a--e) of Corollary \ref{cor12} hold for $n_0+1$.
By Lemma~\ref{le111} for $n_0$, $\phi$,
there exists $\phi'\in \mi{Form}_{\mc L}$ such that (a--e) of Lemma \ref{le111} hold for $n_0$, $\phi$, $\phi'$.
We distinguish three cases for $\phi'$.

Case 1: 
$\phi'\in \overline{C}_{\mc L}-\{\gu\}$.
We put $J_T^\phi=J_T\subseteq \{(i,j) \,|\, i\geq n_0+1\}\subseteq \{(i,j) \,|\, i\geq n_0\}\subseteq \mbb{I}$ and
$S_T^\phi=S_T\subseteq \mi{OrdCl}_{{\mc L}\cup \{\tilde{p}_\mbbm{j} \,|\, \mbbm{j}\in J_T^\phi\}}$.

For every interpretation ${\mf A}$ for ${\mc L}$, 
${\mf A}\not\models \phi'\overset{\text{Lemma \ref{le111}(a)}}{\eqvl\!\!\eqvl\!\!\eqvl\!\!\eqvl\!\!\eqvl\!\!\eqvl\!\!\eqvl\!\!\eqvl\!\!\eqvl}
                    \phi$;
by Corollary \ref{cor12}(b),
there exists an interpretation ${\mf A}$ for ${\mc L}$, and ${\mf A}\models T$, ${\mf A}\not\models \phi$, if and only if 
there exists an interpretation ${\mf A}'$ for ${\mc L}\cup \{\tilde{p}_\mbbm{j} \,|\, \mbbm{j}\in J_T^\phi\}$ and ${\mf A}'\models S_T^\phi$, 
satisfying ${\mf A}={\mf A}'|_{\mc L}$; (i) holds.

(iii) and (iv) can be proved straightforwardly.

Case 2:
$\phi'=\gu$.
We put $J_T^\phi=\emptyset\subseteq \{(i,j) \,|\, i\geq n_0\}\subseteq \mbb{I}$ and $S_T^\phi=\{\square\}\subseteq \mi{OrdCl}_{\mc L}$.

For every interpretation ${\mf A}$ for ${\mc L}$,
${\mf A}\models \phi'\overset{\text{Lemma \ref{le111}(a)}}{\eqvl\!\!\eqvl\!\!\eqvl\!\!\eqvl\!\!\eqvl\!\!\eqvl\!\!\eqvl\!\!\eqvl\!\!\eqvl} \phi$,
${\mf A}\not\models S_T^\phi$;
trivially, 
there exists an interpretation ${\mf A}$ for ${\mc L}$, and ${\mf A}\models T$, ${\mf A}\not\models \phi$, if and only if 
there exists an interpretation ${\mf A}'$ for ${\mc L}$ and ${\mf A}'\models S_T^\phi$, 
satisfying ${\mf A}={\mf A}'={\mf A}'|_{\mc L}$; (i) holds.

(iii) and (iv) can be proved straightforwardly.

Case 3:
$\phi'\not\in \overline{C}_{\mc L}$.
We have that $\phi'\in \mi{Form}_{\mc L}$, (c,d) of Lemma \ref{le111} hold for $\phi'$.
We put $\bar{x}=\mi{varseq}(\phi')$.
Then $\phi'\in \mi{Form}_{\mc L}-\overline{C}_{\mc L}\subseteq \mi{Form}_{\mc L}-\{\gz,\gu\}$,
$\mi{vars}(\phi')=\mi{vars}(\bar{x})$, 
$\forall \bar{x}\, \phi'\in \mi{Form}_{\mc L}-\{\gz,\gu\}$, (c,d) of Lemma \ref{le111} hold for $\forall \bar{x}\, \phi'$,
$\mi{vars}(\forall \bar{x}\, \phi')=\mi{vars}(\bar{x})\cup \mi{vars}(\phi')=\mi{vars}(\bar{x})$,
$|\bar{x}|\leq |\phi'|\leq |\forall \bar{x}\, \phi'|$.
We put $j_\mbbm{i}=0$ and $\mbbm{i}=(n_0,j_\mbbm{i})\in \{(n_0,j) \,|\, j\in \mbb{N}\}\subseteq \mbb{I}$.
$\tilde{p}_\mbbm{i}\in \tilde{\mbb{P}}$.
We put $\mi{ar}(\tilde{p}_\mbbm{i})=|\bar{x}|$.
We get by Lemma \ref{le11} for $n_0$, $\forall \bar{x}\, \phi'$ that
there exist $n_J\geq j_\mbbm{i}$,
$J=\{(n_0,j) \,|\, 1\leq j\leq n_J\}\subseteq \{(n_0,j) \,|\, j\in \mbb{N}\}\subseteq \mbb{I}$, $\mbbm{i}\not\in J$,
$S\subseteq_{\mc F} \mi{OrdCl}_{{\mc L}\cup \{\tilde{p}_\mbbm{i}\}\cup \{\tilde{p}_\mbbm{j} \,|\, \mbbm{j}\in J\}}$, and  
(a--e) of Lemma \ref{le11} hold for $\forall \bar{x}\, \phi'$.
We put $J_T^\phi=J_T\cup \{\mbbm{i}\}\cup J\subseteq \{(i,j) \,|\, i\geq n_0\}\subseteq \mbb{I}$.
Then $J_T\cap (\{\mbbm{i}\}\cup J)\subseteq \{(i,j) \,|\, i\geq n_0+1\}\cap \{(n_0,j) \,|\, j\in \mbb{N}\}=\emptyset$,
\begin{equation}
\label{eq2c}
J_T, \{\mbbm{i}\}, J\ \text{are pairwise disjoint}.
\end{equation}
We put $S_T^\phi=S_T\cup \{\tilde{p}_\mbbm{i}(\bar{x})\gle \gu\}\cup S\subseteq
                 \mi{OrdCl}_{{\mc L}\cup \{\tilde{p}_\mbbm{j} \,|\, \mbbm{j}\in J_T^\phi\}}$.
It can be proved that
\begin{equation}
\label{eq2d}
S_T, \{\tilde{p}_\mbbm{i}(\bar{x})\gle \gu\}, S\ \text{are pairwise disjoint}.
\end{equation}

Let ${\mf A}$ be an interpretation for ${\mc L}$ such that ${\mf A}\models T$ and ${\mf A}\not\models \phi$.
Then, by Corollary~\ref{cor12}(b),
there exists an interpretation ${\mf A}_T$ for ${\mc L}\cup \{\tilde{p}_\mbbm{j} \,|\, \mbbm{j}\in J_T\}$, and
${\mf A}_T\models S_T$, ${\mf A}_T|_{\mc L}={\mf A}$;
we have $\mi{vars}(\bar{x})=\mi{vars}(\phi')$;
$\forall \bar{x}\, \phi'$ is closed,
${\mf A}\not\models \phi\overset{\text{Lemma \ref{le111}(a)}}{\eqvl\!\!\eqvl\!\!\eqvl\!\!\eqvl\!\!\eqvl\!\!\eqvl\!\!\eqvl\!\!\eqvl\!\!\eqvl}
                    \phi'$,
${\mf A}\not\models \forall \bar{x}\, \phi'$, $\|\forall \bar{x}\, \phi'\|^{\mf A}<1$.
We define an expansion ${\mf A}^\#$ of ${\mf A}$ to ${\mc L}\cup \{\tilde{p}_\mbbm{i}\}$ as follows:
\begin{equation*}
\tilde{p}_\mbbm{i}^{{\mf A}^\#}(u_1,\dots,u_{|\bar{x}|})=\|\forall \bar{x}\, \phi'\|^{\mf A}.
\end{equation*}
Then, for all $e\in {\mc S}_{{\mf A}^\#}$, 
$\|\tilde{p}_\mbbm{i}(\bar{x})\|_e^{{\mf A}^\#}=\tilde{p}_\mbbm{i}^{{\mf A}^\#}(\|\bar{x}\|_e^{{\mf A}^\#})=\|\forall \bar{x}\, \phi'\|^{\mf A}=
 \|\forall \bar{x}\, \phi'\|^{{\mf A}^\#}<1$,
$\|\tilde{p}_\mbbm{i}(\bar{x})\gle \gu\|_e^{{\mf A}^\#}=\|\tilde{p}_\mbbm{i}(\bar{x})\|_e^{{\mf A}^\#}\fle 1=1$,
$\|\tilde{p}_\mbbm{i}(\bar{x})\leftrightarrow \forall \bar{x}\, \phi'\|_e^{{\mf A}^\#}=
 (\|\tilde{p}_\mbbm{i}(\bar{x})\|_e^{{\mf A}^\#}\frightarrow \|\forall \bar{x}\, \phi'\|_e^{{\mf A}^\#})\fwedge                     
 (\|\forall \bar{x}\, \phi'\|_e^{{\mf A}^\#}\frightarrow \|\tilde{p}_\mbbm{i}(\bar{x})\|_e^{{\mf A}^\#})=
 \|\forall \bar{x}\, \phi'\|^{{\mf A}^\#}\frightarrow \|\forall \bar{x}\, \phi'\|^{{\mf A}^\#}=1$; 
${\mf A}^\#\models \tilde{p}_\mbbm{i}(\bar{x})\gle \gu$,
${\mf A}^\#\models \tilde{p}_\mbbm{i}(\bar{x})\leftrightarrow \forall \bar{x}\, \phi'$,  
by Lemma \ref{le11}(b) for ${\mf A}^\#$,
there exists an interpretation ${\mf A}_\phi$ for ${\mc L}\cup \{\tilde{p}_\mbbm{i}\}\cup \{\tilde{p}_\mbbm{j} \,|\, \mbbm{j}\in J\}$, and
${\mf A}_\phi\models S$, ${\mf A}_\phi|_{{\mc L}\cup \{\tilde{p}_\mbbm{i}\}}={\mf A}^\#$.
$\{\tilde{p}_\mbbm{j} \,|\, \mbbm{j}\in J_T\}\cap 
 (\{\tilde{p}_\mbbm{i}\}\cup \{\tilde{p}_\mbbm{j} \,|\, \mbbm{j}\in J\})\overset{\text{(\ref{eq2c})}}{=\!\!=} \emptyset$.
We define an expansion ${\mf A}'$ of ${\mf A}$ to ${\mc L}\cup \{\tilde{p}_\mbbm{j} \,|\, \mbbm{j}\in J_T^\phi\}$ as follows:
\begin{equation*}
\tilde{p}_\mbbm{j}^{{\mf A}'}=\left\{\begin{array}{ll} 
                                     \tilde{p}_\mbbm{j}^{{\mf A}_T}    &\ \text{\it if}\ \mbbm{j}\in J_T, \\[1mm]
                                     \tilde{p}_\mbbm{j}^{{\mf A}_\phi} &\ \text{\it if}\ \mbbm{j}\in \{\mbbm{i}\}\cup J.
                                     \end{array}
                              \right.
\end{equation*}
We get ${\mf A}'|_{{\mc L}\cup \{\tilde{p}_\mbbm{j} \,|\, \mbbm{j}\in J_T\}}={\mf A}_T\models S_T$,
${\mf A}'|_{{\mc L}\cup \{\tilde{p}_\mbbm{i}\}}={\mf A}^\#\models \tilde{p}_\mbbm{i}(\bar{x})\gle \gu$,
${\mf A}'|_{{\mc L}\cup \{\tilde{p}_\mbbm{i}\}\cup \{\tilde{p}_\mbbm{j} \,|\, \mbbm{j}\in J\}}={\mf A}_\phi\models S$;
${\mf A}'\models S_T^\phi$, ${\mf A}'|_{\mc L}={\mf A}$.

Let ${\mf A}'$ be an interpretation for ${\mc L}\cup \{\tilde{p}_\mbbm{j} \,|\, \mbbm{j}\in J_T^\phi\}$ such that ${\mf A}'\models S_T^\phi$.
Then                                                                                                                               
${\mf A}'|_{{\mc L}\cup \{\tilde{p}_\mbbm{j} \,|\, \mbbm{j}\in J_T\}}\models S_T$,
${\mf A}'|_{{\mc L}\cup \{\tilde{p}_\mbbm{i}\}}\models \tilde{p}_\mbbm{i}(\bar{x})\gle \gu$,
${\mf A}'|_{{\mc L}\cup \{\tilde{p}_\mbbm{i}\}\cup \{\tilde{p}_\mbbm{j} \,|\, \mbbm{j}\in J\}}\models S$,
by Corollary~\ref{cor12}(b) for ${\mf A}'|_{{\mc L}\cup \{\tilde{p}_\mbbm{j} \,|\, \mbbm{j}\in J_T\}}$, 
${\mf A}'|_{\mc L}\models T$,
by Lemma \ref{le11}(b) for ${\mf A}'|_{{\mc L}\cup \{\tilde{p}_\mbbm{i}\}\cup \{\tilde{p}_\mbbm{j} \,|\, \mbbm{j}\in J\}}$, 
${\mf A}'|_{{\mc L}\cup \{\tilde{p}_\mbbm{i}\}}\models \tilde{p}_\mbbm{i}(\bar{x})\leftrightarrow \forall \bar{x}\, \phi'$;
for all $e\in {\mc S}_{{\mf A}'}$, 
$1=\|\tilde{p}_\mbbm{i}(\bar{x})\gle \gu\|_e^{{\mf A}'}=\|\tilde{p}_\mbbm{i}(\bar{x})\|_e^{{\mf A}'}\fle 1$,
$\|\tilde{p}_\mbbm{i}(\bar{x})\|_e^{{\mf A}'}<1$,
\begin{alignat*}{1}
& 1=\|\tilde{p}_\mbbm{i}(\bar{x})\leftrightarrow \forall \bar{x}\, \phi'\|_e^{{\mf A}'}= \\
& \phantom{1=\mbox{}}
    (\|\tilde{p}_\mbbm{i}(\bar{x})\|_e^{{\mf A}'}\frightarrow \|\forall \bar{x}\, \phi'\|^{{\mf A}'})\fwedge
    (\|\forall \bar{x}\, \phi'\|^{{\mf A}'}\frightarrow \|\tilde{p}_\mbbm{i}(\bar{x})\|_e^{{\mf A}'}),
\end{alignat*}
$\|\forall \bar{x}\, \phi'\|^{{\mf A}'}\frightarrow \|\tilde{p}_\mbbm{i}(\bar{x})\|_e^{{\mf A}'}=1$,
$\|\forall \bar{x}\, \phi'\|^{{\mf A}'}\leq \|\tilde{p}_\mbbm{i}(\bar{x})\|_e^{{\mf A}'}<1$;
${\mf A}'|_{\mc L}\not\models \forall \bar{x}\, \phi'$,
${\mf A}'|_{\mc L}\not\models \phi'\overset{\text{Lemma \ref{le111}(a)}}{\eqvl\!\!\eqvl\!\!\eqvl\!\!\eqvl\!\!\eqvl\!\!\eqvl\!\!\eqvl\!\!\eqvl\!\!\eqvl}
                              \phi$.
We put ${\mf A}={\mf A}'|_{\mc L}$, an interpretation for ${\mc L}$.
Then ${\mf A}\models T$ and ${\mf A}\not\models \phi$, ${\mf A}={\mf A}'|_{\mc L}$; (i) holds.

Let $T\subseteq_{\mc F} \mi{Form}_{\mc L}$.
Then, by Corollary \ref{cor12}(c), $J_T\subseteq_{\mc F} \{(i,j) \,|\, i\geq n_0+1\}\subseteq \{(i,j) \,|\, i\geq n_0\}\subseteq \mbb{I}$, 
$\|J_T\|\leq 2\cdot |T|$,
$S_T\subseteq_{\mc F} \mi{OrdCl}_{{\mc L}\cup \{\tilde{p}_\mbbm{j} \,|\, \mbbm{j}\in J_T\}}$, $|S_T|\in O(|T|^2)$;
we have $\mbbm{i}\in \{(n_0,j) \,|\, j\in \mbb{N}\}\subseteq \mbb{I}$, $|\bar{x}|\leq |\phi'|$;
$J\subseteq_{\mc F} \{(n_0,j) \,|\, j\in \mbb{N}\}\subseteq \mbb{I}$,
$\|J\|\underset{\text{Lemma \ref{le11}(a)}}{\leq} |\forall \bar{x}\, \phi'|-1<|\forall \bar{x}\, \phi'|=
 2\cdot |\bar{x}|+|\phi'|\leq 3\cdot |\phi'|\underset{\text{Lemma \ref{le111}(b)}}{\leq} 6\cdot |\phi|$;
we have $S\subseteq_{\mc F} \mi{OrdCl}_{{\mc L}\cup \{\tilde{p}_\mbbm{i}\}\cup \{\tilde{p}_\mbbm{j} \,|\, \mbbm{j}\in J\}}$;
$|S|, |\forall \bar{x}\, \phi'|\cdot (1+|\bar{x}|)\underset{\text{Lemma \ref{le11}(c)}}{\leq} 
 27\cdot |\forall \bar{x}\, \phi'|\cdot (1+|\bar{x}|)\leq 27\cdot 3\cdot |\phi'|\cdot 2\cdot |\phi'|=
 162\cdot |\phi'|^2\underset{\text{Lemma \ref{le111}(b)}}{\leq} 648\cdot |\phi|^2\in O(|\phi|^2)$;
$J_T^\phi\subseteq_{\mc F} \{(i,j) \,|\, i\geq n_0\}\subseteq \mbb{I}$,
$\|J_T^\phi\|\overset{\text{(\ref{eq2c})}}{=\!\!=} \|J_T\|+\|\{\mbbm{i}\}\|+\|J\|\leq 2\cdot |T|+1+6\cdot |\phi|\in O(|T|+|\phi|)$,
$S_T^\phi\subseteq_{\mc F} \mi{OrdCl}_{{\mc L}\cup \{\tilde{p}_\mbbm{j} \,|\, \mbbm{j}\in J_T^\phi\}}$,
$|S_T^\phi|\overset{\text{(\ref{eq2d})}}{=\!\!=} |S_T|+|\{\tilde{p}_\mbbm{i}(\bar{x})\gle \gu\}|+|S|=
 |S_T|+|\bar{x}|+3+|S|\leq |S_T|+4\cdot |\phi'|+|S|\underset{\text{Lemma \ref{le111}(b)}}{\leq} |S_T|+8\cdot |\phi|+|S|\in O(|T|^2+|\phi|^2)$;
the translation of $T$ and $\phi$ to $S_T^\phi$ uses the input $T$, $\phi$, the output $S_T^\phi$, 
auxiliary $S_T$, $\phi'$, $\tilde{f}_0(\bar{x})$, $\forall \bar{x}\, \phi'$, $\{\tilde{p}_\mbbm{i}(\bar{x})\gle \gu\}$, $S$;
we have that $\phi'$ can be built up from $\phi$ via a postorder traversal of $\phi$ with $\#{\mc O}_1(\phi)\in O(|\phi|)$;
the test $\phi'\not\in \overline{C}_{\mc L}$ is with $\#{\mc O}_2(\phi')\in O(1)$;
by Corollary~\ref{cor12}(c), the number of all elementary operations of the translation of $T$ to $S_T$ is in $O(|T|^2)$;
the time complexity of the translation of $T$ to $S_T$ is in $O(|T|^2\cdot \log (1+n_0+|T|))$;
$\tilde{f}_0(\bar{x})$ can be built up from $\phi'$ via the left-right preorder traversal of $\phi'$ 
with $\#{\mc O}_3(\phi')\in O(|\phi'|)\underset{\text{Lemma \ref{le111}(b)}}{\subseteq} O(|\phi|)$;
$\forall \bar{x}\, \phi'$ can be built up from $\phi'$ and $\tilde{f}_0(\bar{x})$ 
with $\#{\mc O}_4(\phi',\tilde{f}_0(\bar{x}))\in O(|\forall \bar{x}\, \phi'|)\subseteq O(|\phi|)$;
$\{\tilde{p}_\mbbm{i}(\bar{x})\gle \gu\}$ can be built up from $\tilde{f}_0(\bar{x})$ 
with $\#{\mc O}_5(\tilde{f}_0(\bar{x}))\in O(|\{\tilde{p}_\mbbm{i}(\bar{x})\gle \gu\}|)=O(1+|\bar{x}|)\subseteq
                                           O(|\phi'|)\underset{\text{Lemma \ref{le111}(b)}}{\subseteq} O(|\phi|)$;
by Lemma \ref{le11}(c), $S$ can be built up from $\forall \bar{x}\, \phi'$ and $\tilde{f}_0(\bar{x})$ 
via a preorder traversal of $\forall \bar{x}\, \phi'$
with $\#{\mc O}_6(\forall \bar{x}\, \phi',\tilde{f}_0(\bar{x}))\in O(|\forall \bar{x}\, \phi'|\cdot (1+|\bar{x}|))\subseteq O(|\phi|^2)$;
$S_T^\phi$ can be built up from $\{\tilde{p}_\mbbm{i}(\bar{x})\gle \gu\}$ and $S$ by copying and appending to $S_T$
with $\#{\mc O}_7(\{\tilde{p}_\mbbm{i}(\bar{x})\gle \gu\},S)\in O(|\{\tilde{p}_\mbbm{i}(\bar{x})\gle \gu\}|+|S|)\subseteq O(|\phi|^2)$;
$\sum_{i=1}^7 \#{\mc O}_i\in O(|\phi|^2)$,
by (\ref{eq00t}) for $n_0$, $\phi$, $\emptyset$, 
$\phi'$, $\tilde{f}_0(\bar{x})$, $\forall \bar{x}\, \phi'$, $\{\tilde{p}_\mbbm{i}(\bar{x})\gle \gu\}$, $S$,
$\{\tilde{p}_\mbbm{i}(\bar{x})\gle \gu\}$ (a copy), $S$ (a copy), $q=8$, $r=2$,
the time complexity of elementary operations at the stages $1,\dots,7$ is
in $O((\sum_{i=1}^7 \#{\mc O}_i)\cdot (\log (1+n_0)+\log ((\sum_{i=1}^7 \#{\mc O}_i)+|\phi|)))\subseteq
    O(|\phi|^2\cdot (\log (1+n_0)+\log |\phi|))$;
by (\ref{eq00s}) for $n_0$, $\phi$, $\emptyset$,
$\phi'$, $\tilde{f}_0(\bar{x})$, $\forall \bar{x}\, \phi'$, $\{\tilde{p}_\mbbm{i}(\bar{x})\gle \gu\}$, $S$,
$\{\tilde{p}_\mbbm{i}(\bar{x})\gle \gu\}$ (a copy), $S$ (a copy), $q=8$, $r=2$,
the space complexity of elementary operations at the stages $1,\dots,7$ is
in $O(((\sum_{i=1}^7 \#{\mc O}_i)+|\phi|^2)\cdot (\log (1+n_0)+\log |\phi|))\subseteq O(|\phi|^2\cdot (\log (1+n_0)+\log |\phi|))$;
the total number of all elementary operations of the translation of $T$ and $\phi$ to $S_T^\phi$ is in $O(|T|^2+|\phi|^2)$;
the total time complexity of the translation of $T$ and $\phi$ to $S_T^\phi$ is
in $O(|T|^2\cdot \log (1+n_0+|T|)+|\phi|^2\cdot (\log (1+n_0)+\log |\phi|))$; (iii) holds.

(iv) can be proved straightforwardly.

Thus, in all Cases 1--3, (i), (iii), (iv) hold.

Let $T\models \phi$.
Then there does not exist an interpretation ${\mf A}$ for ${\mc L}$, and ${\mf A}\models T$, ${\mf A}\not\models \phi$;
by (i), there does not exist an interpretation ${\mf A}'$ for ${\mc L}\cup \{\tilde{p}_\mbbm{j} \,|\, \mbbm{j}\in J_T^\phi\}$ and
${\mf A}'\models S_T^\phi$; $S_T^\phi$ is unsatisfiable.
Let $S_T^\phi$ be unsatisfiable.
Then there does not exist an interpretation ${\mf A}'$ for ${\mc L}\cup \{\tilde{p}_\mbbm{j} \,|\, \mbbm{j}\in J_T^\phi\}$ and
${\mf A}'\models S_T^\phi$;
by (i), there does not exist an interpretation ${\mf A}$ for ${\mc L}$, and ${\mf A}\models T$, ${\mf A}\not\models \phi$;
$T\models \phi$; (ii) holds.
\qed
\end{proof}

\section{Multi-step fuzzy inference}
\label{S4}

In this section, we provide some implementation of the Mamdani-Assilian type of fuzzy rules and inference in G\"{o}del logic.
We pose three fundamental problems: reachability, stability, and the existence of a $k$-cycle in multi-step fuzzy inference and 
show their reductions to certain deduction and unsatisfiability problems.
The implementation will be illustrated by an example.

Let $\mbb{U}$ be a non-empty set.
We call $\mbb{U}$ the universum.
A fuzzy set $A$ over $\mbb{U}$ is a mapping $A : \mbb{U}\longrightarrow [0,1]$.
We denote the set of all fuzzy sets over $\mbb{U}$ as ${\mc F}_\mbb{U}$.
Let $c\in [0,1]$ and $A_1, A_2\in {\mc F}_\mbb{U}$.
We define the height of $A_1$ as $\mi{height}(A_1)=\bigfvee_{u\in \mbb{U}} A_1(u)\in [0,1]$;
the cut $\mi{cut}(c,A_1)\in {\mc F}_\mbb{U}$ of $A_1$ by $c$ as 
$\mi{cut}(c,A_1) : \mbb{U}\longrightarrow [0,1],\ \mi{cut}(c,A_1)(u)=\mi{min}(c,A_1(u))$;
the union $A_1\cup A_2\in {\mc F}_\mbb{U}$ of $A_1$ and $A_2$ as 
$A_1\cup A_2 : \mbb{U}\longrightarrow [0,1],\ A_1\cup A_2(u)=\mi{max}(A_1(u),A_2(u))$;
the intersection $A_1\cap A_2\in {\mc F}_\mbb{U}$ of $A_1$ and $A_2$ as 
$A_1\cap A_2 : \mbb{U}\longrightarrow [0,1],\ A_1\cap A_2(u)=\mi{min}(A_1(u),A_2(u))$.
Let $\emptyset\neq \mbb{A}\subseteq_{\mc F} {\mc F}_\mbb{U}$.
Let $\mbb{X}$ be a non-empty finite set of variables having fuzzy sets over $\mbb{U}$ as values.
A fuzzy rule $r$ of the Mamdani-Assilian type is an expression of the form
$\mib{IF}\, X_0\, \mi{is}\, A_0\, \mi{and}\, \dots\, \mi{and}\, X_n\, \mi{is}\, A_n\, \mib{THEN}\, X\, \mi{is}\, A$,
$X_i, X\in \mbb{X}$, $A_i, A\in \mbb{A}$ \cite{MAAS75,Mam76}.
We say that $X_i$ are input variables, whereas $X$ is the output variable.
We denote $\mi{in}(r)=\{X_i \,|\, i\leq n\}\subseteq \mbb{X}$ and $\mi{out}(r)=X\in \mbb{X}$.
A fuzzy rule base is a non-empty finite set of fuzzy rules.
A fuzzy variable assignment is a mapping $\mbb{X}\longrightarrow {\mc F}_\mbb{U}$.
We denote the set of all fuzzy variable assignments as ${\mc S}_{\mc F}$.
Let $e\in {\mc S}_{\mc F}$. 
We define the value of $X$ with respect to $e$ and $r$ as 
$\|X\|_e^r=\mi{cut}(\bigfwedge_{i=0}^n \mi{height}(e(X_i)\cap A_i),A)\in {\mc F}_\mbb{U}$.
Let $B$ be a fuzzy rule base and $X\in \mbb{X}$.
We define the value of $X$ with respect to $e$ and $B$ as $\|X\|_e^B=\bigcup \{\|X\|_e^r \,|\, r\in B, \mi{out}(r)=X\}\in {\mc F}_\mbb{U}$.

As another step, we propose translation of fuzzy rules to formulae of G\"{o}del logic.
Assume that the universum $\mbb{U}$ is countable.
Notice that this restriction is reasonable. 
In many cases, it is sufficient to consider fuzzy sets over $\mbb{R}$ which are continuous; 
two such fuzzy sets are equal if their restrictions onto $\mbb{Q}$ are equal.
We shall assume a fresh constant symbol $\tilde{z}$, a fresh unary function symbol $\tilde{s}$, and 
two fresh binary function symbols $\mi{frac}$, $\mi{-frac}$ such that 
$\tilde{z}, \tilde{s}, \mi{frac}, \mi{-frac}\not\in \mi{Func}_{\mc L}\cup \{\tilde{f}_0\}$.
We denote $\tilde{\mbb{Z}}=\{\tilde{z},\tilde{s},\mi{frac},\mi{-frac}\}$.
Using $\tilde{z}$, $\tilde{s}$, and $\mi{(-)frac}$, we can build natural and rational numerals
for representation of natural and rational numbers, respectively.
Let $t\in \mi{GTerm}_{\tilde{\mbb{Z}}}$.
$t$ is a natural numeral iff $t=\tilde{s}^n(\tilde{z})$.
$t$ is a rational numeral iff either $t=\mi{frac}(\tilde{s}^m(\tilde{z}),\tilde{s}^n(\tilde{z}))$, $n>0$, or
$t=\mi{-frac}(\tilde{s}^m(\tilde{z}),\tilde{s}^n(\tilde{z}))$, $m, n>0$.
We shall assume a set of four fresh unary predicate symbols $\tilde{\mbb{D}}=\{\mi{nat},\mi{rat},\mi{time},\mi{uni}\}$ such that
$\tilde{\mbb{D}}\cap (\mi{Pred}_{\mc L}\cup \tilde{\mbb{P}})=\emptyset$.
We shall use these predicate symbols for axiomatisation of certain domain properties.
We shall assume a non-empty finite set of fresh unary predicate symbols 
$\tilde{\mbb{G}}=\{\tilde{G}_A \,|\, A\in \mbb{A}\}$ such that 
$\tilde{\mbb{G}}\cap (\mi{Pred}_{\mc L}\cup \tilde{\mbb{P}}\cup \tilde{\mbb{D}})=\emptyset$.
We shall assume a non-empty finite set of fresh binary predicate symbols 
$\tilde{\mbb{H}}=\{\tilde{H}_X^r \,|\, r\in B, X\in \mbb{X}, \mi{out}(r)=X\}\cup \{\tilde{H}_X \,|\, X\in \mbb{X}\}$ such that 
$\tilde{\mbb{H}}\cap (\mi{Pred}_{\mc L}\cup \tilde{\mbb{P}}\cup \tilde{\mbb{D}}\cup \tilde{\mbb{G}})=\emptyset$.
We put $C_{\mc L}=\{0,1\}\cup \bigcup \{A[\mbb{U}] \,|\, A\in \mbb{A}\}\cup \bigcup \{e(X)[\mbb{U}] \,|\, X\in \mbb{X}\}$; 
$\{0,1\}\subseteq C_{\mc L}\subseteq [0,1]$ is countable.
We shall assume a fixed first-order language ${\mc L}^*$ which is an expansion 
of ${\mc L}\cup \{\tilde{f}_0\}\cup \tilde{\mbb{Z}}\cup \tilde{\mbb{D}}\cup \tilde{\mbb{G}}\cup \tilde{\mbb{H}}$.
Note that $\mi{GTerm}_{{\mc L}^*}\neq \emptyset$; $\tilde{z}\in \mi{Func}_{{\mc L}^*}$.
We denote ${\mc K}=\{{\mc I} \,|\, {\mc I}\ \text{\it is an interpretation for}\ {\mc L}^*, {\mc U}_{\mc I}=\mi{GTerm}_{{\mc L}^*}\}$. 
We shall confine our considerations concerning logical consequence and satisfiability onto ${\mc K}$.
Let $\phi, \phi'\in \mi{Form}_{{\mc L}^*}$, $T, T'\subseteq \mi{Form}_{{\mc L}^*}$, 
$C, C'\in \mi{OrdCl}_{{\mc L}^*}$, $S, S'\subseteq \mi{OrdCl}_{{\mc L}^*}$.
Let $\varepsilon_1\in \{\phi,T,C,S\}$ and $\varepsilon_2\in \{\phi',T',C',S'\}$.
$\varepsilon_2$ is a logical consequence of $\varepsilon_1$ with respect to ${\mc K}$, in symbols $\varepsilon_1\models_{\mc K} \varepsilon_2$,
iff, for every interpretation ${\mc I}\in {\mc K}$, if ${\mc I}\models \varepsilon_1$, then ${\mc I}\models \varepsilon_2$.
$\varepsilon_1$ is satisfiable in ${\mc K}$ iff there exists a model ${\mc I}\in {\mc K}$ of $\varepsilon_1$.
We can axiomatise the domains of natural, rational numbers, 
a domain of time (a linear discrete time with the starting point $\tilde{z}$ and without an endpoint), and the universum $\mbb{U}$ as follows.
Let $\delta$ be a sequence of $\mi{Var}_{{\mc L}^*}$. 
We define $T_D=\{\mi{nat}(\tilde{z}), \mi{nat}(\tilde{s}(x))\leftrightarrow \mi{nat}(x), 
                 \mi{rat}(\mi{frac}(x,\tilde{s}(y)))\leftrightarrow \mi{nat}(x)\wedge \mi{nat}(y),
                 \mi{rat}(\mi{-frac}(\tilde{s}(x),\tilde{s}(y)))\leftrightarrow \mi{nat}(x)\wedge \mi{nat}(y)\}\cup
               \{\mi{nat}(f(\delta|_{\mi{ar}_{{\mc L}^*}(f)}))\geql \gz \,|\, f\in \mi{Func}_{{\mc L}^*}-\{\tilde{z},\tilde{s}\}\}\cup
               \{\mi{rat}(f(\delta|_{\mi{ar}_{{\mc L}^*}(f)}))\geql \gz \,|\, f\in \mi{Func}_{{\mc L}^*}-\{\mi{frac},\mi{-frac}\}\}\cup
               \{\mi{time}(x)\leftrightarrow \mi{nat}(x),\mi{uni}(x)\rightarrow \mi{rat}(x)\}\subseteq \mi{Form}_{{\mc L}^*}$. 
The first four formulae of $T_D$ obviously axiomatise the domains of natural and rational numbers by respective natural and rational numerals.
The next two axioms express some kind of the closed world assumption;
the domains of natural and rational numbers consist only of respective natural and rational numerals in every ${\mc I}\in {\mc K}$.
More precisely, we can prove that 
for all $t\in {\mc U}_{\mc I}$, if ${\mc I}\models \mi{nat}(t)$, then $t$ is a natural numeral.
The proof is by straightforward induction on $|t|$.
Analogously for rational numerals, for all $t\in {\mc U}_{\mc I}$, if ${\mc I}\models \mi{rat}(t)$, then $t$ is a rational numeral.
The proof is by immediate case analysis.  
The last two axioms define the time domain as the domain of natural numbers and the universum $\mbb{U}$ as a subdomain of rational numbers.
We have that $\mbb{U}$ is countable; therefore, its axiomatisation could just be some subdomain of natural numbers.
However, for application purposes, a subdomain of rational numbers is much more convenient. 
Hence, there exists an injection $\gamma : \mbb{U}\longrightarrow \mbb{Q}$.
Let $t\in \mi{GTerm}_{\tilde{\mbb{Z}}}$ be a rational numeral.
If $t=\mi{frac}(\tilde{s}^m(\tilde{z}),\tilde{s}^n(\tilde{z}))$, $n>0$, 
then we define the value of $t$ as $\|t\|=\dfrac{m}{n}\in \mbb{Q}$.
If $t=\mi{-frac}(\tilde{s}^m(\tilde{z}),\tilde{s}^n(\tilde{z}))$, $m, n>0$,
then we define the value of $t$ as $\|t\|=-\dfrac{m}{n}\in \mbb{Q}$.
So, the universum $\mbb{U}$ can be axiomatised as follows.
Let $\tilde{\mbb{U}}\subseteq \{t \,|\, t\in \mi{GTerm}_{\tilde{\mbb{Z}}}\ \text{\it is a rational numeral},\ \|t\|\in \gamma[\mbb{U}]\}$
such that $\{\|\tilde{u}\| \,|\, \tilde{u}\in \tilde{\mbb{U}}\}=\gamma[\mbb{U}]$.
We define $S_U=\{\mi{uni}(\tilde{u})\geql \gu \,|\, \tilde{u}\in \tilde{\mbb{U}}\}\cup 
               \{\mi{uni}(t)\geql \gz \,|\, t\in \mi{GTerm}_{{\mc L}^*}-\tilde{\mbb{U}}\}\subseteq \mi{OrdCl}_{{\mc L}^*}$.
Note that the universum $\mbb{U}$ is interpreted as $\tilde{\mbb{U}}$ in every ${\mc I}\in {\mc K}$;
for all $t\in {\mc U}_{\mc I}$, ${\mc I}\models \mi{uni}(t)$ if and only if $t\in \tilde{\mbb{U}}$. 
Let $\tilde{u}\in \tilde{\mbb{U}}$.
We denote $\langle\tilde{u}\rangle=\gamma^{-1}(\|\tilde{u}\|)\in \mbb{U}$; 
$\langle\tilde{u}\rangle$ denotes a unique element $u\in \mbb{U}$ such that $\gamma(u)=\|\tilde{u}\|\in \mbb{Q}$.
The translation of $\mbb{A}$ is defined as 
$S_\mbb{A}=\{\tilde{G}_A(\tilde{u})\geql \overline{A(\langle\tilde{u}\rangle)} \,|\, A\in \mbb{A}, \tilde{u}\in \tilde{\mbb{U}}\}\subseteq 
           \mi{OrdCl}_{{\mc L}\cup \tilde{\mbb{Z}}\cup \tilde{\mbb{G}}}$. 
The translation of $e$ is defined as
$S_e=\{\tilde{H}_X(\tau,\tilde{u})\geql \overline{e(X)(\langle\tilde{u}\rangle)} \,|\, X\in \mbb{X}, \tilde{u}\in \tilde{\mbb{U}}\}\subseteq
     \mi{OrdCl}_{{\mc L}\cup \tilde{\mbb{Z}}\cup \tilde{\mbb{H}}}$.
The translation of $r$ is defined as
$\phi_r(\tau,y)=\big(\mi{time}(\tau)\wedge \mi{uni}(y)\rightarrow
                     \big(\tilde{H}_X^r(\tilde{s}(\tau),y)\geql 
                          ((\bigwedge_{i=0}^n \exists x\, (\mi{uni}(x)\wedge \tilde{H}_{X_i}(\tau,x)\wedge \tilde{G}_{A_i}(x)))\wedge 
                           \tilde{G}_A(y))\big)\big)\in
                \mi{Form}_{\tilde{\mbb{Z}}\cup \tilde{\mbb{D}}\cup \tilde{\mbb{G}}\cup \tilde{\mbb{H}}}$.
The translation of $B$ is defined as 
$T_B=\{\phi_r(\tau,y) \,|\, r\in B\}\cup 
     \big\{\mi{time}(\tau)\wedge \mi{uni}(y)\rightarrow \big(\tilde{H}_X(\tilde{s}(\tau),y)\geql 
                                                             \bigvee_{r\in B, \mi{out}(r)=X} \tilde{H}_X^r(\tilde{s}(\tau),y)\big) \,|\, 
           X\in \mbb{X}\big\}\subseteq_{\mc F} 
     \mi{Form}_{\tilde{\mbb{Z}}\cup \tilde{\mbb{D}}\cup \tilde{\mbb{G}}\cup \tilde{\mbb{H}}}$.

Let ${\mc D}=e_0,\dots,e_n$, $e_i\in {\mc S}_{\mc F}$.
${\mc D}$ is a fuzzy derivation of $e_n$ from $e_0$ using $B$ iff, 
for all $1\leq i\leq n$, $e_i=\{(X,\|X\|_{e_{i-1}}^B) \,|\, X\in \mbb{X}\}$.

\begin{lemma}
\label{le2}
Let ${\mc D}=e_0,\dots,e_\eta$ be a fuzzy derivation.
$T_D\cup S_U\cup S_\mbb{A}\cup T_B\cup S_{e_0}(\tau/\tilde{z})\models_{\mc K} S_{e_\eta}(\tau/\tilde{s}^\eta(\tilde{z}))$.
\end{lemma}

\begin{proof}
Let ${\mf A}\in {\mc K}$ such that ${\mf A}\models T_D\cup S_U\cup S_\mbb{A}\cup T_B\cup S_{e_0}(\tau/\tilde{z})\subseteq \mi{Form}_{{\mc L}^*}$.
We show that for all $\kappa\leq \eta$, ${\mf A}\models S_{e_\kappa}(\tau/\tilde{s}^\kappa(\tilde{z}))\subseteq \mi{OrdCl}_{{\mc L}^*}$.
We proceed by induction on $\kappa\leq \eta$.

Case 1 (the base case):
$\kappa=0$.
${\mf A}\models S_{e_0}(\tau/\tilde{z})$.
The statement holds.

Case 2 (the induction case):
$0<\kappa\leq \eta$.
By induction hypothesis for $\kappa-1$, ${\mf A}\models S_{e_{\kappa-1}}(\tau/\tilde{s}^{\kappa-1}(\tilde{z}))$.
Then $e_\kappa=\{(X,\|X\|_{e_{\kappa-1}}^B) \,|\, X\in \mbb{X}\}$,
for all $X\in \mbb{X}$,
$e_\kappa(X)=\|X\|_{e_{\kappa-1}}^B=\bigcup \{\|X\|_{e_{\kappa-1}}^r \,|\, r\in B, \mi{out}(r)=X\}$,
${\mf A}\models \mi{time}(\tau)\wedge \mi{uni}(y)\rightarrow 
                \big(\tilde{H}_X(\tilde{s}(\tau),y)\geql \bigvee_{r\in B, \mi{out}(r)=X} \tilde{H}_X^r(\tilde{s}(\tau),y)\big)\in T_B$,
for all $\tilde{u}\in \tilde{\mbb{U}}$,
for all $r\in B$ and $\mi{out}(r)=X$,
$\|X\|_{e_{\kappa-1}}^r=\mi{cut}(\bigfwedge_{i=0}^n \mi{height}(e_{\kappa-1}(X_i)\cap A_i),A)$, $X_i\in \mbb{X}$, $A_i, A\in \mbb{A}$,
${\mf A}\models \phi_r(\tau,y)\in T_B$,                                                                                            
${\mf A}\models \phi_r(\tilde{s}^{\kappa-1}(\tilde{z}),\tilde{u})=
                \big(\mi{time}(\tilde{s}^{\kappa-1}(\tilde{z}))\wedge \mi{uni}(\tilde{u})\rightarrow
                \big(\tilde{H}_X^r(\tilde{s}^\kappa(\tilde{z}),\tilde{u})\geql
                     ((\bigwedge_{i=0}^n \exists x\, (\mi{uni}(x)\wedge \tilde{H}_{X_i}(\tilde{s}^{\kappa-1}(\tilde{z}),x)\wedge 
                                                      \tilde{G}_{A_i}(x)))\wedge 
                      \tilde{G}_A(\tilde{u}))\big)\big)$,
${\mf A}\models T_D\models \mi{time}(\tilde{s}^{\kappa-1}(\tilde{z}))$, 
${\mf A}\models S_U\models \mi{uni}(\tilde{u})$,
${\mf A}\models \tilde{H}_X^r(\tilde{s}^\kappa(\tilde{z}),\tilde{u})\geql
                ((\bigwedge_{i=0}^n \exists x\, (\mi{uni}(x)\wedge \tilde{H}_{X_i}(\tilde{s}^{\kappa-1}(\tilde{z}),x)\wedge 
                                                 \tilde{G}_{A_i}(x)))\wedge
                 \tilde{G}_A(\tilde{u}))$,
for all $i\leq n$ and $\tilde{v}\in \tilde{\mbb{U}}$,
${\mf A}\models \tilde{H}_{X_i}(\tilde{s}^{\kappa-1}(\tilde{z}),\tilde{v})\geql \overline{e_{\kappa-1}(X_i)(\langle\tilde{v}\rangle)}\in 
                S_{e_{\kappa-1}}(\tau/\tilde{s}^{\kappa-1}(\tilde{z}))$,
${\mf A}\models \tilde{G}_{A_i}(\tilde{v})\geql \overline{A_i(\langle\tilde{v}\rangle)}\in S_\mbb{A}$,
${\mf A}\models \tilde{G}_A(\tilde{u})\geql \overline{A(\langle\tilde{u}\rangle)}\in S_\mbb{A}$,
${\mf A}\models \tilde{H}_X^r(\tilde{s}^\kappa(\tilde{z}),\tilde{u})\geql \overline{\|X\|_{e_{\kappa-1}}^r(\langle\tilde{u}\rangle)}$;
${\mf A}\models \mi{time}(\tilde{s}^{\kappa-1}(\tilde{z}))\wedge \mi{uni}(\tilde{u})\rightarrow 
                \big(\tilde{H}_X(\tilde{s}^\kappa(\tilde{z}),\tilde{u})\geql
                     \bigvee_{r\in B, \mi{out}(r)=X} \tilde{H}_X^r(\tilde{s}^\kappa(\tilde{z}),\tilde{u})\big)$,
${\mf A}\models \tilde{H}_X(\tilde{s}^\kappa(\tilde{z}),\tilde{u})\geql 
                \bigvee_{r\in B, \mi{out}(r)=X} \tilde{H}_X^r(\tilde{s}^\kappa(\tilde{z}),\tilde{u})$,
${\mf A}\models \tilde{H}_X(\tilde{s}^\kappa(\tilde{z}),\tilde{u})\geql
                \bigvee_{r\in B, \mi{out}(r)=X} \overline{\|X\|_{e_{\kappa-1}}^r(\langle\tilde{u}\rangle)}$,                       
${\mf A}\models \tilde{H}_X(\tilde{s}^\kappa(\tilde{z}),\tilde{u})\geql \overline{\|X\|_{e_{\kappa-1}}^B(\langle\tilde{u}\rangle)}$,
${\mf A}\models \tilde{H}_X(\tilde{s}^\kappa(\tilde{z}),\tilde{u})\geql \overline{e_\kappa(X)(\langle\tilde{u}\rangle)}$;
${\mf A}\models S_{e_\kappa}(\tau/\tilde{s}^\kappa(\tilde{z}))$.
The statement holds.

So, in both Cases 1 and 2, the statement holds.
The induction is completed.

For $\kappa=\eta$, ${\mf A}\models S_{e_\eta}(\tau/\tilde{s}^\eta(\tilde{z}))$;
$T_D\cup S_U\cup S_\mbb{A}\cup T_B\cup S_{e_0}(\tau/\tilde{z})\models_{\mc K} S_{e_\eta}(\tau/\tilde{s}^\eta(\tilde{z}))$.
\qed 
\end{proof}

We are in position to formulate the reachability, stability, and the existence of a $k$-cycle problems.
Let $X_0,\dots,X_n\in \mbb{X}$, $A_0,\dots,A_n\in \mbb{A}$, ${\mc D}=e_0,\dots,e_\eta$ be a fuzzy derivation, $k\geq 1$.
$(A_0,\dots,A_n)$ is reachable in ${\mc D}$ iff there exists $\kappa\leq \eta$ such that for all $i\leq n$, $e_\kappa(X_i)=A_i$.  
${\mc D}$ is stable iff there exists $\kappa<\eta$ such that $e_\kappa=e_{\kappa+1}$.
There exists a $k$-cycle in ${\mc D}$ iff there exists $\kappa<\eta$ such that $\kappa+k\leq \eta$ and $e_\kappa=e_{\kappa+k}$.
Obviously, the stability problem is the existence of a $1$-cycle problem. 
The formulations of the problems can be translated to formulae of G\"{o}del logic.
$(A_0,\dots,A_n)$ is reachable in ${\mc D}$ as
$\phi_r=\exists \tau\, (\mi{time}(\tau)\wedge 
                        \bigwedge_{i=0}^n \forall x\, (\mi{uni}(x)\rightarrow \tilde{H}_{X_i}(\tau,x)\geql \tilde{G}_{A_i}(x)))$.
${\mc D}$ is stable as 
$\phi_s=\exists \tau\, (\mi{time}(\tau)\wedge 
                        \bigwedge_{X\in \mbb{X}} \forall x\, (\mi{uni}(x)\rightarrow \tilde{H}_X(\tau,x)\geql \tilde{H}_X(\tilde{s}(\tau),x)))$.
There exists a $k$-cycle in ${\mc D}$ as
$\phi_{k-c}=\exists \tau\, (\mi{time}(\tau)\wedge
                            \bigwedge_{X\in \mbb{X}} \forall x\, (\mi{uni}(x)\rightarrow 
                                                                  \tilde{H}_X(\tau,x)\geql \tilde{H}_X(\tilde{s}^k(\tau),x)))$.
The problems can be reduced to deduction problems with respect to ${\mc K}$ as follows.

\begin{theorem}
\label{T2}
Let $X_0,\dots,X_n\in \mbb{X}$, $A_0,\dots,A_n\in \mbb{A}$, ${\mc D}=e_0,\dots,e_\eta$ be a fuzzy derivation, $k\geq 1$.
\begin{enumerate}[\rm (i)]
\item
$(A_0,\dots,A_n)$ is reachable in ${\mc D}$ if and only if 
$T_D\cup S_U\cup S_\mbb{A}\cup T_B\cup S_{e_0}(\tau/\tilde{z})\models_{\mc K} \phi_r$.
\item
${\mc D}$ is stable if and only if 
$T_D\cup S_U\cup S_\mbb{A}\cup T_B\cup S_{e_0}(\tau/\tilde{z})\models_{\mc K} \phi_s$.
\item
There exists a $k$-cycle in ${\mc D}$ if and only if 
$T_D\cup S_U\cup S_\mbb{A}\cup T_B\cup S_{e_0}(\tau/\tilde{z})\models_{\mc K} \phi_{k-c}$.
\end{enumerate}
\end{theorem}

\begin{proof}
By Lemma \ref{le2}, for all $\kappa\leq \eta$, 
$T_D\cup S_U\cup S_\mbb{A}\cup T_B\cup S_{e_0}(\tau/\tilde{z})\models_{\mc K} S_{e_\kappa}(\tau/\tilde{s}^\kappa(\tilde{z}))$.
Then (i--iii) can be proved straightforwardly.
\qed 
\end{proof}

Subsequently, the deduction problems can be reduced to unsatisfiability problems in ${\mc K}$ using Theorem \ref{T1}.

We illustrate the proposed translation by the following example on a fuzzy inference system modelling a simple thermodynamic system.
We shall model an engine with inner combustion and cooling medium.
We shall consider physical quantities such as the temperature (t), density (d) of the cooling medium, and the rotation (r) of the engine 
together with their first derivatives.
For simplicity, we put the universum $\mbb{U}=\{0,1,2,3,4\}$ (in this case, $\gamma$ is just the identity mapping on $\mbb{U}$).
For every physical quantity $k$, we define fuzzy sets $\mi{low}_k$, $\mi{medium}_k$, $\mi{high}_k$, and
for its derivative $\dot{k}$, fuzzy sets $\mi{negative}_{\dot{k}}$, $\mi{zero}_{\dot{k}}$, $\mi{positive}_{\dot{k}}$ in Table~\ref{tab4}.
Hence, we put 
$\mbb{A}=\bigcup_{k\in \{t,d,r\}} \{\mi{low}_k,\mi{medium}_k,\mi{high}_k,\mi{negative}_{\dot{k}},\mi{zero}_{\dot{k}},              \linebreak[4]
                                    \mi{positive}_{\dot{k}}\}$ and 
$\mbb{X}=\{X_i \,|\, i\leq 5\}$ where the variables $X_0$ and $X_3$ correspond to the temperature and its derivative, 
                                      $X_1$ and $X_4$ to the density and its derivative,
                                      $X_2$ and $X_5$ to the rotation and its derivative, respectively.
In Table \ref{tab4}, an underlying fuzzy rule base $B$ is devised. 
\begin{table*}[p]
\vspace{-6mm}
\caption{Fuzzy inference system}\label{tab4}
\vspace{-6mm}
\centering
%\hspace*{-5mm}
\begin{minipage}[t]{\linewidth-15mm}
\footnotesize
\begin{IEEEeqnarray*}{LLLL}
\hline \hline \\[1mm]
\IEEEeqnarraymulticol{4}{c}{\text{\bf \normalsize{Fuzzy sets}}} \\[2mm]
\text{\bf \small{Quantity:}}\quad & \mi{low}_k=\Big\{\frac{1}{0},\frac{0.5}{1},\frac{0}{2},\frac{0}{3},\frac{0}{4}\Big\} \quad
                                  & \mi{medium}_k=\Big\{\frac{0}{0},\frac{0.5}{1},\frac{1}{2},\frac{0.5}{3},\frac{0}{4}\Big\} \quad
                                  & \mi{high}_k=\Big\{\frac{0}{0},\frac{0}{1},\frac{0}{2},\frac{0.5}{3},\frac{1}{4}\Big\} \\[1mm]
\IEEEeqnarraymulticol{4}{l}{\text{$k$ stands for $t$ -- temperature, $d$ -- density, $r$ -- rotation}} \\[2mm] 
\text{\bf \small{Derivative:}}\quad & \mi{negative}_{\dot{k}}=\Big\{\frac{1}{0},\frac{0.5}{1},\frac{0}{2},\frac{0}{3},\frac{0}{4}\Big\} \quad
                                    & \mi{zero}_{\dot{k}}=\Big\{\frac{0}{0},\frac{0.5}{1},\frac{1}{2},\frac{0.5}{3},\frac{0}{4}\Big\} \quad
                                    & \mi{positive}_{\dot{k}}=\Big\{\frac{0}{0},\frac{0}{1},\frac{0}{2},\frac{0.5}{3},\frac{1}{4}\Big\} \\[1mm]  
\IEEEeqnarraymulticol{4}{l}{\text{$\dot{k}$ stands for $\dot{t}$ -- derivative of temperature, 
                                                       $\dot{d}$ -- derivative of density, 
                                                       $\dot{r}$ -- derivative of rotation}} 
\end{IEEEeqnarray*}
\vspace{-2mm}
\begin{IEEEeqnarray*}{LL}
\IEEEeqnarraymulticol{2}{c}{\text{\bf \normalsize{Fuzzy rule base}}} \\[2mm]
\IEEEeqnarraymulticol{2}{c}{\text{\small{$B=\big\{R_i \,|\, i\in \{1,\dots,29\}\cup \{30a,\dots,32a,30b,\dots,32b\}\cup
                                                                 \{33,\dots,47\}\cup \{48a,\dots,50a,48b,\dots,50b\}\cup \{51,\dots,56\}
                                            \big\}$}}} \\[4mm]
\IEEEeqnarraymulticol{2}{l}{\text{\bf \small{Inversion rules:}}} \\[0mm]
R_1:    & 
\mib{IF}\, X_0\, \mi{is}\, \mi{low}_t\, \mib{THEN}\, X_1\, \mi{is}\, \mi{high}_d \\[0mm]
R_2:    &
\mib{IF}\, X_0\, \mi{is}\, \mi{medium}_t\, \mib{THEN}\, X_1\, \mi{is}\, \mi{medium}_d \\[0mm]
R_3:    &
\mib{IF}\, X_0\, \mi{is}\, \mi{high}_t\, \mib{THEN}\, X_1\, \mi{is}\, \mi{low}_d \\[-16.85mm]
\IEEEeqnarraymulticol{2}{l}{\hspace{80mm}\text{\bf \small{Thermic rules:}}} \\[0mm]
\IEEEeqnarraymulticol{2}{l}{\hspace{80mm}
R_4:    \ 
\mib{IF}\, X_2\, \mi{is}\, \mi{low}_r\, \mib{THEN}\, X_3\, \mi{is}\, \mi{negative}_{\dot{t}}} \\[0mm]
\IEEEeqnarraymulticol{2}{l}{\hspace{80mm}
R_5:    \ 
\mib{IF}\, X_2\, \mi{is}\, \mi{medium}_r\, \mib{THEN}\, X_3\, \mi{is}\, \mi{zero}_{\dot{t}}} \\[0mm]
\IEEEeqnarraymulticol{2}{l}{\hspace{80mm}
R_6:    \
\mib{IF}\, X_2\, \mi{is}\, \mi{high}_r\, \mib{THEN}\, X_3\, \mi{is}\, \mi{positive}_{\dot{t}}} \\[2mm] 
\IEEEeqnarraymulticol{2}{l}{\text{\bf \small{Friction rules:}}} \\[0mm]
R_7:    & 
\mib{IF}\, X_1\, \mi{is}\, \mi{medium}_d\, \mi{and}\, X_2\, \mi{is}\, \mi{high}_r\, \mib{THEN}\, X_5\, \mi{is}\, \mi{zero}_{\dot{r}} \\[0mm]
R_8:    & 
\mib{IF}\, X_1\, \mi{is}\, \mi{high}_d\, \mi{and}\, X_2\, \mi{is}\, \mi{high}_r\, \mib{THEN}\, X_5\, \mi{is}\, \mi{negative}_{\dot{r}} \\[2mm]
\IEEEeqnarraymulticol{2}{l}{\text{\bf \small{Dynamic rules:}}} \\[0mm]
R_{9-11}:    & 
\mib{IF}\, X_i\, \mi{is}\, \mi{medium}_k\, \mi{and}\, X_{i+3}\, \mi{is}\, \mi{negative}_{\dot{k}}\,
\mib{THEN}\, X_i\, \mi{is}\, \mi{low}_k,\ (i,k)\in \{(0,t),(1,d),(2,r)\} \\[0mm]
R_{12-14}: & 
\mib{IF}\, X_i\, \mi{is}\, \mi{high}_k\, \mi{and}\, X_{i+3}\, \mi{is}\, \mi{negative}_{\dot{k}}\,
\mib{THEN}\, X_i\, \mi{is}\, \mi{medium}_k,\ (i,k)\in \{(0,t),(1,d),(2,r)\} \\[0mm]
R_{15-23}: &
\mib{IF}\, X_i\, \mi{is}\, \mi{W}_k\, \mi{and}\, X_{i+3}\, \mi{is}\, \mi{zero}_{\dot{k}}\,
\mib{THEN}\, X_i\, \mi{is}\, \mi{W}_k,\ (i,k)\in \{(0,t),(1,d),(2,r)\}, \mi{W}\in \{\mi{low},\mi{medium},\mi{high}\} \\[0mm]
R_{24-26}: &
\mib{IF}\, X_i\, \mi{is}\, \mi{low}_k\, \mi{and}\, X_{i+3}\, \mi{is}\, \mi{positive}_{\dot{k}}\,
\mib{THEN}\, X_i\, \mi{is}\, \mi{medium}_k,\ (i,k)\in \{(0,t),(1,d),(2,r)\} \\[0mm]
R_{27-29}: &
\mib{IF}\, X_i\, \mi{is}\, \mi{medium}_k\, \mi{and}\, X_{i+3}\, \mi{is}\, \mi{positive}_{\dot{k}}\,  
\mib{THEN}\, X_i\, \mi{is}\, \mi{high}_k,\ (i,k)\in \{(0,t),(1,d),(2,r)\} \\[2mm]
\IEEEeqnarraymulticol{2}{l}{\text{\bf \small{Limit rules:}}} \\[0mm]
R_{30a-32a}: &
\mib{IF}\, X_i\, \mi{is}\, \mi{low}_k\, \mi{and}\, X_{i+3}\, \mi{is}\, \mi{negative}_{\dot{k}}\, 
\mib{THEN}\, X_i\, \mi{is}\, \mi{low}_k,\ (i,k)\in \{(0,t),(1,d),(2,r)\} \\[0mm]
R_{30b-32b}: &
\mib{IF}\, X_i\, \mi{is}\, \mi{low}_k\, \mi{and}\, X_{i+3}\, \mi{is}\, \mi{negative}_{\dot{k}}\,
\mib{THEN}\, X_{i+3}\, \mi{is}\, \mi{positive}_{\dot{k}},\ (i,k)\in \{(0,t),(1,d),(2,r)\} \\[0mm]
R_{33-38}: &
\mib{IF}\, X_i\, \mi{is}\, \mi{low}_k\, \mi{and}\, X_{i+3}\, \mi{is}\, \mi{W}_{\dot{k}}\,
\mib{THEN}\, X_{i+3}\, \mi{is}\, \mi{W}_{\dot{k}},\ (i,k)\in \{(0,t),(1,d),(2,r)\}, \mi{W}\in \{\mi{zero},\mi{positive}\} \\[0mm]
R_{39-47}: &
\mib{IF}\, X_i\, \mi{is}\, \mi{medium}_k\, \mi{and}\, X_{i+3}\, \mi{is}\, \mi{W}_{\dot{k}}\,       
\mib{THEN}\, X_{i+3}\, \mi{is}\, \mi{W}_{\dot{k}},\ (i,k)\in \{(0,t),(1,d),(2,r)\}, \mi{W}\in \{\mi{negative},\mi{zero},\mi{positive}\} \\[0mm]
R_{48a-50a}: &
\mib{IF}\, X_i\, \mi{is}\, \mi{high}_k\, \mi{and}\, X_{i+3}\, \mi{is}\, \mi{positive}_{\dot{k}}\,
\mib{THEN}\, X_i\, \mi{is}\, \mi{high}_k,\ (i,k)\in \{(0,t),(1,d),(2,r)\} \\[0mm]
R_{48b-50b}: &
\mib{IF}\, X_i\, \mi{is}\, \mi{high}_k\, \mi{and}\, X_{i+3}\, \mi{is}\, \mi{positive}_{\dot{k}}\,
\mib{THEN}\, X_{i+3}\, \mi{is}\, \mi{negative}_{\dot{k}},\ (i,k)\in \{(0,t),(1,d),(2,r)\} \\[0mm]
R_{51-56}: &
\mib{IF}\, X_i\, \mi{is}\, \mi{high}_k\, \mi{and}\, X_{i+3}\, \mi{is}\, \mi{W}_{\dot{k}}\,
\mib{THEN}\, X_{i+3}\, \mi{is}\, \mi{W}_{\dot{k}},\ (i,k)\in \{(0,t),(1,d),(2,r)\}, \mi{W}\in \{\mi{negative},\mi{zero}\} \\[1mm]
\hline \hline 
\end{IEEEeqnarray*}
\end{minipage}  
\vspace{-2mm}
\end{table*}
We put $\tilde{\mbb{U}}=\{\mi{frac}(\tilde{z},\tilde{s}(\tilde{z})),\mi{frac}(\tilde{s}(\tilde{z}),\tilde{s}(\tilde{z})),
                          \mi{frac}(\tilde{s}^2(\tilde{z}),\tilde{s}(\tilde{z})),                                                  \linebreak[4]      
                          \mi{frac}(\tilde{s}^3(\tilde{z}),\tilde{s}(\tilde{z})),
                          \mi{frac}(\tilde{s}^4(\tilde{z}),\tilde{s}(\tilde{z}))\}$ 
(for simplicity, we denote $\mi{frac}(\tilde{s}^n(\tilde{z}),\tilde{s}(\tilde{z}))$ as $\tilde{n}$) and
$S_U=\{\mi{uni}(\tilde{u})\geql \gu \,|\, \tilde{u}\in \tilde{\mbb{U}}\}\cup 
     \{\mi{uni}(t)\geql \gz \,|\, t\in \mi{GTerm}_{{\mc L}^*}-\tilde{\mbb{U}}\}$.
Subsequently, we can translate the proposed fuzzy sets and fuzzy rules 
to obtain $S_\mbb{A}$ (Table~\ref{tab5}) and $T_B$ (Tables~\ref{tab55}--\ref{tab555}).
\begin{table*}[p]
\vspace{-6mm}
\caption{Translation of the fuzzy sets}\label{tab5}
\vspace{-6mm}
\centering
\begin{minipage}[t]{\linewidth-15mm}
\scriptsize
\begin{IEEEeqnarray*}{LLL}
\hline \hline \\[1mm]
\scaleto{S_\mbb{A}}{10pt}= &
\scaleto{\bigcup_{k\in \{t,d,r\}}}{27pt} \Big\{ & 
\tilde{G}_{\mi{low}_k}(\tilde{0})\geql \bar{1}, 
\tilde{G}_{\mi{low}_k}(\tilde{1})\geql \overline{0.5},
\tilde{G}_{\mi{low}_k}(\tilde{2})\geql \bar{0},
\tilde{G}_{\mi{low}_k}(\tilde{3})\geql \bar{0},
\tilde{G}_{\mi{low}_k}(\tilde{4})\geql \bar{0}, \\[0mm]
& &
\tilde{G}_{\mi{medium}_k}(\tilde{0})\geql \bar{0},
\tilde{G}_{\mi{medium}_k}(\tilde{1})\geql \overline{0.5},
\tilde{G}_{\mi{medium}_k}(\tilde{2})\geql \bar{1},
\tilde{G}_{\mi{medium}_k}(\tilde{3})\geql \overline{0.5},
\tilde{G}_{\mi{medium}_k}(\tilde{4})\geql \bar{0}, \\[0mm]
& &
\tilde{G}_{\mi{high}_k}(\tilde{0})\geql \bar{0},
\tilde{G}_{\mi{high}_k}(\tilde{1})\geql \bar{0},
\tilde{G}_{\mi{high}_k}(\tilde{2})\geql \bar{0},
\tilde{G}_{\mi{high}_k}(\tilde{3})\geql \overline{0.5},
\tilde{G}_{\mi{high}_k}(\tilde{4})\geql \bar{1}\Big\}\operatorname{\scaleto{\cup}{7pt}} \\[2mm]
& 
\scaleto{\bigcup_{k\in \{t,d,r\}}}{27pt} \Big\{ &
\tilde{G}_{\mi{negative}_{\dot{k}}}(\tilde{0})\geql \bar{1},
\tilde{G}_{\mi{negative}_{\dot{k}}}(\tilde{1})\geql \overline{0.5},
\tilde{G}_{\mi{negative}_{\dot{k}}}(\tilde{2})\geql \bar{0},  
\tilde{G}_{\mi{negative}_{\dot{k}}}(\tilde{3})\geql \bar{0},  
\tilde{G}_{\mi{negative}_{\dot{k}}}(\tilde{4})\geql \bar{0}, \\[0mm]
& &
\tilde{G}_{\mi{zero}_{\dot{k}}}(\tilde{0})\geql \bar{0},
\tilde{G}_{\mi{zero}_{\dot{k}}}(\tilde{1})\geql \overline{0.5},
\tilde{G}_{\mi{zero}_{\dot{k}}}(\tilde{2})\geql \bar{1},  
\tilde{G}_{\mi{zero}_{\dot{k}}}(\tilde{3})\geql \overline{0.5},
\tilde{G}_{\mi{zero}_{\dot{k}}}(\tilde{4})\geql \bar{0}, \\[0mm]
& &
\tilde{G}_{\mi{positive}_{\dot{k}}}(\tilde{0})\geql \bar{0},
\tilde{G}_{\mi{positive}_{\dot{k}}}(\tilde{1})\geql \bar{0},
\tilde{G}_{\mi{positive}_{\dot{k}}}(\tilde{2})\geql \bar{0},
\tilde{G}_{\mi{positive}_{\dot{k}}}(\tilde{3})\geql \overline{0.5},
\tilde{G}_{\mi{positive}_{\dot{k}}}(\tilde{4})\geql \bar{1}\Big\} \\[1mm]
\hline \hline 
\end{IEEEeqnarray*}
\end{minipage}
\vspace{-2mm}
\end{table*}
\begin{table*}[p]
\vspace{-6mm}
\caption{Translation of the fuzzy rule base $B$}\label{tab55}
\vspace{-6mm}
\centering
\hspace*{-2.5mm}
\begin{minipage}[t]{\linewidth+5mm}
\scriptsize
\begin{IEEEeqnarray*}{L}
\hline \hline \\[1mm]
\scaleto{T_B=\big\{\phi_i(\tau,y) \,|\, i\in \{1,\dots,29\}\cup \{30a,\dots,32a,30b,\dots,32b\}\cup 
                                             \{33,\dots,47\}\cup \{48a,\dots,50a,48b,\dots,50b\}\cup \{51,\dots,56\}\big\}\cup \mbox{}}{10pt} \\[0mm]
\phantom{\scaleto{T_B=\vphantom{\big\{}}{10pt}}
\scaleto{\big\{\text{\bf Aggregation rules}\big\}}{10pt} \\[2mm]
\text{\bf \small{Inversion rules:}} \\[0mm]
\phi_1(\tau,y)=\big(\mi{time}(\tau)\wedge \mi{uni}(y)\rightarrow 
                    \big(\tilde{H}_{X_1}^1(\tilde{s}(\tau),y)\geql
                         (\exists x\, (\mi{uni}(x)\wedge \tilde{H}_{X_0}(\tau,x)\wedge \tilde{G}_{\mi{low}_t}(x))\wedge \tilde{G}_{\mi{high}_d}(y))\big)\big) \\[0mm]
\phi_2(\tau,y)=\big(\mi{time}(\tau)\wedge \mi{uni}(y)\rightarrow 
                    \big(\tilde{H}_{X_1}^2(\tilde{s}(\tau),y)\geql
                         (\exists x\, (\mi{uni}(x)\wedge \tilde{H}_{X_0}(\tau,x)\wedge \tilde{G}_{\mi{medium}_t}(x))\wedge \tilde{G}_{\mi{medium}_d}(y))\big)\big) \\[0mm]
\phi_3(\tau,y)=\big(\mi{time}(\tau)\wedge \mi{uni}(y)\rightarrow 
                    \big(\tilde{H}_{X_1}^3(\tilde{s}(\tau),y)\geql
                         (\exists x\, (\mi{uni}(x)\wedge \tilde{H}_{X_0}(\tau,x)\wedge \tilde{G}_{\mi{high}_t}(x))\wedge \tilde{G}_{\mi{low}_d}(y))\big)\big) \\[2mm]
\text{\bf \small{Thermic rules:}} \\[0mm]
\phi_4(\tau,y)=\big(\mi{time}(\tau)\wedge \mi{uni}(y)\rightarrow 
                    \big(\tilde{H}_{X_3}^4(\tilde{s}(\tau),y)\geql
                         (\exists x\, (\mi{uni}(x)\wedge \tilde{H}_{X_2}(\tau,x)\wedge \tilde{G}_{\mi{low}_r}(x))\wedge \tilde{G}_{\mi{negative}_{\dot{t}}}(y))\big)\big) \\[0mm]
\phi_5(\tau,y)=\big(\mi{time}(\tau)\wedge \mi{uni}(y)\rightarrow 
                    \big(\tilde{H}_{X_3}^5(\tilde{s}(\tau),y)\geql
                         (\exists x\, (\mi{uni}(x)\wedge \tilde{H}_{X_2}(\tau,x)\wedge \tilde{G}_{\mi{medium}_r}(x))\wedge \tilde{G}_{\mi{zero}_{\dot{t}}}(y))\big)\big) \\[0mm]
\phi_6(\tau,y)=\big(\mi{time}(\tau)\wedge \mi{uni}(y)\rightarrow 
                    \big(\tilde{H}_{X_3}^6(\tilde{s}(\tau),y)\geql
                         (\exists x\, (\mi{uni}(x)\wedge \tilde{H}_{X_2}(\tau,x)\wedge \tilde{G}_{\mi{high}_r}(x))\wedge \tilde{G}_{\mi{positive}_{\dot{t}}}(y))\big)\big) \\[2mm]
\text{\bf \small{Friction rules:}} \\[0mm]
\phi_7(\tau,y)=\big(\mi{time}(\tau)\wedge \mi{uni}(y)\rightarrow \\[0mm]
\phantom{\phi_7(\tau,y)=\big(}
                    \big(\tilde{H}_{X_5}^7(\tilde{s}(\tau),y)\geql
                         (\exists x\, (\mi{uni}(x)\wedge \tilde{H}_{X_1}(\tau,x)\wedge \tilde{G}_{\mi{medium}_d}(x))\wedge
                          \exists x\, (\mi{uni}(x)\wedge \tilde{H}_{X_2}(\tau,x)\wedge \tilde{G}_{\mi{high}_r}(x))\wedge 
                          \tilde{G}_{\mi{zero}_{\dot{r}}}(y))\big)\big) \\[0mm]
\phi_8(\tau,y)=\big(\mi{time}(\tau)\wedge \mi{uni}(y)\rightarrow \\[0mm]
\phantom{\phi_8(\tau,y)=\big(}
                    \big(\tilde{H}_{X_5}^8(\tilde{s}(\tau),y)\geql
                         (\exists x\, (\mi{uni}(x)\wedge \tilde{H}_{X_1}(\tau,x)\wedge \tilde{G}_{\mi{high}_d}(x))\wedge
                          \exists x\, (\mi{uni}(x)\wedge \tilde{H}_{X_2}(\tau,x)\wedge \tilde{G}_{\mi{high}_r}(x))\wedge   
                          \tilde{G}_{\mi{negative}_{\dot{r}}}(y))\big)\big) \\[2mm]
\text{\bf \small{Dynamic rules:}} \\[0mm]
\phi_{9-11}(\tau,y)=\big(\mi{time}(\tau)\wedge \mi{uni}(y)\rightarrow \\[0mm]
\phantom{\phi_{9-11}(\tau,y)=\big(}
                         \big(\tilde{H}_{X_i}^r(\tilde{s}(\tau),y)\geql
                              (\exists x\, (\mi{uni}(x)\wedge \tilde{H}_{X_i}(\tau,x)\wedge \tilde{G}_{\mi{medium}_k}(x))\wedge
                               \exists x\, (\mi{uni}(x)\wedge \tilde{H}_{X_{i+3}}(\tau,x)\wedge \tilde{G}_{\mi{negative}_{\dot{k}}}(x))\wedge
                               \tilde{G}_{\mi{low}_k}(y))\big)\big), \\[0mm]
\mbox{}\hfill (r,i,k)\in \{(9,0,t),(10,1,d),(11,2,r)\} \\[0mm]
\phi_{12-14}(\tau,y)=\big(\mi{time}(\tau)\wedge \mi{uni}(y)\rightarrow \\[0mm]
\phantom{\phi_{12-14}(\tau,y)=\big(}
                          \big(\tilde{H}_{X_i}^r(\tilde{s}(\tau),y)\geql
                               (\exists x\, (\mi{uni}(x)\wedge \tilde{H}_{X_i}(\tau,x)\wedge \tilde{G}_{\mi{high}_k}(x))\wedge
                                \exists x\, (\mi{uni}(x)\wedge \tilde{H}_{X_{i+3}}(\tau,x)\wedge \tilde{G}_{\mi{negative}_{\dot{k}}}(x))\wedge
                                \tilde{G}_{\mi{medium}_k}(y))\big)\big), \\[0mm] 
\mbox{}\hfill (r,i,k)\in \{(12,0,t),(13,1,d),(14,2,r)\} \\[0mm]
\phi_{15-23}(\tau,y)=\big(\mi{time}(\tau)\wedge \mi{uni}(y)\rightarrow \\[0mm]
\phantom{\phi_{15-23}(\tau,y)=\big(}
                          \big(\tilde{H}_{X_i}^r(\tilde{s}(\tau),y)\geql
                               (\exists x\, (\mi{uni}(x)\wedge \tilde{H}_{X_i}(\tau,x)\wedge \tilde{G}_{\mi{W}_k}(x))\wedge  
                                \exists x\, (\mi{uni}(x)\wedge \tilde{H}_{X_{i+3}}(\tau,x)\wedge \tilde{G}_{\mi{zero}_{\dot{k}}}(x))\wedge
                                \tilde{G}_{\mi{W}_k}(y))\big)\big), \\[0mm] 
\mbox{}\hfill \begin{alignedat}{1}
              & (r,i,k,\mi{W})\in \{(15,0,t,\mi{low}),(16,0,t,\mi{medium}),(17,0,t,\mi{high}),
                                    (18,1,d,\mi{low}),(19,1,d,\mi{medium}),(20,1,d,\mi{high}), \\[0mm]
              & \phantom{(r,i,k,\mi{W})\in \{}
                                    (21,2,r,\mi{low}),(22,2,r,\mi{medium}),(23,2,r,\mi{high})\} 
              \end{alignedat} \\[0mm]
\phi_{24-26}(\tau,y)=\big(\mi{time}(\tau)\wedge \mi{uni}(y)\rightarrow \\[0mm]
\phantom{\phi_{24-26}(\tau,y)=\big(}
                          \big(\tilde{H}_{X_i}^r(\tilde{s}(\tau),y)\geql
                               (\exists x\, (\mi{uni}(x)\wedge \tilde{H}_{X_i}(\tau,x)\wedge \tilde{G}_{\mi{low}_k}(x))\wedge
                                \exists x\, (\mi{uni}(x)\wedge \tilde{H}_{X_{i+3}}(\tau,x)\wedge \tilde{G}_{\mi{positive}_{\dot{k}}}(x))\wedge
                                \tilde{G}_{\mi{medium}_k}(y))\big)\big), \\[0mm] 
\mbox{}\hfill (r,i,k)\in \{(24,0,t),(25,1,d),(26,2,r)\} \\[0mm]
\phi_{27-29}(\tau,y)=\big(\mi{time}(\tau)\wedge \mi{uni}(y)\rightarrow \\[0mm]
\phantom{\phi_{27-29}(\tau,y)=\big(}
                          \big(\tilde{H}_{X_i}^r(\tilde{s}(\tau),y)\geql
                               (\exists x\, (\tilde{H}_{X_i}(\tau,x)\wedge \tilde{G}_{\mi{medium}_k}(x))\wedge
                                \exists x\, (\tilde{H}_{X_{i+3}}(\tau,x)\wedge \tilde{G}_{\mi{positive}_{\dot{k}}}(x))\wedge
                                \tilde{G}_{\mi{high}_k}(y))\big)\big), \\[0mm] 
\mbox{}\hfill (r,i,k)\in \{(27,0,t),(28,1,d),(29,2,r)\} \\[1mm]
\hline \hline
\end{IEEEeqnarray*}
\end{minipage}
\vspace{-2mm}
\end{table*}
\begin{table*}[p]
\vspace{-6mm}
\caption{Translation of the fuzzy rule base $B$}\label{tab555}
\vspace{-6mm}
\centering
\hspace*{-2.5mm}
\begin{minipage}[t]{\linewidth+5mm}
\scriptsize
\begin{IEEEeqnarray*}{L}
\hline \hline \\[1mm]
\text{\bf \small{Limit rules:}} \\[0mm]
\phi_{30a-32a}(\tau,y)=\big(\mi{time}(\tau)\wedge \mi{uni}(y)\rightarrow \\[0mm]
\phantom{\phi_{30a-32a}(\tau,y)=\big(}
                            \big(\tilde{H}_{X_{i+3}}^r(\tilde{s}(\tau),y)\geql
                                 (\exists x\, (\mi{uni}(x)\wedge \tilde{H}_{X_i}(\tau,x)\wedge \tilde{G}_{\mi{low}_k}(x))\wedge
                                  \exists x\, (\mi{uni}(x)\wedge \tilde{H}_{X_{i+3}}(\tau,x)\wedge \tilde{G}_{\mi{negative}_{\dot{k}}}(x))\wedge
                                  \tilde{G}_{\mi{low}_k}(y))\big)\big), \\[0mm]
\mbox{}\hfill (r,i,k)\in \{(30a,0,t),(31a,1,d),(32a,2,r)\} \\[0mm]
\phi_{30b-32b}(\tau,y)=\big(\mi{time}(\tau)\wedge \mi{uni}(y)\rightarrow \\[0mm]
\phantom{\phi_{30b-32b}(\tau,y)=\big(}
                            \big(\tilde{H}_{X_{i+3}}^r(\tilde{s}(\tau),y)\geql
                                 (\exists x\, (\mi{uni}(x)\wedge \tilde{H}_{X_i}(\tau,x)\wedge \tilde{G}_{\mi{low}_k}(x))\wedge
                                  \exists x\, (\mi{uni}(x)\wedge \tilde{H}_{X_{i+3}}(\tau,x)\wedge \tilde{G}_{\mi{negative}_{\dot{k}}}(x))\wedge
                                  \tilde{G}_{\mi{positive}_{\dot{k}}}(y))\big)\big), \\[0mm]
\mbox{}\hfill (r,i,k)\in \{(30b,0,t),(31b,1,d),(32b,2,r)\} \\[0mm]
\phi_{33-38}(\tau,y)=\big(\mi{time}(\tau)\wedge \mi{uni}(y)\rightarrow \\[0mm]
\phantom{\phi_{33-38}(\tau,y)=\big(}
                          \big(\tilde{H}_{X_{i+3}}^r(\tilde{s}(\tau),y)\geql
                               (\exists x\, (\mi{uni}(x)\wedge \tilde{H}_{X_i}(\tau,x)\wedge \tilde{G}_{\mi{low}_k}(x))\wedge
                                \exists x\, (\mi{uni}(x)\wedge \tilde{H}_{X_{i+3}}(\tau,x)\wedge \tilde{G}_{\mi{W}_{\dot{k}}}(x))\wedge
                                \tilde{G}_{\mi{W}_{\dot{k}}}(y))\big)\big), \\[0mm]
\mbox{}\hfill (r,i,k,\mi{W})\in \{(33,0,t,\mi{zero}),(34,0,t,\mi{positive}),(35,1,d,\mi{zero}),(36,1,d,\mi{positive}), 
                                  (37,2,r,\mi{zero}),(38,2,r,\mi{positive})\} \\[0mm]
\phi_{39-47}(\tau,y)=\big(\mi{time}(\tau)\wedge \mi{uni}(y)\rightarrow \\[0mm]
\phantom{\phi_{39-47}(\tau,y)=\big(}
                          \big(\tilde{H}_{X_{i+3}}^r(\tilde{s}(\tau),y)\geql
                               (\exists x\, (\mi{uni}(x)\wedge \tilde{H}_{X_i}(\tau,x)\wedge \tilde{G}_{\mi{medium}_k}(x))\wedge
                                \exists x\, (\mi{uni}(x)\wedge \tilde{H}_{X_{i+3}}(\tau,x)\wedge \tilde{G}_{\mi{W}_{\dot{k}}}(x))\wedge
                                \tilde{G}_{\mi{W}_{\dot{k}}}(y))\big)\big), \\[0mm]
\mbox{}\hfill \begin{alignedat}{1}
              & (r,i,k,\mi{W})\in \{(39,0,t,\mi{negative}),(40,0,t,\mi{zero}),(41,0,t,\mi{positive}), 
                                    (42,1,d,\mi{negative}),(43,1,d,\mi{zero}),(44,1,d,\mi{positive}), \\[0mm]
              & \phantom{(r,i,k,\mi{W})\in \{}
                                    (45,2,r,\mi{negative}),(46,2,r,\mi{zero}),(47,2,r,\mi{positive})\} 
              \end{alignedat} \\[0mm]
\phi_{48a-50a}(\tau,y)=\big(\mi{time}(\tau)\wedge \mi{uni}(y)\rightarrow \\[0mm]
\phantom{\phi_{48a-50a}(\tau,y)=\big(}
                            \big(\tilde{H}_{X_{i+3}}^r(\tilde{s}(\tau),y)\geql
                                 (\exists x\, (\mi{uni}(x)\wedge \tilde{H}_{X_i}(\tau,x)\wedge \tilde{G}_{\mi{high}_k}(x))\wedge
                                  \exists x\, (\mi{uni}(x)\wedge \tilde{H}_{X_{i+3}}(\tau,x)\wedge \tilde{G}_{\mi{positive}_{\dot{k}}}(x))\wedge
                                  \tilde{G}_{\mi{high}_k}(y))\big)\big), \\[0mm]
\mbox{}\hfill (r,i,k)\in \{(48a,0,t),(49a,1,d),(50a,2,r)\} \\[0mm]
\phi_{48b-50b}(\tau,y)=\big(\mi{time}(\tau)\wedge \mi{uni}(y)\rightarrow \\[0mm]
\phantom{\phi_{48b-50b}(\tau,y)=\big(}
                            \big(\tilde{H}_{X_{i+3}}^r(\tilde{s}(\tau),y)\geql  
                                 (\exists x\, (\mi{uni}(x)\wedge \tilde{H}_{X_i}(\tau,x)\wedge \tilde{G}_{\mi{high}_k}(x))\wedge 
                                  \exists x\, (\mi{uni}(x)\wedge \tilde{H}_{X_{i+3}}(\tau,x)\wedge \tilde{G}_{\mi{positive}_{\dot{k}}}(x))\wedge  
                                  \tilde{G}_{\mi{negative}_{\dot{k}}}(y))\big)\big), \\[0mm]  
\mbox{}\hfill (r,i,k)\in \{(48b,0,t),(49b,1,d),(50b,2,r)\} \\[0mm]
\phi_{51-56}(\tau,y)=\big(\mi{time}(\tau)\wedge \mi{uni}(y)\rightarrow \\[0mm]
\phantom{\phi_{51-56}(\tau,y)=\big(}
                          \big(\tilde{H}_{X_{i+3}}^r(\tilde{s}(\tau),y)\geql
                               (\exists x\, (\mi{uni}(x)\wedge \tilde{H}_{X_i}(\tau,x)\wedge \tilde{G}_{\mi{high}_k}(x))\wedge
                                \exists x\, (\mi{uni}(x)\wedge \tilde{H}_{X_{i+3}}(\tau,x)\wedge \tilde{G}_{\mi{W}_{\dot{k}}}(x))\wedge
                                \tilde{G}_{\mi{W}_{\dot{k}}}(y))\big)\big), \\[0mm]
\mbox{}\hfill (r,i,k,\mi{W})\in \{(51,0,t,\mi{negative}),(52,0,t,\mi{zero}),(53,1,d,\mi{negative}),(54,1,d,\mi{zero}),
                                  (55,2,r,\mi{negative}),(56,2,r,\mi{zero})\} \\[2mm]
\text{\bf \small{Aggregation rules:}} \\[0mm]
\mi{time}(\tau)\wedge \mi{uni}(y)\rightarrow \big(\tilde{H}_X(\tilde{s}(\tau),y)\geql 
                                                  \bigvee_{R_i\in B, \mi{out}(R_i)=X} \tilde{H}_X^i(\tilde{s}(\tau),y)\big),\ 
X\in \mbb{X} \\[1mm]
\hline \hline
\end{IEEEeqnarray*}
\end{minipage}     
\vspace{-2mm}      
\end{table*}        
The fuzzy rules -- formulae from $T_B$ can further be translated to clausal form, 
e.g. $\phi_1(\tau,y)$ (Tables~\ref{tab6}--\ref{tab666}) and $\phi_7(\tau,y)$ (Tables~\ref{tab7}--\ref{tab7777}).
\begin{table*}[p]
\vspace{-6mm}
\caption{Translation of the formula $\phi_1(\tau,y)\in T_B$ to clausal form}\label{tab6}
\vspace{-6mm}
\centering
\begin{minipage}[t]{\linewidth}
\footnotesize
\begin{IEEEeqnarray*}{LR} 
\hline \hline \\[2mm]
\Big\{
  \tilde{p}_{1,0}(\tau,x,y)\geql \gu,
& \\
\phantom{\Big\{}
  \tilde{p}_{1,0}(\tau,x,y)\leftrightarrow
  \big(\underbrace{\mi{time}(\tau)\wedge \mi{uni}(y)}_{\tilde{p}_{1,1}(\tau,x,y)}\rightarrow        
       \big(\underbrace{\tilde{H}_{X_1}^1(\tilde{s}(\tau),y)\geql        
                        (\exists x\, (\mi{uni}(x)\wedge \tilde{H}_{X_0}(\tau,x)\wedge \tilde{G}_{\mi{low}_t}(x))\wedge 
                         \tilde{G}_{\mi{high}_d}(y))}_{\tilde{p}_{1,2}(\tau,x,y)}\big)\big)\Big\}
& \quad (\ref{eq0rr3+}) \\
\Big\{
  \tilde{p}_{1,0}(\tau,x,y)\geql \gu,
& \\
\phantom{\Big\{}
  \tilde{p}_{1,1}(\tau,x,y)\gle \tilde{p}_{1,2}(\tau,x,y)\vee \tilde{p}_{1,1}(\tau,x,y)\geql \tilde{p}_{1,2}(\tau,x,y)\vee 
  \tilde{p}_{1,0}(\tau,x,y)\geql \tilde{p}_{1,2}(\tau,x,y),
& \\
\phantom{\Big\{}
  \tilde{p}_{1,2}(\tau,x,y)\gle \tilde{p}_{1,1}(\tau,x,y)\vee \tilde{p}_{1,0}(\tau,x,y)\geql \gu,
& \\
\phantom{\Big\{}
  \tilde{p}_{1,1}(\tau,x,y)\leftrightarrow \underbrace{\mi{time}(\tau)}_{\tilde{p}_{1,3}(\tau,x,y)}\wedge 
                                           \underbrace{\mi{uni}(y)}_{\tilde{p}_{1,4}(\tau,x,y)},
& \\
\phantom{\Big\{}
  \tilde{p}_{1,2}(\tau,x,y)\leftrightarrow 
  \big(\underbrace{\tilde{H}_{X_1}^1(\tilde{s}(\tau),y)}_{\tilde{p}_{1,5}(\tau,x,y)}\geql
       (\underbrace{\exists x\, (\mi{uni}(x)\wedge \tilde{H}_{X_0}(\tau,x)\wedge \tilde{G}_{\mi{low}_t}(x))\wedge 
                    \tilde{G}_{\mi{high}_d}(y)}_{\tilde{p}_{1,6}(\tau,x,y)})\big)\Big\}
& \quad (\ref{eq0rr1+}), (\ref{eq0rr7+}) \\
\Big\{
  \tilde{p}_{1,0}(\tau,x,y)\geql \gu,
& \\
\phantom{\Big\{}
  \tilde{p}_{1,1}(\tau,x,y)\gle \tilde{p}_{1,2}(\tau,x,y)\vee \tilde{p}_{1,1}(\tau,x,y)\geql \tilde{p}_{1,2}(\tau,x,y)\vee
  \tilde{p}_{1,0}(\tau,x,y)\geql \tilde{p}_{1,2}(\tau,x,y),
& \\
\phantom{\Big\{}
  \tilde{p}_{1,2}(\tau,x,y)\gle \tilde{p}_{1,1}(\tau,x,y)\vee \tilde{p}_{1,0}(\tau,x,y)\geql \gu,
& \\
\phantom{\Big\{}
  \tilde{p}_{1,3}(\tau,x,y)\gle \tilde{p}_{1,4}(\tau,x,y)\vee \tilde{p}_{1,3}(\tau,x,y)\geql \tilde{p}_{1,4}(\tau,x,y)\vee
  \tilde{p}_{1,1}(\tau,x,y)\geql \tilde{p}_{1,4}(\tau,x,y),
& \\
\phantom{\Big\{}
  \tilde{p}_{1,4}(\tau,x,y)\gle \tilde{p}_{1,3}(\tau,x,y)\vee \tilde{p}_{1,1}(\tau,x,y)\geql \tilde{p}_{1,3}(\tau,x,y),  
& \\
\phantom{\Big\{}
  \tilde{p}_{1,3}(\tau,x,y)\geql \mi{time}(\tau), \tilde{p}_{1,4}(\tau,x,y)\geql \mi{uni}(y),
& \\
\phantom{\Big\{}
  \tilde{p}_{1,5}(\tau,x,y)\geql \tilde{p}_{1,6}(\tau,x,y)\vee \tilde{p}_{1,2}(\tau,x,y)\geql \gz,
  \tilde{p}_{1,5}(\tau,x,y)\gle \tilde{p}_{1,6}(\tau,x,y)\vee \tilde{p}_{1,6}(\tau,x,y)\gle \tilde{p}_{1,5}(\tau,x,y)\vee 
  \tilde{p}_{1,2}(\tau,x,y)\geql \gu,
& \\
\phantom{\Big\{}
  \tilde{p}_{1,5}(\tau,x,y)\geql \tilde{H}_{X_1}^1(\tilde{s}(\tau),y),
  \tilde{p}_{1,6}(\tau,x,y)\leftrightarrow \big(\underbrace{\exists x\, (\mi{uni}(x)\wedge \tilde{H}_{X_0}(\tau,x)\wedge \tilde{G}_{\mi{low}_t}(x))}_{\tilde{p}_{1,7}(\tau,x,y)}\wedge
                                                \underbrace{\tilde{G}_{\mi{high}_d}(y)}_{\tilde{p}_{1,8}(\tau,x,y)}\big)\Big\}
& \quad (\ref{eq0rr1+}) \\
\Big\{
  \tilde{p}_{1,0}(\tau,x,y)\geql \gu,
& \\
\phantom{\Big\{}
  \tilde{p}_{1,1}(\tau,x,y)\gle \tilde{p}_{1,2}(\tau,x,y)\vee \tilde{p}_{1,1}(\tau,x,y)\geql \tilde{p}_{1,2}(\tau,x,y)\vee
  \tilde{p}_{1,0}(\tau,x,y)\geql \tilde{p}_{1,2}(\tau,x,y),
& \\
\phantom{\Big\{}
  \tilde{p}_{1,2}(\tau,x,y)\gle \tilde{p}_{1,1}(\tau,x,y)\vee \tilde{p}_{1,0}(\tau,x,y)\geql \gu,
& \\
\phantom{\Big\{}
  \tilde{p}_{1,3}(\tau,x,y)\gle \tilde{p}_{1,4}(\tau,x,y)\vee \tilde{p}_{1,3}(\tau,x,y)\geql \tilde{p}_{1,4}(\tau,x,y)\vee
  \tilde{p}_{1,1}(\tau,x,y)\geql \tilde{p}_{1,4}(\tau,x,y),
& \\
\phantom{\Big\{}
  \tilde{p}_{1,4}(\tau,x,y)\gle \tilde{p}_{1,3}(\tau,x,y)\vee \tilde{p}_{1,1}(\tau,x,y)\geql \tilde{p}_{1,3}(\tau,x,y),
& \\
\phantom{\Big\{}
  \tilde{p}_{1,3}(\tau,x,y)\geql \mi{time}(\tau), \tilde{p}_{1,4}(\tau,x,y)\geql \mi{uni}(y),
& \\
\phantom{\Big\{}
  \tilde{p}_{1,5}(\tau,x,y)\geql \tilde{p}_{1,6}(\tau,x,y)\vee \tilde{p}_{1,2}(\tau,x,y)\geql \gz,
  \tilde{p}_{1,5}(\tau,x,y)\gle \tilde{p}_{1,6}(\tau,x,y)\vee \tilde{p}_{1,6}(\tau,x,y)\gle \tilde{p}_{1,5}(\tau,x,y)\vee
  \tilde{p}_{1,2}(\tau,x,y)\geql \gu,
& \\
\phantom{\Big\{}
  \tilde{p}_{1,5}(\tau,x,y)\geql \tilde{H}_{X_1}^1(\tilde{s}(\tau),y),
& \\
\phantom{\Big\{}
  \tilde{p}_{1,7}(\tau,x,y)\gle \tilde{p}_{1,8}(\tau,x,y)\vee \tilde{p}_{1,7}(\tau,x,y)\geql \tilde{p}_{1,8}(\tau,x,y)\vee 
  \tilde{p}_{1,6}(\tau,x,y)\geql \tilde{p}_{1,8}(\tau,x,y),
& \\
\phantom{\Big\{}
  \tilde{p}_{1,8}(\tau,x,y)\gle \tilde{p}_{1,7}(\tau,x,y)\vee \tilde{p}_{1,6}(\tau,x,y)\geql \tilde{p}_{1,7}(\tau,x,y),
& \\
\phantom{\Big\{}
  \tilde{p}_{1,7}(\tau,x,y)\leftrightarrow \exists x\, (\underbrace{\mi{uni}(x)\wedge \tilde{H}_{X_0}(\tau,x)\wedge \tilde{G}_{\mi{low}_t}(x)}_{\tilde{p}_{1,9}(\tau,x,y)}),
  \tilde{p}_{1,8}(\tau,x,y)\geql \tilde{G}_{\mi{high}_d}(y)\Big\}
& \quad (\ref{eq0rr6+}) \\[1mm]
\hline \hline
\end{IEEEeqnarray*}
\end{minipage}     
\vspace{-2mm}      
\end{table*}        
\begin{table*}[p]
\vspace{-6mm}
\caption{Translation of the formula $\phi_1(\tau,y)\in T_B$ to clausal form}\label{tab66}
\vspace{-6mm}
\centering
\begin{minipage}[t]{\linewidth}
\footnotesize   
\begin{IEEEeqnarray*}{LR} 
\hline \hline \\[2mm]
\Big\{
  \tilde{p}_{1,0}(\tau,x,y)\geql \gu,
& \\
\phantom{\Big\{}
  \tilde{p}_{1,1}(\tau,x,y)\gle \tilde{p}_{1,2}(\tau,x,y)\vee \tilde{p}_{1,1}(\tau,x,y)\geql \tilde{p}_{1,2}(\tau,x,y)\vee
  \tilde{p}_{1,0}(\tau,x,y)\geql \tilde{p}_{1,2}(\tau,x,y),
& \\
\phantom{\Big\{}
  \tilde{p}_{1,2}(\tau,x,y)\gle \tilde{p}_{1,1}(\tau,x,y)\vee \tilde{p}_{1,0}(\tau,x,y)\geql \gu,
& \\
\phantom{\Big\{}
  \tilde{p}_{1,3}(\tau,x,y)\gle \tilde{p}_{1,4}(\tau,x,y)\vee \tilde{p}_{1,3}(\tau,x,y)\geql \tilde{p}_{1,4}(\tau,x,y)\vee
  \tilde{p}_{1,1}(\tau,x,y)\geql \tilde{p}_{1,4}(\tau,x,y),
& \\
\phantom{\Big\{}
  \tilde{p}_{1,4}(\tau,x,y)\gle \tilde{p}_{1,3}(\tau,x,y)\vee \tilde{p}_{1,1}(\tau,x,y)\geql \tilde{p}_{1,3}(\tau,x,y),
& \\
\phantom{\Big\{}
  \tilde{p}_{1,3}(\tau,x,y)\geql \mi{time}(\tau), \tilde{p}_{1,4}(\tau,x,y)\geql \mi{uni}(y),
& \\
\phantom{\Big\{}
  \tilde{p}_{1,5}(\tau,x,y)\geql \tilde{p}_{1,6}(\tau,x,y)\vee \tilde{p}_{1,2}(\tau,x,y)\geql \gz,
  \tilde{p}_{1,5}(\tau,x,y)\gle \tilde{p}_{1,6}(\tau,x,y)\vee \tilde{p}_{1,6}(\tau,x,y)\gle \tilde{p}_{1,5}(\tau,x,y)\vee
  \tilde{p}_{1,2}(\tau,x,y)\geql \gu,
& \\
\phantom{\Big\{}
  \tilde{p}_{1,5}(\tau,x,y)\geql \tilde{H}_{X_1}^1(\tilde{s}(\tau),y),
& \\
\phantom{\Big\{}
  \tilde{p}_{1,7}(\tau,x,y)\gle \tilde{p}_{1,8}(\tau,x,y)\vee \tilde{p}_{1,7}(\tau,x,y)\geql \tilde{p}_{1,8}(\tau,x,y)\vee
  \tilde{p}_{1,6}(\tau,x,y)\geql \tilde{p}_{1,8}(\tau,x,y),
& \\
\phantom{\Big\{}
  \tilde{p}_{1,8}(\tau,x,y)\gle \tilde{p}_{1,7}(\tau,x,y)\vee \tilde{p}_{1,6}(\tau,x,y)\geql \tilde{p}_{1,7}(\tau,x,y),
& \\
\phantom{\Big\{}
  \tilde{p}_{1,7}(\tau,x,y)\geql \exists x\, \tilde{p}_{1,9}(\tau,x,y),
  \tilde{p}_{1,9}(\tau,x,y)\leftrightarrow \underbrace{\mi{uni}(x)}_{\tilde{p}_{1,10}(\tau,x,y)}\wedge 
                                           \underbrace{\tilde{H}_{X_0}(\tau,x)\wedge \tilde{G}_{\mi{low}_t}(x)}_{\tilde{p}_{1,11}(\tau,x,y)}
& \\
\phantom{\Big\{}
  \tilde{p}_{1,8}(\tau,x,y)\geql \tilde{G}_{\mi{high}_d}(y)\Big\}
& \quad (\ref{eq0rr1+}) \\
\Big\{
  \tilde{p}_{1,0}(\tau,x,y)\geql \gu,
& \\
\phantom{\Big\{}
  \tilde{p}_{1,1}(\tau,x,y)\gle \tilde{p}_{1,2}(\tau,x,y)\vee \tilde{p}_{1,1}(\tau,x,y)\geql \tilde{p}_{1,2}(\tau,x,y)\vee
  \tilde{p}_{1,0}(\tau,x,y)\geql \tilde{p}_{1,2}(\tau,x,y),
& \\
\phantom{\Big\{}
  \tilde{p}_{1,2}(\tau,x,y)\gle \tilde{p}_{1,1}(\tau,x,y)\vee \tilde{p}_{1,0}(\tau,x,y)\geql \gu,
& \\
\phantom{\Big\{}
  \tilde{p}_{1,3}(\tau,x,y)\gle \tilde{p}_{1,4}(\tau,x,y)\vee \tilde{p}_{1,3}(\tau,x,y)\geql \tilde{p}_{1,4}(\tau,x,y)\vee
  \tilde{p}_{1,1}(\tau,x,y)\geql \tilde{p}_{1,4}(\tau,x,y),
& \\
\phantom{\Big\{}
  \tilde{p}_{1,4}(\tau,x,y)\gle \tilde{p}_{1,3}(\tau,x,y)\vee \tilde{p}_{1,1}(\tau,x,y)\geql \tilde{p}_{1,3}(\tau,x,y),
& \\
\phantom{\Big\{}
  \tilde{p}_{1,3}(\tau,x,y)\geql \mi{time}(\tau), \tilde{p}_{1,4}(\tau,x,y)\geql \mi{uni}(y),
& \\
\phantom{\Big\{}
  \tilde{p}_{1,5}(\tau,x,y)\geql \tilde{p}_{1,6}(\tau,x,y)\vee \tilde{p}_{1,2}(\tau,x,y)\geql \gz,
  \tilde{p}_{1,5}(\tau,x,y)\gle \tilde{p}_{1,6}(\tau,x,y)\vee \tilde{p}_{1,6}(\tau,x,y)\gle \tilde{p}_{1,5}(\tau,x,y)\vee
  \tilde{p}_{1,2}(\tau,x,y)\geql \gu,
& \\
\phantom{\Big\{}
  \tilde{p}_{1,5}(\tau,x,y)\geql \tilde{H}_{X_1}^1(\tilde{s}(\tau),y),
& \\
\phantom{\Big\{}
  \tilde{p}_{1,7}(\tau,x,y)\gle \tilde{p}_{1,8}(\tau,x,y)\vee \tilde{p}_{1,7}(\tau,x,y)\geql \tilde{p}_{1,8}(\tau,x,y)\vee
  \tilde{p}_{1,6}(\tau,x,y)\geql \tilde{p}_{1,8}(\tau,x,y),
& \\
\phantom{\Big\{}
  \tilde{p}_{1,8}(\tau,x,y)\gle \tilde{p}_{1,7}(\tau,x,y)\vee \tilde{p}_{1,6}(\tau,x,y)\geql \tilde{p}_{1,7}(\tau,x,y),
& \\
\phantom{\Big\{}
  \tilde{p}_{1,7}(\tau,x,y)\geql \exists x\, \tilde{p}_{1,9}(\tau,x,y),
& \\
\phantom{\Big\{}
    \tilde{p}_{1,10}(\tau,x,y)\gle \tilde{p}_{1,11}(\tau,x,y)\vee \tilde{p}_{1,10}(\tau,x,y)\geql \tilde{p}_{1,11}(\tau,x,y)\vee
    \tilde{p}_{1,9}(\tau,x,y)\geql \tilde{p}_{1,11}(\tau,x,y),
& \\
\phantom{\Big\{}
  \tilde{p}_{1,11}(\tau,x,y)\gle \tilde{p}_{1,10}(\tau,x,y)\vee \tilde{p}_{1,9}(\tau,x,y)\geql \tilde{p}_{1,10}(\tau,x,y),
& \\
\phantom{\Big\{}
  \tilde{p}_{1,10}(\tau,x,y)\geql \mi{uni}(x),
  \tilde{p}_{1,11}(\tau,x,y)\leftrightarrow \underbrace{\tilde{H}_{X_0}(\tau,x)}_{\tilde{p}_{1,12}(\tau,x,y)}\wedge 
                                            \underbrace{\tilde{G}_{\mi{low}_t}(x)}_{\tilde{p}_{1,13}(\tau,x,y)},
& \\
\phantom{\Big\{}
  \tilde{p}_{1,8}(\tau,x,y)\geql \tilde{G}_{\mi{high}_d}(y)\Big\}
& \quad (\ref{eq0rr1+}) \\[1mm]
\hline \hline
\end{IEEEeqnarray*}
\end{minipage}     
\vspace{-2mm}      
\end{table*}        
\begin{table*}[p]
\vspace{-6mm}
\caption{Translation of the formula $\phi_1(\tau,y)\in T_B$ to clausal form}\label{tab666}
\vspace{-6mm}
\centering
\begin{minipage}[t]{\linewidth}
\footnotesize   
\begin{IEEEeqnarray*}{LR} 
\hline \hline \\[2mm]
\Big\{
  \tilde{p}_{1,0}(\tau,x,y)\geql \gu,
& \\
\phantom{\Big\{}
  \tilde{p}_{1,1}(\tau,x,y)\gle \tilde{p}_{1,2}(\tau,x,y)\vee \tilde{p}_{1,1}(\tau,x,y)\geql \tilde{p}_{1,2}(\tau,x,y)\vee
  \tilde{p}_{1,0}(\tau,x,y)\geql \tilde{p}_{1,2}(\tau,x,y),
& \\
\phantom{\Big\{}
  \tilde{p}_{1,2}(\tau,x,y)\gle \tilde{p}_{1,1}(\tau,x,y)\vee \tilde{p}_{1,0}(\tau,x,y)\geql \gu,
& \\
\phantom{\Big\{}
  \tilde{p}_{1,3}(\tau,x,y)\gle \tilde{p}_{1,4}(\tau,x,y)\vee \tilde{p}_{1,3}(\tau,x,y)\geql \tilde{p}_{1,4}(\tau,x,y)\vee
  \tilde{p}_{1,1}(\tau,x,y)\geql \tilde{p}_{1,4}(\tau,x,y),
& \\
\phantom{\Big\{}
  \tilde{p}_{1,4}(\tau,x,y)\gle \tilde{p}_{1,3}(\tau,x,y)\vee \tilde{p}_{1,1}(\tau,x,y)\geql \tilde{p}_{1,3}(\tau,x,y),
& \\
\phantom{\Big\{}
  \tilde{p}_{1,3}(\tau,x,y)\geql \mi{time}(\tau), \tilde{p}_{1,4}(\tau,x,y)\geql \mi{uni}(y),
& \\
\phantom{\Big\{}
  \tilde{p}_{1,5}(\tau,x,y)\geql \tilde{p}_{1,6}(\tau,x,y)\vee \tilde{p}_{1,2}(\tau,x,y)\geql \gz,
  \tilde{p}_{1,5}(\tau,x,y)\gle \tilde{p}_{1,6}(\tau,x,y)\vee \tilde{p}_{1,6}(\tau,x,y)\gle \tilde{p}_{1,5}(\tau,x,y)\vee
  \tilde{p}_{1,2}(\tau,x,y)\geql \gu,
& \\
\phantom{\Big\{}
  \tilde{p}_{1,5}(\tau,x,y)\geql \tilde{H}_{X_1}^1(\tilde{s}(\tau),y),
& \\
\phantom{\Big\{}
  \tilde{p}_{1,7}(\tau,x,y)\gle \tilde{p}_{1,8}(\tau,x,y)\vee \tilde{p}_{1,7}(\tau,x,y)\geql \tilde{p}_{1,8}(\tau,x,y)\vee
  \tilde{p}_{1,6}(\tau,x,y)\geql \tilde{p}_{1,8}(\tau,x,y),
& \\
\phantom{\Big\{}
  \tilde{p}_{1,8}(\tau,x,y)\gle \tilde{p}_{1,7}(\tau,x,y)\vee \tilde{p}_{1,6}(\tau,x,y)\geql \tilde{p}_{1,7}(\tau,x,y),
& \\
\phantom{\Big\{}
  \tilde{p}_{1,7}(\tau,x,y)\geql \exists x\, \tilde{p}_{1,9}(\tau,x,y),
& \\
\phantom{\Big\{}
    \tilde{p}_{1,10}(\tau,x,y)\gle \tilde{p}_{1,11}(\tau,x,y)\vee \tilde{p}_{1,10}(\tau,x,y)\geql \tilde{p}_{1,11}(\tau,x,y)\vee
    \tilde{p}_{1,9}(\tau,x,y)\geql \tilde{p}_{1,11}(\tau,x,y),
& \\
\phantom{\Big\{}
  \tilde{p}_{1,11}(\tau,x,y)\gle \tilde{p}_{1,10}(\tau,x,y)\vee \tilde{p}_{1,9}(\tau,x,y)\geql \tilde{p}_{1,10}(\tau,x,y),
& \\
\phantom{\Big\{}
  \tilde{p}_{1,10}(\tau,x,y)\geql \mi{uni}(x),
& \\
\phantom{\Big\{}
    \tilde{p}_{1,12}(\tau,x,y)\gle \tilde{p}_{1,13}(\tau,x,y)\vee \tilde{p}_{1,12}(\tau,x,y)\geql \tilde{p}_{1,13}(\tau,x,y)\vee
    \tilde{p}_{1,11}(\tau,x,y)\geql \tilde{p}_{1,13}(\tau,x,y),
& \\
\phantom{\Big\{}
  \tilde{p}_{1,13}(\tau,x,y)\gle \tilde{p}_{1,12}(\tau,x,y)\vee \tilde{p}_{1,11}(\tau,x,y)\geql \tilde{p}_{1,12}(\tau,x,y),
& \\
\phantom{\Big\{}
  \tilde{p}_{1,12}(\tau,x,y)\geql \tilde{H}_{X_0}(\tau,x),
  \tilde{p}_{1,13}(\tau,x,y)\geql \tilde{G}_{\mi{low}_t}(x),
& \\
\phantom{\Big\{}
  \tilde{p}_{1,8}(\tau,x,y)\geql \tilde{G}_{\mi{high}_d}(y)\Big\} \\[1mm]
\hline \hline
\end{IEEEeqnarray*}
\end{minipage}     
\vspace{-2mm}      
\end{table*}        
\begin{table*}[p]
\vspace{-6mm}
\caption{Translation of the formula $\phi_7(\tau,y)\in T_B$ to clausal form}\label{tab7}
\vspace{-6mm}
\centering
%\hspace*{-12mm}
\begin{minipage}[t]{\linewidth}
\scriptsize
\begin{IEEEeqnarray*}{LR} 
\hline \hline \\[2mm]
\Big\{
  \tilde{p}_{7,0}(\tau,x,y)\geql \gu,
& \\
\phantom{\Big\{}
  \tilde{p}_{7,0}(\tau,x,y)\leftrightarrow
  \big(\underbrace{\mi{time}(\tau)\wedge \mi{uni}(y)}_{\tilde{p}_{7,1}(\tau,x,y)}\rightarrow
& \\
\phantom{\Big\{\tilde{p}_{7,0}(\tau,x,y)\leftrightarrow \big(}
       \big(\underbrace{\tilde{H}_{X_5}^7(\tilde{s}(\tau),y)\geql
                        (\exists x\, (\mi{uni}(x)\wedge \tilde{H}_{X_1}(\tau,x)\wedge \tilde{G}_{\mi{medium}_d}(x))\wedge
                         \exists x\, (\mi{uni}(x)\wedge \tilde{H}_{X_2}(\tau,x)\wedge \tilde{G}_{\mi{high}_r}(x))\wedge  
                         \tilde{G}_{\mi{zero}_{\dot{r}}}(y))}_{\tilde{p}_{7,2}(\tau,x,y)}\big)\big)\Big\}
& \quad (\ref{eq0rr3+}) \\
\Big\{
  \tilde{p}_{7,0}(\tau,x,y)\geql \gu,
& \\
\phantom{\Big\{}
  \tilde{p}_{7,1}(\tau,x,y)\gle \tilde{p}_{7,2}(\tau,x,y)\vee \tilde{p}_{7,1}(\tau,x,y)\geql \tilde{p}_{7,2}(\tau,x,y)\vee
  \tilde{p}_{7,0}(\tau,x,y)\geql \tilde{p}_{7,2}(\tau,x,y),
& \\
\phantom{\Big\{}
  \tilde{p}_{7,2}(\tau,x,y)\gle \tilde{p}_{7,1}(\tau,x,y)\vee \tilde{p}_{7,0}(\tau,x,y)\geql \gu,
& \\
\phantom{\Big\{}
  \tilde{p}_{7,1}(\tau,x,y)\leftrightarrow \underbrace{\mi{time}(\tau)}_{\tilde{p}_{7,3}(\tau,x,y)}\wedge
                                           \underbrace{\mi{uni}(y)}_{\tilde{p}_{7,4}(\tau,x,y)},
& \\
\phantom{\Big\{}
  \tilde{p}_{7,2}(\tau,x,y)\leftrightarrow 
  \big(\underbrace{\tilde{H}_{X_5}^7(\tilde{s}(\tau),y)}_{\tilde{p}_{7,5}(\tau,x,y)}\geql
       (\underbrace{\exists x\, (\mi{uni}(x)\wedge \tilde{H}_{X_1}(\tau,x)\wedge \tilde{G}_{\mi{medium}_d}(x))\wedge
                    \exists x\, (\mi{uni}(x)\wedge \tilde{H}_{X_2}(\tau,x)\wedge \tilde{G}_{\mi{high}_r}(x))\wedge  
                    \tilde{G}_{\mi{zero}_{\dot{r}}}(y)}_{\tilde{p}_{7,6}(\tau,x,y)})\big)\Big\}
& \quad (\ref{eq0rr1+}), (\ref{eq0rr7+}) \\
\Big\{
  \tilde{p}_{7,0}(\tau,x,y)\geql \gu,
& \\
\phantom{\Big\{}
  \tilde{p}_{7,1}(\tau,x,y)\gle \tilde{p}_{7,2}(\tau,x,y)\vee \tilde{p}_{7,1}(\tau,x,y)\geql \tilde{p}_{7,2}(\tau,x,y)\vee
  \tilde{p}_{7,0}(\tau,x,y)\geql \tilde{p}_{7,2}(\tau,x,y),
& \\
\phantom{\Big\{}
  \tilde{p}_{7,2}(\tau,x,y)\gle \tilde{p}_{7,1}(\tau,x,y)\vee \tilde{p}_{7,0}(\tau,x,y)\geql \gu,
& \\
\phantom{\Big\{}
  \tilde{p}_{7,3}(\tau,x,y)\gle \tilde{p}_{7,4}(\tau,x,y)\vee \tilde{p}_{7,3}(\tau,x,y)\geql \tilde{p}_{7,4}(\tau,x,y)\vee
  \tilde{p}_{7,1}(\tau,x,y)\geql \tilde{p}_{7,4}(\tau,x,y),
& \\
\phantom{\Big\{}
  \tilde{p}_{7,4}(\tau,x,y)\gle \tilde{p}_{7,3}(\tau,x,y)\vee \tilde{p}_{7,1}(\tau,x,y)\geql \tilde{p}_{7,3}(\tau,x,y),
& \\
\phantom{\Big\{}
  \tilde{p}_{7,3}(\tau,x,y)\geql \mi{time}(\tau), \tilde{p}_{7,4}(\tau,x,y)\geql \mi{uni}(y),
& \\
\phantom{\Big\{}
  \tilde{p}_{7,5}(\tau,x,y)\geql \tilde{p}_{7,6}(\tau,x,y)\vee \tilde{p}_{7,2}(\tau,x,y)\geql \gz,
  \tilde{p}_{7,5}(\tau,x,y)\gle \tilde{p}_{7,6}(\tau,x,y)\vee \tilde{p}_{7,6}(\tau,x,y)\gle \tilde{p}_{7,5}(\tau,x,y)\vee
  \tilde{p}_{7,2}(\tau,x,y)\geql \gu,
& \\
\phantom{\Big\{}
  \tilde{p}_{7,5}(\tau,x,y)\geql \tilde{H}_{X_5}^7(\tilde{s}(\tau),y),
& \\
\phantom{\Big\{}
  \tilde{p}_{7,6}(\tau,x,y)\leftrightarrow 
  \big((\underbrace{\exists x\, (\mi{uni}(x)\wedge \tilde{H}_{X_1}(\tau,x)\wedge \tilde{G}_{\mi{medium}_d}(x))\wedge
                    \exists x\, (\mi{uni}(x)\wedge \tilde{H}_{X_2}(\tau,x)\wedge \tilde{G}_{\mi{high}_r}(x))}_{\tilde{p}_{7,7}(\tau,x,y)})\wedge  
       \underbrace{\tilde{G}_{\mi{zero}_{\dot{r}}}(y)}_{\tilde{p}_{7,8}(\tau,x,y)}\big)\Big\}
& \quad (\ref{eq0rr1+}) \\
\Big\{
  \tilde{p}_{7,0}(\tau,x,y)\geql \gu,
& \\
\phantom{\Big\{}
  \tilde{p}_{7,1}(\tau,x,y)\gle \tilde{p}_{7,2}(\tau,x,y)\vee \tilde{p}_{7,1}(\tau,x,y)\geql \tilde{p}_{7,2}(\tau,x,y)\vee
  \tilde{p}_{7,0}(\tau,x,y)\geql \tilde{p}_{7,2}(\tau,x,y),
& \\
\phantom{\Big\{}
  \tilde{p}_{7,2}(\tau,x,y)\gle \tilde{p}_{7,1}(\tau,x,y)\vee \tilde{p}_{7,0}(\tau,x,y)\geql \gu,
& \\
\phantom{\Big\{}
  \tilde{p}_{7,3}(\tau,x,y)\gle \tilde{p}_{7,4}(\tau,x,y)\vee \tilde{p}_{7,3}(\tau,x,y)\geql \tilde{p}_{7,4}(\tau,x,y)\vee
  \tilde{p}_{7,1}(\tau,x,y)\geql \tilde{p}_{7,4}(\tau,x,y),
& \\
\phantom{\Big\{}
  \tilde{p}_{7,4}(\tau,x,y)\gle \tilde{p}_{7,3}(\tau,x,y)\vee \tilde{p}_{7,1}(\tau,x,y)\geql \tilde{p}_{7,3}(\tau,x,y),
& \\
\phantom{\Big\{}
  \tilde{p}_{7,3}(\tau,x,y)\geql \mi{time}(\tau), \tilde{p}_{7,4}(\tau,x,y)\geql \mi{uni}(y),
& \\
\phantom{\Big\{}
  \tilde{p}_{7,5}(\tau,x,y)\geql \tilde{p}_{7,6}(\tau,x,y)\vee \tilde{p}_{7,2}(\tau,x,y)\geql \gz,
  \tilde{p}_{7,5}(\tau,x,y)\gle \tilde{p}_{7,6}(\tau,x,y)\vee \tilde{p}_{7,6}(\tau,x,y)\gle \tilde{p}_{7,5}(\tau,x,y)\vee
  \tilde{p}_{7,2}(\tau,x,y)\geql \gu,
& \\
\phantom{\Big\{}
  \tilde{p}_{7,5}(\tau,x,y)\geql \tilde{H}_{X_5}^7(\tilde{s}(\tau),y),
& \\
\phantom{\Big\{}
  \tilde{p}_{7,7}(\tau,x,y)\gle \tilde{p}_{7,8}(\tau,x,y)\vee \tilde{p}_{7,7}(\tau,x,y)\geql \tilde{p}_{7,8}(\tau,x,y)\vee 
  \tilde{p}_{7,6}(\tau,x,y)\geql \tilde{p}_{7,8}(\tau,x,y),
& \\
\phantom{\Big\{}
  \tilde{p}_{7,8}(\tau,x,y)\gle \tilde{p}_{7,7}(\tau,x,y)\vee \tilde{p}_{7,6}(\tau,x,y)\geql \tilde{p}_{7,7}(\tau,x,y),
& \\
\phantom{\Big\{}
  \tilde{p}_{7,7}(\tau,x,y)\leftrightarrow 
  \big(\underbrace{\exists x\, (\mi{uni}(x)\wedge \tilde{H}_{X_1}(\tau,x)\wedge \tilde{G}_{\mi{medium}_d}(x))}_{\tilde{p}_{7,9}(\tau,x,y)}\wedge
       \underbrace{\exists x\, (\mi{uni}(x)\wedge \tilde{H}_{X_2}(\tau,x)\wedge \tilde{G}_{\mi{high}_r}(x))}_{\tilde{p}_{7,10}(\tau,x,y)}\big),
  \tilde{p}_{7,8}(\tau,x,y)\geql \tilde{G}_{\mi{zero}_{\dot{r}}}(y)\Big\}
& \quad (\ref{eq0rr1+}) \\[1mm]
\hline \hline
\end{IEEEeqnarray*}
\end{minipage}
\vspace{-2mm}
\end{table*}
\begin{table*}[p]
\vspace{-6mm}
\caption{Translation of the formula $\phi_7(\tau,y)\in T_B$ to clausal form}\label{tab77}
\vspace{-6mm}  
\centering
%\hspace*{-12mm}
\begin{minipage}[t]{\linewidth}
\scriptsize   
\begin{IEEEeqnarray*}{LR}
\hline \hline \\[2mm]
\Big\{
  \tilde{p}_{7,0}(\tau,x,y)\geql \gu,
& \\
\phantom{\Big\{}
  \tilde{p}_{7,1}(\tau,x,y)\gle \tilde{p}_{7,2}(\tau,x,y)\vee \tilde{p}_{7,1}(\tau,x,y)\geql \tilde{p}_{7,2}(\tau,x,y)\vee
  \tilde{p}_{7,0}(\tau,x,y)\geql \tilde{p}_{7,2}(\tau,x,y),
& \\
\phantom{\Big\{}
  \tilde{p}_{7,2}(\tau,x,y)\gle \tilde{p}_{7,1}(\tau,x,y)\vee \tilde{p}_{7,0}(\tau,x,y)\geql \gu,
& \\
\phantom{\Big\{}
  \tilde{p}_{7,3}(\tau,x,y)\gle \tilde{p}_{7,4}(\tau,x,y)\vee \tilde{p}_{7,3}(\tau,x,y)\geql \tilde{p}_{7,4}(\tau,x,y)\vee
  \tilde{p}_{7,1}(\tau,x,y)\geql \tilde{p}_{7,4}(\tau,x,y),
& \\
\phantom{\Big\{}
  \tilde{p}_{7,4}(\tau,x,y)\gle \tilde{p}_{7,3}(\tau,x,y)\vee \tilde{p}_{7,1}(\tau,x,y)\geql \tilde{p}_{7,3}(\tau,x,y),
& \\
\phantom{\Big\{}
  \tilde{p}_{7,3}(\tau,x,y)\geql \mi{time}(\tau), \tilde{p}_{7,4}(\tau,x,y)\geql \mi{uni}(y),
& \\
\phantom{\Big\{}
  \tilde{p}_{7,5}(\tau,x,y)\geql \tilde{p}_{7,6}(\tau,x,y)\vee \tilde{p}_{7,2}(\tau,x,y)\geql \gz,
  \tilde{p}_{7,5}(\tau,x,y)\gle \tilde{p}_{7,6}(\tau,x,y)\vee \tilde{p}_{7,6}(\tau,x,y)\gle \tilde{p}_{7,5}(\tau,x,y)\vee
  \tilde{p}_{7,2}(\tau,x,y)\geql \gu,
& \\
\phantom{\Big\{}
  \tilde{p}_{7,5}(\tau,x,y)\geql \tilde{H}_{X_5}^7(\tilde{s}(\tau),y),
& \\
\phantom{\Big\{}
  \tilde{p}_{7,7}(\tau,x,y)\gle \tilde{p}_{7,8}(\tau,x,y)\vee \tilde{p}_{7,7}(\tau,x,y)\geql \tilde{p}_{7,8}(\tau,x,y)\vee 
  \tilde{p}_{7,6}(\tau,x,y)\geql \tilde{p}_{7,8}(\tau,x,y),
& \\
\phantom{\Big\{}
  \tilde{p}_{7,8}(\tau,x,y)\gle \tilde{p}_{7,7}(\tau,x,y)\vee \tilde{p}_{7,6}(\tau,x,y)\geql \tilde{p}_{7,7}(\tau,x,y),
& \\
\phantom{\Big\{}
  \tilde{p}_{7,9}(\tau,x,y)\gle \tilde{p}_{7,10}(\tau,x,y)\vee \tilde{p}_{7,9}(\tau,x,y)\geql \tilde{p}_{7,10}(\tau,x,y)\vee 
  \tilde{p}_{7,7}(\tau,x,y)\geql \tilde{p}_{7,10}(\tau,x,y),
& \\
\phantom{\Big\{}
  \tilde{p}_{7,10}(\tau,x,y)\gle \tilde{p}_{7,9}(\tau,x,y)\vee \tilde{p}_{7,7}(\tau,x,y)\geql \tilde{p}_{7,9}(\tau,x,y),
& \\
\phantom{\Big\{}
  \tilde{p}_{7,9}(\tau,x,y)\leftrightarrow 
  \exists x\, (\underbrace{\mi{uni}(x)\wedge \tilde{H}_{X_1}(\tau,x)\wedge \tilde{G}_{\mi{medium}_d}(x)}_{\tilde{p}_{7,11}(\tau,x,y)}),
  \tilde{p}_{7,10}(\tau,x,y)\leftrightarrow 
  \exists x\, (\underbrace{\mi{uni}(x)\wedge \tilde{H}_{X_2}(\tau,x)\wedge \tilde{G}_{\mi{high}_r}(x)}_{\tilde{p}_{7,12}(\tau,x,y)}),
& \\
\phantom{\Big\{}
  \tilde{p}_{7,8}(\tau,x,y)\geql \tilde{G}_{\mi{zero}_{\dot{r}}}(y)\Big\}
& \quad (\ref{eq0rr6+}) \\
\Big\{
  \tilde{p}_{7,0}(\tau,x,y)\geql \gu,
& \\
\phantom{\Big\{}
  \tilde{p}_{7,1}(\tau,x,y)\gle \tilde{p}_{7,2}(\tau,x,y)\vee \tilde{p}_{7,1}(\tau,x,y)\geql \tilde{p}_{7,2}(\tau,x,y)\vee
  \tilde{p}_{7,0}(\tau,x,y)\geql \tilde{p}_{7,2}(\tau,x,y),
& \\
\phantom{\Big\{}
  \tilde{p}_{7,2}(\tau,x,y)\gle \tilde{p}_{7,1}(\tau,x,y)\vee \tilde{p}_{7,0}(\tau,x,y)\geql \gu,
& \\
\phantom{\Big\{}
  \tilde{p}_{7,3}(\tau,x,y)\gle \tilde{p}_{7,4}(\tau,x,y)\vee \tilde{p}_{7,3}(\tau,x,y)\geql \tilde{p}_{7,4}(\tau,x,y)\vee
  \tilde{p}_{7,1}(\tau,x,y)\geql \tilde{p}_{7,4}(\tau,x,y),
& \\
\phantom{\Big\{}
  \tilde{p}_{7,4}(\tau,x,y)\gle \tilde{p}_{7,3}(\tau,x,y)\vee \tilde{p}_{7,1}(\tau,x,y)\geql \tilde{p}_{7,3}(\tau,x,y),
& \\
\phantom{\Big\{}
  \tilde{p}_{7,3}(\tau,x,y)\geql \mi{time}(\tau), \tilde{p}_{7,4}(\tau,x,y)\geql \mi{uni}(y),
& \\
\phantom{\Big\{}
  \tilde{p}_{7,5}(\tau,x,y)\geql \tilde{p}_{7,6}(\tau,x,y)\vee \tilde{p}_{7,2}(\tau,x,y)\geql \gz,
  \tilde{p}_{7,5}(\tau,x,y)\gle \tilde{p}_{7,6}(\tau,x,y)\vee \tilde{p}_{7,6}(\tau,x,y)\gle \tilde{p}_{7,5}(\tau,x,y)\vee
  \tilde{p}_{7,2}(\tau,x,y)\geql \gu,
& \\
\phantom{\Big\{}
  \tilde{p}_{7,5}(\tau,x,y)\geql \tilde{H}_{X_5}^7(\tilde{s}(\tau),y),
& \\
\phantom{\Big\{}
  \tilde{p}_{7,7}(\tau,x,y)\gle \tilde{p}_{7,8}(\tau,x,y)\vee \tilde{p}_{7,7}(\tau,x,y)\geql \tilde{p}_{7,8}(\tau,x,y)\vee 
  \tilde{p}_{7,6}(\tau,x,y)\geql \tilde{p}_{7,8}(\tau,x,y),
& \\
\phantom{\Big\{}
  \tilde{p}_{7,8}(\tau,x,y)\gle \tilde{p}_{7,7}(\tau,x,y)\vee \tilde{p}_{7,6}(\tau,x,y)\geql \tilde{p}_{7,7}(\tau,x,y),
& \\
\phantom{\Big\{}
  \tilde{p}_{7,9}(\tau,x,y)\gle \tilde{p}_{7,10}(\tau,x,y)\vee \tilde{p}_{7,9}(\tau,x,y)\geql \tilde{p}_{7,10}(\tau,x,y)\vee 
  \tilde{p}_{7,7}(\tau,x,y)\geql \tilde{p}_{7,10}(\tau,x,y),
& \\
\phantom{\Big\{}
  \tilde{p}_{7,10}(\tau,x,y)\gle \tilde{p}_{7,9}(\tau,x,y)\vee \tilde{p}_{7,7}(\tau,x,y)\geql \tilde{p}_{7,9}(\tau,x,y),
& \\
\phantom{\Big\{}
  \tilde{p}_{7,9}(\tau,x,y)\geql \exists x\, \tilde{p}_{7,11}(\tau,x,y),
  \tilde{p}_{7,11}(\tau,x,y)\leftrightarrow \underbrace{\mi{uni}(x)}_{\tilde{p}_{7,13}(\tau,x,y)}\wedge 
                                            \underbrace{\tilde{H}_{X_1}(\tau,x)\wedge \tilde{G}_{\mi{medium}_d}(x)}_{\tilde{p}_{7,14}(\tau,x,y)}, 
& \\
\phantom{\Big\{}
  \tilde{p}_{7,10}(\tau,x,y)\geql \exists x\, \tilde{p}_{7,12}(\tau,x,y),
  \tilde{p}_{7,12}(\tau,x,y)\leftrightarrow \underbrace{\mi{uni}(x)}_{\tilde{p}_{7,15}(\tau,x,y)}\wedge
                                            \underbrace{\tilde{H}_{X_2}(\tau,x)\wedge \tilde{G}_{\mi{high}_r}(x)}_{\tilde{p}_{7,16}(\tau,x,y)},
& \\
\phantom{\Big\{}
  \tilde{p}_{7,8}(\tau,x,y)\geql \tilde{G}_{\mi{zero}_{\dot{r}}}(y)\Big\}
& \quad (\ref{eq0rr1+}) \\[1mm]
\hline \hline
\end{IEEEeqnarray*}
\end{minipage}     
\vspace{-2mm}      
\end{table*}        
\begin{table*}[p]
\vspace{-6mm}
\caption{Translation of the formula $\phi_7(\tau,y)\in T_B$ to clausal form}\label{tab777}
\vspace{-6mm}
\centering
%\hspace*{-12mm}
\begin{minipage}[t]{\linewidth}
\scriptsize     
\begin{IEEEeqnarray*}{LR}
\hline \hline \\[2mm]
\Big\{
  \tilde{p}_{7,0}(\tau,x,y)\geql \gu,
& \\
\phantom{\Big\{}
  \tilde{p}_{7,1}(\tau,x,y)\gle \tilde{p}_{7,2}(\tau,x,y)\vee \tilde{p}_{7,1}(\tau,x,y)\geql \tilde{p}_{7,2}(\tau,x,y)\vee
  \tilde{p}_{7,0}(\tau,x,y)\geql \tilde{p}_{7,2}(\tau,x,y),
& \\
\phantom{\Big\{}
  \tilde{p}_{7,2}(\tau,x,y)\gle \tilde{p}_{7,1}(\tau,x,y)\vee \tilde{p}_{7,0}(\tau,x,y)\geql \gu,
& \\
\phantom{\Big\{}
  \tilde{p}_{7,3}(\tau,x,y)\gle \tilde{p}_{7,4}(\tau,x,y)\vee \tilde{p}_{7,3}(\tau,x,y)\geql \tilde{p}_{7,4}(\tau,x,y)\vee
  \tilde{p}_{7,1}(\tau,x,y)\geql \tilde{p}_{7,4}(\tau,x,y),
& \\
\phantom{\Big\{}
  \tilde{p}_{7,4}(\tau,x,y)\gle \tilde{p}_{7,3}(\tau,x,y)\vee \tilde{p}_{7,1}(\tau,x,y)\geql \tilde{p}_{7,3}(\tau,x,y),
& \\
\phantom{\Big\{}
  \tilde{p}_{7,3}(\tau,x,y)\geql \mi{time}(\tau), \tilde{p}_{7,4}(\tau,x,y)\geql \mi{uni}(y),
& \\
\phantom{\Big\{}
  \tilde{p}_{7,5}(\tau,x,y)\geql \tilde{p}_{7,6}(\tau,x,y)\vee \tilde{p}_{7,2}(\tau,x,y)\geql \gz,
  \tilde{p}_{7,5}(\tau,x,y)\gle \tilde{p}_{7,6}(\tau,x,y)\vee \tilde{p}_{7,6}(\tau,x,y)\gle \tilde{p}_{7,5}(\tau,x,y)\vee
  \tilde{p}_{7,2}(\tau,x,y)\geql \gu,
& \\
\phantom{\Big\{}
  \tilde{p}_{7,5}(\tau,x,y)\geql \tilde{H}_{X_5}^7(\tilde{s}(\tau),y),
& \\
\phantom{\Big\{}
  \tilde{p}_{7,7}(\tau,x,y)\gle \tilde{p}_{7,8}(\tau,x,y)\vee \tilde{p}_{7,7}(\tau,x,y)\geql \tilde{p}_{7,8}(\tau,x,y)\vee 
  \tilde{p}_{7,6}(\tau,x,y)\geql \tilde{p}_{7,8}(\tau,x,y),
& \\
\phantom{\Big\{}
  \tilde{p}_{7,8}(\tau,x,y)\gle \tilde{p}_{7,7}(\tau,x,y)\vee \tilde{p}_{7,6}(\tau,x,y)\geql \tilde{p}_{7,7}(\tau,x,y),
& \\
\phantom{\Big\{}
  \tilde{p}_{7,9}(\tau,x,y)\gle \tilde{p}_{7,10}(\tau,x,y)\vee \tilde{p}_{7,9}(\tau,x,y)\geql \tilde{p}_{7,10}(\tau,x,y)\vee 
  \tilde{p}_{7,7}(\tau,x,y)\geql \tilde{p}_{7,10}(\tau,x,y),
& \\
\phantom{\Big\{}
  \tilde{p}_{7,10}(\tau,x,y)\gle \tilde{p}_{7,9}(\tau,x,y)\vee \tilde{p}_{7,7}(\tau,x,y)\geql \tilde{p}_{7,9}(\tau,x,y),
& \\
\phantom{\Big\{}
  \tilde{p}_{7,9}(\tau,x,y)\geql \exists x\, \tilde{p}_{7,11}(\tau,x,y),
& \\
\phantom{\Big\{}
  \tilde{p}_{7,13}(\tau,x,y)\gle \tilde{p}_{7,14}(\tau,x,y)\vee \tilde{p}_{7,13}(\tau,x,y)\geql \tilde{p}_{7,14}(\tau,x,y)\vee
  \tilde{p}_{7,11}(\tau,x,y)\geql \tilde{p}_{7,14}(\tau,x,y),
& \\
\phantom{\Big\{}
  \tilde{p}_{7,14}(\tau,x,y)\gle \tilde{p}_{7,13}(\tau,x,y)\vee \tilde{p}_{7,11}(\tau,x,y)\geql \tilde{p}_{7,13}(\tau,x,y),
& \\
\phantom{\Big\{}
  \tilde{p}_{7,13}(\tau,x,y)\geql \mi{uni}(x), 
  \tilde{p}_{7,14}(\tau,x,y)\leftrightarrow \underbrace{\tilde{H}_{X_1}(\tau,x)}_{\tilde{p}_{7,17}(\tau,x,y)}\wedge
                                            \underbrace{\tilde{G}_{\mi{medium}_d}(x)}_{\tilde{p}_{7,18}(\tau,x,y)},
& \\
\phantom{\Big\{}
  \tilde{p}_{7,10}(\tau,x,y)\geql \exists x\, \tilde{p}_{7,12}(\tau,x,y),
& \\
\phantom{\Big\{}
  \tilde{p}_{7,15}(\tau,x,y)\gle \tilde{p}_{7,16}(\tau,x,y)\vee \tilde{p}_{7,15}(\tau,x,y)\geql \tilde{p}_{7,16}(\tau,x,y)\vee
  \tilde{p}_{7,12}(\tau,x,y)\geql \tilde{p}_{7,16}(\tau,x,y),
& \\
\phantom{\Big\{}
  \tilde{p}_{7,16}(\tau,x,y)\gle \tilde{p}_{7,15}(\tau,x,y)\vee \tilde{p}_{7,12}(\tau,x,y)\geql \tilde{p}_{7,15}(\tau,x,y),
& \\
\phantom{\Big\{}
  \tilde{p}_{7,15}(\tau,x,y)\geql \mi{uni}(x),
  \tilde{p}_{7,16}(\tau,x,y)\leftrightarrow \underbrace{\tilde{H}_{X_2}(\tau,x)}_{\tilde{p}_{7,19}(\tau,x,y)}\wedge
                                            \underbrace{\tilde{G}_{\mi{high}_r}(x)}_{\tilde{p}_{7,20}(\tau,x,y)},
& \\
\phantom{\Big\{}
  \tilde{p}_{7,8}(\tau,x,y)\geql \tilde{G}_{\mi{zero}_{\dot{r}}}(y)\Big\}
& \quad (\ref{eq0rr1+}) \\[1mm]
\hline \hline   
\end{IEEEeqnarray*}
\end{minipage}     
\vspace{-2mm}
\end{table*}     
\begin{table*}[p]
\vspace{-6mm}
\caption{Translation of the formula $\phi_7(\tau,y)\in T_B$ to clausal form}\label{tab7777}
\vspace{-6mm}
\centering
%\hspace*{-12mm}
\begin{minipage}[t]{\linewidth}
\scriptsize     
\begin{IEEEeqnarray*}{LR}
\hline \hline \\[2mm]
\Big\{
  \tilde{p}_{7,0}(\tau,x,y)\geql \gu,
& \\
\phantom{\Big\{}
  \tilde{p}_{7,1}(\tau,x,y)\gle \tilde{p}_{7,2}(\tau,x,y)\vee \tilde{p}_{7,1}(\tau,x,y)\geql \tilde{p}_{7,2}(\tau,x,y)\vee
  \tilde{p}_{7,0}(\tau,x,y)\geql \tilde{p}_{7,2}(\tau,x,y),
& \\
\phantom{\Big\{}
  \tilde{p}_{7,2}(\tau,x,y)\gle \tilde{p}_{7,1}(\tau,x,y)\vee \tilde{p}_{7,0}(\tau,x,y)\geql \gu,
& \\
\phantom{\Big\{}
  \tilde{p}_{7,3}(\tau,x,y)\gle \tilde{p}_{7,4}(\tau,x,y)\vee \tilde{p}_{7,3}(\tau,x,y)\geql \tilde{p}_{7,4}(\tau,x,y)\vee
  \tilde{p}_{7,1}(\tau,x,y)\geql \tilde{p}_{7,4}(\tau,x,y),
& \\
\phantom{\Big\{}
  \tilde{p}_{7,4}(\tau,x,y)\gle \tilde{p}_{7,3}(\tau,x,y)\vee \tilde{p}_{7,1}(\tau,x,y)\geql \tilde{p}_{7,3}(\tau,x,y),
& \\
\phantom{\Big\{}
  \tilde{p}_{7,3}(\tau,x,y)\geql \mi{time}(\tau), \tilde{p}_{7,4}(\tau,x,y)\geql \mi{uni}(y),
& \\
\phantom{\Big\{}
  \tilde{p}_{7,5}(\tau,x,y)\geql \tilde{p}_{7,6}(\tau,x,y)\vee \tilde{p}_{7,2}(\tau,x,y)\geql \gz,
  \tilde{p}_{7,5}(\tau,x,y)\gle \tilde{p}_{7,6}(\tau,x,y)\vee \tilde{p}_{7,6}(\tau,x,y)\gle \tilde{p}_{7,5}(\tau,x,y)\vee
  \tilde{p}_{7,2}(\tau,x,y)\geql \gu,
& \\
\phantom{\Big\{}
  \tilde{p}_{7,5}(\tau,x,y)\geql \tilde{H}_{X_5}^7(\tilde{s}(\tau),y),
& \\
\phantom{\Big\{}
  \tilde{p}_{7,7}(\tau,x,y)\gle \tilde{p}_{7,8}(\tau,x,y)\vee \tilde{p}_{7,7}(\tau,x,y)\geql \tilde{p}_{7,8}(\tau,x,y)\vee 
  \tilde{p}_{7,6}(\tau,x,y)\geql \tilde{p}_{7,8}(\tau,x,y),
& \\
\phantom{\Big\{}
  \tilde{p}_{7,8}(\tau,x,y)\gle \tilde{p}_{7,7}(\tau,x,y)\vee \tilde{p}_{7,6}(\tau,x,y)\geql \tilde{p}_{7,7}(\tau,x,y),
& \\
\phantom{\Big\{}
  \tilde{p}_{7,9}(\tau,x,y)\gle \tilde{p}_{7,10}(\tau,x,y)\vee \tilde{p}_{7,9}(\tau,x,y)\geql \tilde{p}_{7,10}(\tau,x,y)\vee 
  \tilde{p}_{7,7}(\tau,x,y)\geql \tilde{p}_{7,10}(\tau,x,y),
& \\
\phantom{\Big\{}
  \tilde{p}_{7,10}(\tau,x,y)\gle \tilde{p}_{7,9}(\tau,x,y)\vee \tilde{p}_{7,7}(\tau,x,y)\geql \tilde{p}_{7,9}(\tau,x,y),
& \\
\phantom{\Big\{}
  \tilde{p}_{7,9}(\tau,x,y)\geql \exists x\, \tilde{p}_{7,11}(\tau,x,y),
& \\
\phantom{\Big\{}
  \tilde{p}_{7,13}(\tau,x,y)\gle \tilde{p}_{7,14}(\tau,x,y)\vee \tilde{p}_{7,13}(\tau,x,y)\geql \tilde{p}_{7,14}(\tau,x,y)\vee
  \tilde{p}_{7,11}(\tau,x,y)\geql \tilde{p}_{7,14}(\tau,x,y),
& \\
\phantom{\Big\{}
  \tilde{p}_{7,14}(\tau,x,y)\gle \tilde{p}_{7,13}(\tau,x,y)\vee \tilde{p}_{7,11}(\tau,x,y)\geql \tilde{p}_{7,13}(\tau,x,y),
& \\
\phantom{\Big\{}
  \tilde{p}_{7,13}(\tau,x,y)\geql \mi{uni}(x), 
& \\
\phantom{\Big\{}
  \tilde{p}_{7,17}(\tau,x,y)\gle \tilde{p}_{7,18}(\tau,x,y)\vee \tilde{p}_{7,17}(\tau,x,y)\geql \tilde{p}_{7,18}(\tau,x,y)\vee
  \tilde{p}_{7,14}(\tau,x,y)\geql \tilde{p}_{7,18}(\tau,x,y),
& \\
\phantom{\Big\{}
  \tilde{p}_{7,18}(\tau,x,y)\gle \tilde{p}_{7,17}(\tau,x,y)\vee \tilde{p}_{7,14}(\tau,x,y)\geql \tilde{p}_{7,17}(\tau,x,y),
& \\
\phantom{\Big\{}
  \tilde{p}_{7,17}(\tau,x,y)\geql \tilde{H}_{X_1}(\tau,x), 
  \tilde{p}_{7,18}(\tau,x,y)\geql \tilde{G}_{\mi{medium}_d}(x),
& \\
\phantom{\Big\{}
  \tilde{p}_{7,10}(\tau,x,y)\geql \exists x\, \tilde{p}_{7,12}(\tau,x,y),
& \\
\phantom{\Big\{}
  \tilde{p}_{7,15}(\tau,x,y)\gle \tilde{p}_{7,16}(\tau,x,y)\vee \tilde{p}_{7,15}(\tau,x,y)\geql \tilde{p}_{7,16}(\tau,x,y)\vee
  \tilde{p}_{7,12}(\tau,x,y)\geql \tilde{p}_{7,16}(\tau,x,y),
& \\
\phantom{\Big\{}
  \tilde{p}_{7,16}(\tau,x,y)\gle \tilde{p}_{7,15}(\tau,x,y)\vee \tilde{p}_{7,12}(\tau,x,y)\geql \tilde{p}_{7,15}(\tau,x,y),
& \\
\phantom{\Big\{}
  \tilde{p}_{7,15}(\tau,x,y)\geql \mi{uni}(x),
& \\
\phantom{\Big\{}
  \tilde{p}_{7,19}(\tau,x,y)\gle \tilde{p}_{7,20}(\tau,x,y)\vee \tilde{p}_{7,19}(\tau,x,y)\geql \tilde{p}_{7,20}(\tau,x,y)\vee
  \tilde{p}_{7,16}(\tau,x,y)\geql \tilde{p}_{7,20}(\tau,x,y),
& \\
\phantom{\Big\{}
  \tilde{p}_{7,20}(\tau,x,y)\gle \tilde{p}_{7,19}(\tau,x,y)\vee \tilde{p}_{7,16}(\tau,x,y)\geql \tilde{p}_{7,19}(\tau,x,y),
& \\
\phantom{\Big\{}
  \tilde{p}_{7,19}(\tau,x,y)\geql \tilde{H}_{X_2}(\tau,x),
  \tilde{p}_{7,20}(\tau,x,y)\geql \tilde{G}_{\mi{high}_r}(x),
& \\
\phantom{\Big\{}
  \tilde{p}_{7,8}(\tau,x,y)\geql \tilde{G}_{\mi{zero}_{\dot{r}}}(y)\Big\} \\[1mm]
\hline \hline   
\end{IEEEeqnarray*}
\end{minipage}     
\vspace{0mm}
\end{table*}
The translations of $\phi_1(\tau,y)$ (one input fuzzy variable) and $\phi_7(\tau,y)$ (two input fuzzy variables) serve as prototypes; 
the translations of the other fuzzy rules -- formulae from $T_B$ can slightly be cloned from them.
Also the translations of the remaining aggregation rules from $T_B$ and of the theory $T_D$ to clausal form
can be performed straightforwardly, denoted as $S_B$ and $S_D$, respectively.
For convenient experimenting with fuzzy inference using the fuzzy rule base $B$,
we have devised a rule-based system%
\footnote{Download link: www.dai.fmph.uniba.sk/$\sim$guller/tfs18A.clp} 
in the language (IDE) CLIPS \cite{GIRI98}. 
As the initial state of fuzzy derivation, we suppose that 
the temperature is low; the density is high; the rotation is low;
the first derivative of the temperature is zero;
the first derivative of the density is zero;
the first derivative of the rotation is positive;
hence, the initial variable assignment 
$e_0=\{(X_0,\mi{low}_t),(X_1,\mi{high}_d),(X_2,\mi{low}_r),
       (X_3,\mi{zero}_{\dot{t}}),                                                                                                  \linebreak[4]
       (X_4,\mi{zero}_{\dot{d}}),(X_5,\mi{positive}_{\dot{r}})\}$.
In Table \ref{tab9}, the first thirteen steps are listed.
Concerning the problems posed in the previous paragraph,
e.g. for $X_2$ (rotation) and some contradictory fuzzy set $A=\{\frac{1}{0},\frac{0.5}{1},\frac{1}{2},\frac{0.5}{3},\frac{1}{4}\}$ 
(added to $\mbb{A}$), 
we reach $e_{11}(X_2)=A$ at the 11th step.
We see that $e_{12}=e_{13}$; hence, the fuzzy derivation is stable.
We may slightly change the fuzzy rule base $B$ by deleting the friction rules to get a fuzzy rule base $B'$, and subsequently $T_{B'}$, $S_{B'}$%
\footnote{Download link: www.dai.fmph.uniba.sk/$\sim$guller/tfs18B.clp}.
Starting from the same initial state as with $B$, in Table \ref{tab99}, the first seventeen steps are listed.
We see that there exists a $12$-cycle in the fuzzy derivation.
So, using our framework, we obtain the following reductions of the problems to deduction ones with respect to ${\mc K}$:
$\phi_r=\exists \tau\, (\mi{time}(\tau)\wedge \forall x\, (\mi{uni}(x)\rightarrow \tilde{H}_{X_2}(\tau,x)\geql \tilde{G}_A(x)))$,
$T_D\cup S_U\cup S_\mbb{A}\cup T_B\cup S_{e_0}(\tau/\tilde{z})\models_{\mc K} \phi_r$ (reachability);
$\phi_s=\exists \tau\, (\mi{time}(\tau)\wedge
                        \bigwedge_{X\in \mbb{X}} \forall x\, (\mi{uni}(x)\rightarrow \tilde{H}_X(\tau,x)\geql \tilde{H}_X(\tilde{s}(\tau),x)))$,
$T_D\cup S_U\cup S_\mbb{A}\cup T_B\cup S_{e_0}(\tau/\tilde{z})\models_{\mc K} \phi_s$ (stability);
$\phi_{12-c}=\exists \tau\, (\mi{time}(\tau)\wedge
                             \bigwedge_{X\in \mbb{X}} \forall x\, (\mi{uni}(x)\rightarrow 
                                                                   \tilde{H}_X(\tau,x)\geql \tilde{H}_X(\tilde{s}^{12}(\tau),x)))$,
$T_D\cup S_U\cup S_\mbb{A}\cup T_{B'}\cup S_{e_0}(\tau/\tilde{z})\models_{\mc K} \phi_{12-c}$ (the existence of a $12$-cycle).
Subsequently, using Theorem \ref{T1}, we can reduce the deduction problems to unsatisfiability ones in ${\mc K}$.
By Lemma \ref{le11} for $n_0$, $\tilde{p}_{(n_0,0)}$, $\phi_r$, $\phi_s$, $\phi_{12-c}$, 
$\tilde{p}_{(n_0,0)}\leftrightarrow \phi_r$, 
$\tilde{p}_{(n_0,0)}\leftrightarrow \phi_s$,
$\tilde{p}_{(n_0,0)}\leftrightarrow \phi_{12-c}$ can be translated to $S_r$, $S_s$, $S_{12-c}$, respectively.
The resulting clausal theories are as follows.
$S_D\cup S_U\cup S_\mbb{A}\cup S_B\cup S_{e_0}(\tau/\tilde{z})\cup \{\tilde{p}_{(n_0,0)}(\bar{x})\gle \gu\}\cup S_r$,
$S_D\cup S_U\cup S_\mbb{A}\cup S_B\cup S_{e_0}(\tau/\tilde{z})\cup \{\tilde{p}_{(n_0,0)}(\bar{x})\gle \gu\}\cup S_s$, 
$S_D\cup S_U\cup S_\mbb{A}\cup S_{B'}\cup S_{e_0}(\tau/\tilde{z})\cup \{\tilde{p}_{(n_0,0)}(\bar{x})\gle \gu\}\cup S_{12-c}$.
Technically, we assume that the translations $S_D$, $S_B$, $S_{B'}$ are done with the offset $n_0+1$ according to the proof of Theorem \ref{T1}.
In our example, on condition that ${\mc L}^*$ is a finite expansion 
of ${\mc L}\cup \{\tilde{f}_0\}\cup \tilde{\mbb{Z}}\cup \tilde{\mbb{D}}\cup \tilde{\mbb{G}}\cup \tilde{\mbb{H}}$, we obtain that 
$\mbb{U}$, $\mbb{A}$, $\mbb{X}$, $\tilde{\mbb{U}}$, $C_{\mc L}$ are finite;
the set of function symbols $\{\tilde{f}_0\}\cup \tilde{\mbb{Z}}$ is finite;
$T_D$, $S_D$, $\{\tilde{p}_\mbbm{j} \,|\, \mbbm{j}\in J_{T_D}\}$ are finite;
$T_B$, $S_B$, $\{\tilde{p}_\mbbm{j} \,|\, \mbbm{j}\in J_{T_B}\}$ are finite; 
$S_\mbb{A}$, $S_{e_0}(\tau/\tilde{z})$ are finite;
$S_r$, $S_s$, $S_{12-c}$, 
$\{\tilde{p}_\mbbm{j} \,|\, \mbbm{j}\in J_{\phi_r}\}$, $\{\tilde{p}_\mbbm{j} \,|\, \mbbm{j}\in J_{\phi_s}\}$,
$\{\tilde{p}_\mbbm{j} \,|\, \mbbm{j}\in J_{\phi_{12-c}}\}$ are finite;
the set of predicate symbols 
$\tilde{\mbb{D}}\cup \tilde{\mbb{G}}\cup \tilde{\mbb{H}}\cup 
 \{\tilde{p}_\mbbm{j} \,|\, \mbbm{j}\in J_{T_D}\cup J_{T_B}\cup J_{\phi_{r,s,12-c}}\}$ is finite; 
$\{\mi{uni}(\tilde{u})\geql \gu \,|\, \tilde{u}\in \tilde{\mbb{U}}\}$ is finite;
only $\{\mi{uni}(t)\geql \gz \,|\, t\in \mi{GTerm}_{{\mc L}^*}-\tilde{\mbb{U}}\}$, and hence, $S_U$ are countably infinite.
The resulting clausal theories should then be tested for unsatisfiability by some suitable generalised hyperresolution proof method 
of that one in \cite{Guller2016c}, which is a subject of our further research. 
\begin{table*}[p]
\vspace{-6mm}
\caption{Fuzzy derivation using the fuzzy rule base $B$}\label{tab9}
\vspace{-6mm}
\centering
\begin{minipage}[t]{\linewidth-25mm}
\scriptsize
\begin{IEEEeqnarray*}{RLLLLLLL}
\hline \hline \\[1mm]
\text{\bf Time} & \qquad & \IEEEeqnarraymulticol{6}{l}{\text{\bf State}\ (X_0\ X_1\ X_2\ X_3\ X_4\ X_5)} \\[0mm] 
\hline \\[1mm]
0:  & & (\mi{low}_t\ & \mi{high}_d\ & \mi{low}_r\ & \mi{zero}_{\dot{t}}\ & \mi{zero}_{\dot{d}}\ & \mi{positive}_{\dot{r}}) \\[1mm]
1:  & & \big((1\ 0.5\ 0.5\ 0.5\ 0.5)\ & (0.5\ 0.5\ 0.5\ 0.5\ 1)\ & (0.5\ 0.5\ 1\ 0.5\ 0.5)\ & 
             (1\ 0.5\ 1\ 0.5\ 0.5)\   & (0.5\ 0.5\ 1\ 0.5\ 0.5)\ & (0\ 0.5\ 0.5\ 0.5\ 1)\big) \\[1mm]
2:  & & \big((1\ 0.5\ 0.5\ 0.5\ 0.5)\ & (0.5\ 0.5\ 0.5\ 0.5\ 1)\ & (0.5\ 0.5\ 0.5\ 0.5\ 1)\ &
             (0.5\ 0.5\ 1\ 0.5\ 1)\   & (0.5\ 0.5\ 1\ 0.5\ 0.5)\ & (0.5\ 0.5\ 0.5\ 0.5\ 1)\big) \\[1mm]
3:  & & \big((1\ 0.5\ 1\ 0.5\ 0.5)\   & (0.5\ 0.5\ 0.5\ 0.5\ 1)\ & (0.5\ 0.5\ 0.5\ 0.5\ 1)\ &
             (0.5\ 0.5\ 1\ 0.5\ 1)\   & (0.5\ 0.5\ 1\ 0.5\ 0.5)\ & (1\ 0.5\ 0.5\ 0.5\ 0.5)\big) \\[1mm]
4:  & & \big((1\ 0.5\ 1\ 0.5\ 1)\     & (0.5\ 0.5\ 1\ 0.5\ 1)\   & (0.5\ 0.5\ 1\ 0.5\ 0.5)\ &
             (0.5\ 0.5\ 1\ 0.5\ 1)\   & (0.5\ 0.5\ 1\ 0.5\ 0.5)\ & (1\ 0.5\ 0.5\ 0.5\ 0.5)\big) \\[1mm]
5:  & & \big((1\ 0.5\ 1\ 0.5\ 1)\     & (1\ 0.5\ 1\ 0.5\ 1)\     & (1\ 0.5\ 0.5\ 0.5\ 0.5)\ &
             (1\ 0.5\ 1\ 0.5\ 1)\     & (0.5\ 0.5\ 1\ 0.5\ 0.5)\ & (1\ 0.5\ 0.5\ 0.5\ 0.5)\big) \\[1mm]
6:  & & \big((1\ 0.5\ 1\ 0.5\ 1)\     & (1\ 0.5\ 1\ 0.5\ 1)\     & (1\ 0.5\ 0.5\ 0.5\ 0.5)\ &
             (1\ 0.5\ 1\ 0.5\ 1)\     & (0.5\ 0.5\ 1\ 0.5\ 0.5)\ & (0.5\ 0.5\ 0.5\ 0.5\ 1)\big) \\[1mm]
7:  & & \big((1\ 0.5\ 1\ 0.5\ 1)\     & (1\ 0.5\ 1\ 0.5\ 1)\     & (0.5\ 0.5\ 1\ 0.5\ 0.5)\ &
             (1\ 0.5\ 1\ 0.5\ 1)\     & (0.5\ 0.5\ 1\ 0.5\ 0.5)\ & (0.5\ 0.5\ 0.5\ 0.5\ 1)\big) \\[1mm]
8:  & & \big((1\ 0.5\ 1\ 0.5\ 1)\     & (1\ 0.5\ 1\ 0.5\ 1)\     & (0.5\ 0.5\ 0.5\ 0.5\ 1)\ &
             (1\ 0.5\ 1\ 0.5\ 1)\     & (0.5\ 0.5\ 1\ 0.5\ 0.5)\ & (0.5\ 0.5\ 0.5\ 0.5\ 1)\big) \\[1mm]
9:  & & \big((1\ 0.5\ 1\ 0.5\ 1)\     & (1\ 0.5\ 1\ 0.5\ 1)\     & (0.5\ 0.5\ 0.5\ 0.5\ 1)\ &
             (1\ 0.5\ 1\ 0.5\ 1)\     & (0.5\ 0.5\ 1\ 0.5\ 0.5)\ & (1\ 0.5\ 1\ 0.5\ 0.5)\big) \\[1mm]
10: & & \big((1\ 0.5\ 1\ 0.5\ 1)\     & (1\ 0.5\ 1\ 0.5\ 1)\     & (0.5\ 0.5\ 1\ 0.5\ 1)\   &
             (1\ 0.5\ 1\ 0.5\ 1)\     & (0.5\ 0.5\ 1\ 0.5\ 0.5)\ & (1\ 0.5\ 1\ 0.5\ 0.5)\big) \\[1mm]
11: & & \big((1\ 0.5\ 1\ 0.5\ 1)\     & (1\ 0.5\ 1\ 0.5\ 1)\     & (1\ 0.5\ 1\ 0.5\ 1)\     &
             (1\ 0.5\ 1\ 0.5\ 1)\     & (0.5\ 0.5\ 1\ 0.5\ 0.5)\ & (1\ 0.5\ 1\ 0.5\ 0.5)\big) \\[1mm]
{\bf {*12}}: 
    & & {\bf \big((1\ 0.5\ 1\ 0.5\ 1)}\     & {\bf (1\ 0.5\ 1\ 0.5\ 1)}\     & {\bf (1\ 0.5\ 1\ 0.5\ 1)}\     &
        {\bf      (1\ 0.5\ 1\ 0.5\ 1)}\     & {\bf (0.5\ 0.5\ 1\ 0.5\ 0.5)}\ & {\bf (1\ 0.5\ 1\ 0.5\ 1)\big)} \\[1mm]
{\bf *13}: 
    & & {\bf \big((1\ 0.5\ 1\ 0.5\ 1)}\     & {\bf (1\ 0.5\ 1\ 0.5\ 1)}\     & {\bf (1\ 0.5\ 1\ 0.5\ 1)}\     &
        {\bf      (1\ 0.5\ 1\ 0.5\ 1)}\     & {\bf (0.5\ 0.5\ 1\ 0.5\ 0.5)}\ & {\bf (1\ 0.5\ 1\ 0.5\ 1)\big)} \\[1mm]
\hline \hline
\end{IEEEeqnarray*}
\end{minipage}     
\vspace{-2mm}
\end{table*}        
\begin{table*}[p]
\vspace{-6mm}
\caption{Fuzzy derivation using the fuzzy rule base $B'$}\label{tab99}
\vspace{-6mm}
\centering
\begin{minipage}[t]{\linewidth-25mm}
\scriptsize
\begin{IEEEeqnarray*}{RLLLLLLL}
\hline \hline \\[1mm]
\text{\bf Time} & \qquad & \IEEEeqnarraymulticol{6}{l}{\text{\bf State}\ (X_0\ X_1\ X_2\ X_3\ X_4\ X_5)} \\[0mm] 
\hline \\[1mm]
0:  & & (\mi{low}_t\ & \mi{high}_d\ & \mi{low}_r\ & \mi{zero}_{\dot{t}}\ & \mi{zero}_{\dot{d}}\ & \mi{positive}_{\dot{r}}) \\[1mm]
1:  & & \big((1\ 0.5\ 0.5\ 0.5\ 0.5)\ & (0.5\ 0.5\ 0.5\ 0.5\ 1)\ & (0.5\ 0.5\ 1\ 0.5\ 0.5)\ & 
             (1\ 0.5\ 1\ 0.5\ 0.5)\   & (0.5\ 0.5\ 1\ 0.5\ 0.5)\ & (0\ 0.5\ 0.5\ 0.5\ 1)\big) \\[1mm]
2:  & & \big((1\ 0.5\ 0.5\ 0.5\ 0.5)\ & (0.5\ 0.5\ 0.5\ 0.5\ 1)\ & (0.5\ 0.5\ 0.5\ 0.5\ 1)\ &
             (0.5\ 0.5\ 1\ 0.5\ 1)\   & (0.5\ 0.5\ 1\ 0.5\ 0.5)\ & (0.5\ 0.5\ 0.5\ 0.5\ 1)\big) \\[1mm]
3:  & & \big((1\ 0.5\ 1\ 0.5\ 0.5)\   & (0.5\ 0.5\ 0.5\ 0.5\ 1)\ & (0.5\ 0.5\ 0.5\ 0.5\ 1)\ &
             (0.5\ 0.5\ 1\ 0.5\ 1)\   & (0.5\ 0.5\ 1\ 0.5\ 0.5)\ & (1\ 0.5\ 0.5\ 0.5\ 0.5)\big) \\[1mm]
4:  & & \big((1\ 0.5\ 1\ 0.5\ 1)\     & (0.5\ 0.5\ 1\ 0.5\ 1)\   & (0.5\ 0.5\ 1\ 0.5\ 0.5)\ &
             (0.5\ 0.5\ 1\ 0.5\ 1)\   & (0.5\ 0.5\ 1\ 0.5\ 0.5)\ & (1\ 0.5\ 0.5\ 0.5\ 0.5)\big) \\[1mm]
{\bf *5}:  
    & & {\bf \big((1\ 0.5\ 1\ 0.5\ 1)}\     & {\bf (1\ 0.5\ 1\ 0.5\ 1)}\     & {\bf (1\ 0.5\ 0.5\ 0.5\ 0.5)}\ &
        {\bf      (1\ 0.5\ 1\ 0.5\ 1)}\     & {\bf (0.5\ 0.5\ 1\ 0.5\ 0.5)}\ & {\bf (1\ 0.5\ 0.5\ 0.5\ 0.5)\big)} \\[1mm]
\hline \\[1mm]
6:  & & \big((1\ 0.5\ 1\ 0.5\ 1)\     & (1\ 0.5\ 1\ 0.5\ 1)\     & (1\ 0.5\ 0.5\ 0.5\ 0.5)\ &
             (1\ 0.5\ 1\ 0.5\ 1)\     & (0.5\ 0.5\ 1\ 0.5\ 0.5)\ & (0.5\ 0.5\ 0.5\ 0.5\ 1)\big) \\[1mm]
7:  & & \big((1\ 0.5\ 1\ 0.5\ 1)\     & (1\ 0.5\ 1\ 0.5\ 1)\     & (0.5\ 0.5\ 1\ 0.5\ 0.5)\ &
             (1\ 0.5\ 1\ 0.5\ 1)\     & (0.5\ 0.5\ 1\ 0.5\ 0.5)\ & (0.5\ 0.5\ 0.5\ 0.5\ 1)\big) \\[1mm]
8:  & & \big((1\ 0.5\ 1\ 0.5\ 1)\     & (1\ 0.5\ 1\ 0.5\ 1)\     & (0.5\ 0.5\ 0.5\ 0.5\ 1)\ &
             (1\ 0.5\ 1\ 0.5\ 1)\     & (0.5\ 0.5\ 1\ 0.5\ 0.5)\ & (0.5\ 0.5\ 0.5\ 0.5\ 1)\big) \\[1mm]
9:  & & \big((1\ 0.5\ 1\ 0.5\ 1)\     & (1\ 0.5\ 1\ 0.5\ 1)\     & (0.5\ 0.5\ 0.5\ 0.5\ 1)\ &
             (1\ 0.5\ 1\ 0.5\ 1)\     & (0.5\ 0.5\ 1\ 0.5\ 0.5)\ & (1\ 0.5\ 1\ 0.5\ 0.5)\big) \\[1mm]
10: & & \big((1\ 0.5\ 1\ 0.5\ 1)\     & (1\ 0.5\ 1\ 0.5\ 1)\     & (0.5\ 0.5\ 1\ 0.5\ 0.5)\ &
             (1\ 0.5\ 1\ 0.5\ 1)\     & (0.5\ 0.5\ 1\ 0.5\ 0.5)\ & (1\ 0.5\ 0.5\ 0.5\ 0.5)\big) \\[1mm]
11: & & \big((1\ 0.5\ 1\ 0.5\ 1)\     & (1\ 0.5\ 1\ 0.5\ 1)\     & (1\ 0.5\ 0.5\ 0.5\ 0.5)\ &
             (1\ 0.5\ 1\ 0.5\ 1)\     & (0.5\ 0.5\ 1\ 0.5\ 0.5)\ & (1\ 0.5\ 0.5\ 0.5\ 0.5)\big) \\[1mm]
12: & & \big((1\ 0.5\ 1\ 0.5\ 1)\     & (1\ 0.5\ 1\ 0.5\ 1)\     & (1\ 0.5\ 0.5\ 0.5\ 0.5)\ &
             (1\ 0.5\ 1\ 0.5\ 1)\     & (0.5\ 0.5\ 1\ 0.5\ 0.5)\ & (0.5\ 0.5\ 0.5\ 0.5\ 1)\big) \\[1mm]
13: & & \big((1\ 0.5\ 1\ 0.5\ 1)\     & (1\ 0.5\ 1\ 0.5\ 1)\     & (0.5\ 0.5\ 1\ 0.5\ 0.5)\ &
             (1\ 0.5\ 1\ 0.5\ 1)\     & (0.5\ 0.5\ 1\ 0.5\ 0.5)\ & (0.5\ 0.5\ 0.5\ 0.5\ 1)\big) \\[1mm] 
14: & & \big((1\ 0.5\ 1\ 0.5\ 1)\     & (1\ 0.5\ 1\ 0.5\ 1)\     & (0.5\ 0.5\ 0.5\ 0.5\ 1)\ &  
             (1\ 0.5\ 1\ 0.5\ 1)\     & (0.5\ 0.5\ 1\ 0.5\ 0.5)\ & (0.5\ 0.5\ 0.5\ 0.5\ 1)\big) \\[1mm]
15: & & \big((1\ 0.5\ 1\ 0.5\ 1)\     & (1\ 0.5\ 1\ 0.5\ 1)\     & (0.5\ 0.5\ 0.5\ 0.5\ 1)\ &
             (1\ 0.5\ 1\ 0.5\ 1)\     & (0.5\ 0.5\ 1\ 0.5\ 0.5)\ & (1\ 0.5\ 0.5\ 0.5\ 0.5)\big) \\[1mm]
16: & & \big((1\ 0.5\ 1\ 0.5\ 1)\     & (1\ 0.5\ 1\ 0.5\ 1)\     & (0.5\ 0.5\ 1\ 0.5\ 0.5)\ &
             (1\ 0.5\ 1\ 0.5\ 1)\     & (0.5\ 0.5\ 1\ 0.5\ 0.5)\ & (1\ 0.5\ 0.5\ 0.5\ 0.5)\big) \\[1mm]
{\bf *17}: 
    & & {\bf \big((1\ 0.5\ 1\ 0.5\ 1)}\     & {\bf (1\ 0.5\ 1\ 0.5\ 1)}\     & {\bf (1\ 0.5\ 0.5\ 0.5\ 0.5)}\ &
        {\bf      (1\ 0.5\ 1\ 0.5\ 1)}\     & {\bf (0.5\ 0.5\ 1\ 0.5\ 0.5)}\ & {\bf (1\ 0.5\ 0.5\ 0.5\ 0.5)\big)} \\[1mm]
\hline \hline
\end{IEEEeqnarray*}
\end{minipage}     
\vspace{-2mm}
\end{table*}

\section{Conclusions}
\label{S5}

In the paper, we have investigated multi-step fuzzy inference using the Mamdani-Assilian type of fuzzy rules.
To obtain a logical and computational characterisation of it, we have implemented this type of fuzzy inference 
in G\"{o}del logic with truth constants.
We have posed three fundamental problems: reachability, stability, and the existence of a $k$-cycle in multi-step fuzzy inference.
These problems may be formulated as formulae of G\"{o}del logic and reduced to certain deduction problems.
Subsequently, the corresponding theories may be translated to clausal form, and 
the deduction problems may be reduced to unsatisfiability ones.
We believe that under some finitary restrictions, we shall be able to modify the hyperresolution calculus from \cite{Guller2016c}
in order to devise a semi-decision procedure for unsatisfiability testing of the resulting clausal theories.

\bibliographystyle{IEEEtran}
\bibliography{IEEEabrv,tii18}

\end{document}